\documentclass{theoretics}

\title{Highway Dimension: a Metric View}

\ThCSauthor[Sheffield]{Andreas Emil Feldmann}{feldmann.a.e@gmail.com}[0000-0001-6229-5332]
\ThCSauthor[BIU]{Arnold Filtser}{arnold.filtser@biu.ac.il}[0000-0001-9578-9304]
\ThCSaffil[Sheffield]{University of Sheffield}
\ThCSaffil[BIU]{Bar-Ilan University}
\ThCSthanks{For a version with essentially the same content but with alphabetic citation labels and other formatting choices preferred by the authors, see \href{https://arxiv.org/abs/2412.20490v3}{arXiv version~3}.
A preliminary version of this article appeared at SODA 2025~\cite{DBLP:conf/soda/FeldmannF25}. The second author was supported by the Israel Science Foundation (grant No. 1042/22).}
\ThCSshortnames{A.\ E.\ Feldmann, A.\ Filtser}  
\ThCSshorttitle{Highway Dimension: a Metric View}
\ThCSyear{2026}
\ThCSarticlenum{16}
\ThCSreceived{May 22, 2025}
\ThCSrevised{Mar 6, 2026}
\ThCSaccepted{Jul 17, 2026}
\ThCSpublished{Oct 13, 2026}
\ThCSdoicreatedtrue
\ThCSkeywords{Highway dimension, metric spaces, traveling salesperson problem (TSP), polynomial-time approximation schemes (PTAS), padded decompositions, sparse covers, tree covers, doubling dimension.}

\declaretheorem[style=italicized,sibling=theorem,name=Original Definition]{originaldefinition}  %  To improve?
\makeatletter
\newcommand\IfRestateTF{%
    \ifx\label\thmt@gobble@label % or just compared to \@gobble
    \expandafter\@firstoftwo
    \else
    \expandafter\@secondoftwo
    \fi
}
\makeatother
\newcommand{\RestateRemark}{\IfRestateTF{{\normalfont\bfseries (Restated) }}{}}

\usepackage{tcolorbox}
\usepackage{multirow}
\usepackage{thm-restate}
\usepackage{csquotes}

\definecolor{Darkblue}{rgb}{0,0,0.4}
\definecolor{Brown}{cmyk}{0,0.61,1.,0.60}
\definecolor{Purple}{cmyk}{0.45,0.86,0,0}
\definecolor{Darkgreen}{rgb}{0.133,0.543,0.133}

\newcommand{\SPCS}{\textsf{SPCS}\xspace}

\newcommand{\Lip}{{\rm Lip}}

\newcommand{\polylog}{\mathrm{polylog}}

\DeclareMathOperator{\MST}{\mathrm{MST}}
\newcommand{\R}{\mathbb{R}}

\newcommand{\N}{\mathbb{N}}

\newcommand{\opt}{\mathrm{opt}}

\newcommand{\supp}{\mathrm{supp}}

\newcommand{\diam}{\mathrm{diam}}

\newcommand{\TSP}{{TSP}\xspace}

\newcommand{\SPC}{\textsf{SPC}}

\DeclareMathOperator{\cost}{cost}

\def\cC{\ensuremath{\mathcal{C}}}

\def\cF{\ensuremath{\mathcal{F}}}

\def\cP{\ensuremath{\mathcal{P}}}

\def\cS{\ensuremath{\mathcal{S}}}
\def\cT{\ensuremath{\mathcal{T}}}

\newcommand{\ProblemName}[1]{\textsf{#1}}
\newcommand{\ZEX}{{$0$-Extension}\xspace}

\newcommand{\Texp}{\mathsf{Texp}}

\def\eps{\varepsilon}

\DeclareMathAlphabet{\mathpzc}{OT1}{pzc}{m}{it}
\newcommand{\etal}{\emph{ et al.\ }}

\newlength{\dhatheight}

\begin{document}
	\maketitle
	\begin{abstract}
		Realistic metric spaces (such as road/transportation networks) tend to 
be much more algorithmically tractable than general metrics. In an attempt to 
formalize this intuition, Abraham et~al.\ (SODA 2010, JACM 2016) introduced the 
notion of highway dimension. A weighted graph $G$ has highway dimension $h$ if 
for every ball $B$ of radius $\approx4r$ there is a hitting set of size $h$ 
hitting all the shortest paths of length $>r$ in $B$. Unfortunately, this 
definition fails to incorporate some very natural metric spaces such as the grid 
graph, and the Euclidean plane. 

        We relax the definition of highway dimension by demanding to hit only 
approximate shortest paths. In addition to generalizing the original definition, 
this new definition also incorporates all doubling spaces (in particular the 
grid graph and the Euclidean plane). We then construct a PTAS for TSP 
under this new definition (improving a QPTAS w.r.t.\ the original more 
restrictive definition of Feldmann et~al.\ (SICOMP 2018)).  Finally, we develop 
a basic metric toolkit for spaces with small highway dimension by constructing 
padded decompositions, sparse covers/partitions, and tree covers. An abundance 
of applications follow.
	\end{abstract}
    
\section{Introduction}

General metric spaces can be very complicated. Indeed, consider an optimization benchmark problem such as the Travelling Salesperson problem (\TSP). On general metric spaces, it is APX-hard.
On the other hand, on more \enquote{realistic} distance measures, such optimization problems tend to be more tractable. Indeed, in low dimensional Euclidean space, \TSP admits a polynomial time approximation scheme (PTAS) \cite{Arora98,Mitchell99}. More generally, metric spaces with constant doubling dimension\footnote{A metric space $(X,d_X)$ has doubling dimension $d$ if every ball of radius $2r$ can be covered by $2^{d}$ balls of radius~$r$. This is a generalization of Euclidean dimension, where it is known that $(\R^d,\|\cdot\|_2)$ has doubling dimension $\Theta(d)$.\label{foot:doubling}} admit a PTAS for \TSP \cite{BGK16}.
However, \enquote{realistic} metrics are not necessarily doubling. For instance, typical hub-and-spoke networks used in air traffic models are non-planar and have high doubling dimension, since they feature high-degree stars.

The right way to define realistic distance measure turned out to be quite elusive. Bast \etal \cite{BFMSS07,BFM06} observed that road networks are highly structured. Specifically, when one travels long distances, one will necessarily travel through a small set of important traffic junctions (such as central airports and train stations). Bast \etal used this observation to present a heuristic shortest-path algorithm (called transit node routing) and demonstrated experimentally that it improves over previously best algorithms by several orders of magnitude.
Abraham \etal \cite{AFGW10}, motivated by Bast \emph{et al.}'s observation, introduced the notion of highway dimension to formally model transportation networks: the shortest path metric of a weighted graph $G$ has highway dimension $h$ if for every $r>0$ and for every ball  $B_G(v,4r)$ of radius $4r$, there is a small hitting set $H$ of size $|H|\le h$, such that every shortest path $P$ in $B_G(v,4r)$ of length more than $r$ intersects~$H$.

In the journal version of \cite{AFGW10}, Abraham \etal \cite{ADFGW16} introduced an alternative stronger definition of highway dimension.
They used the alternative definition to improve the analysis of a shortest path algorithm based on \emph{transit nodes} that preforms well in practice \cite{BFMSS07,BFM06}. This alternative stronger definition implied constant doubling dimension. This was unfortunate, as we already observed that some \enquote{realistic} metrics (such as hub-and-spoke networks) do not have constant doubling dimension.
Later, Feldmann \etal \cite{FFKP18} slightly generalized the original definition of highway dimension \cite{AFGW10} by requiring that every ball of radius $c\cdot r$ will have a small hitting set for all the shortest paths of length more than $r$ (see \autoref{def:orig-HD}). Here $c\ge 4$, and the difference $\lambda=c-4$ is called the violation.
For any fixed violation~$\lambda>0$, Feldmann \etal \cite{FFKP18} obtain a quasi polynomial approximation scheme (QPTAS) for \TSP.
Later, Feldmann and Saulpic \cite{FS21} obtained a PTAS for $k$-Median, $k$-Means, and Facility Location for graphs with fixed highway dimension; see also a previous FPT-PTAS \cite{BJKW21} for $k$-Median (both under \autoref{def:orig-HD}).
For further discussion and comparisons between the different previous definitions of highway dimension we refer to Blum \cite{Blum19} (see also section~9 in \cite{FFKP18}).

All these previous definitions have some major drawbacks: mainly they do not catch all the metric spaces which are intuitively \enquote{realistic}. One such example is the $\sqrt{n}\times \sqrt{n}$ unweighted grid. This is a planar graph of constant doubling dimension which roughly resembles the road network of some modern cities (e.g. Manhattan), thus it clearly should be considered \enquote{realistic}. However, it has highway dimension $\Omega(\sqrt{n})$. Indeed, even if one slightly perturbs the weights so that every vertex pair will have a unique shortest path, still the columns constitute $\sqrt{n}$ disjoint shortest paths, and thus any hitting set will have size $\Omega(\sqrt{n})$.
A second problem is that the definition of highway dimension is inherently restricted to graphs, and cannot be applied to continuous metric spaces such as the Euclidean plane (which in our view should be considered \enquote{realistic} as well).
We would like a definition which is more expressive: one that includes road/transportation networks and hub-and-spoke networks, but also the grid, the continuous Euclidean plane, and desirably every metric space with constant doubling dimension. At the same time, a good definition should also not be too expressive, so that it is still useful from an algorithmic viewpoint.

\paragraph*{Our contribution (high level).}
We propose a new definition for the highway dimension by relaxing the requirement for exact shortest paths in \autoref{def:orig-HD}. That is for every ball of radius roughly $4r$, there is a set of hubs $|H|$ of size $|H|\le h$, such that for every pair of points $u,z$ at distance at least $r$ in the ball, there is an approximate $u-z$ shortest path going though $x\in H$ (see \autoref{def:HD}).
This new definition has it all:
(1) It generalizes the \autoref{def:orig-HD} \cite{AFGW10,FFKP18}; indeed hitting exact shortest paths is only harder than hitting approximate ones.
(2) It generalizes the doubling dimension (see \autoref{obs:doublingToHighway}); that is, given a metric with doubling dimension $d$, an $\eps\cdot r$-net (see \autoref{def:net}) in $B_X(v,O(r))$ is a set of $\eps^{-O(d)}$ points hitting an approximate $(1+\eps)$-shortest path for every pair.
(3) It naturally applies to continuous spaces; that is, there is nothing inherently graphic about it.

First, we re-construct the toolkit of \cite{FFKP18} developed for graphs of low highway dimension, and adapt it to the new definition (specifically \enquote{global} shortest path covers, towns, and sprawl; see \autoref{sec:BasicProperties}).
Next we construct a PTAS for \TSP for metric spaces with small highway dimension (see \autoref{thm:TSPHighway} in \autoref{sec:TSP}). Note that this answers the main open question left in \cite{FFKP18} (as \cite{FFKP18} had a QPTAS under a less expressive definition).
Finally, we continue our metric view and develop the basic metric toolkit for spaces with small highway dimension (see \autoref{sec:MetricToolkit}). Specifically, we construct padded decompositions, sparse covers, sparse partitions, and tree covers. An abundance of applications follows.

Due to space considerations, related work and some (standard) preliminaries are deferred to the appendix. See Sections~\ref{sec:Related} and \ref{sec:prelims} respectively.

\subsection{Definition of Highway Dimension}

    The definition of the highway dimension by Feldmann \etal \cite{FFKP18} builds on the definition by Abraham \etal \cite{AFGW10}. The latter definition is obtained by setting $c=4$ in the following, however, as was shown in \cite{FFKP18}, it is useful to pick a constant $c$ strictly greater than $4$.

    \begin{originaldefinition}[Highway dimension]\label{def:orig-HD}
       Let $c\geq 4$ be a universal constant.
       A weighted graph $G=(V,E,w)$ has \emph{highway dimension} $h$ if for every $r>0$, and $v\in V$, there exists a subset $H\subseteq V$ with $|H|\leq h$ vertices such that the following property of the ball        $B_v=B_G(v,cr)$ holds: \\
       For every pair of vertices $u,z\in B_v$ such that $d_{G[B_v]}(u,z)=d_{G}(u,z)>r$, every shortest $u-z$ path~$P_{u,z}$ in $G[B_v]$ is intersecting $H$, i.e., $V(P_{u,z})\cap H\ne\emptyset$.
    \end{originaldefinition}
    
    Note that this definition is only concerned with shortest paths of length more than~$r$ that lie within the ball~$B_v$. In particular, if there is a vertex pair $u,z$  for which the distance is at most~$r$ in the graph $G$, but their shortest path $P_{u,z}$ within the ball has length more than $r$, then $P_{u,z}$ does not have to be hit by $H$ (note that this is exactly represented by failing the condition $d_{G[B_v]}(u,z)=d_{G}(u,z)>r$).
    Some examples are due (see \autoref{fig:TBP1}, see also \cite{AFGW10,FFKP18}): the star graph has highway dimension $1$ but doubling dimension (see Footnote~\ref{foot:doubling}) $\Theta(\log n)$.
    However, the spider graph
    which consist of $n$ paths of two edges,
    has highway dimension $n$. On the other hand, the $k\times k$ grid has highway dimension $\Theta(k)$
    but doubling dimension $O(1)$. Thus doubling and highway dimension are incomparable.

%    \begin{center}
%        \includegraphics[width=0.8\textwidth]{fig_edited/StarSpiderGrid}
%    \end{center}

        \begin{figure}
            \includegraphics[width=0.9\linewidth]{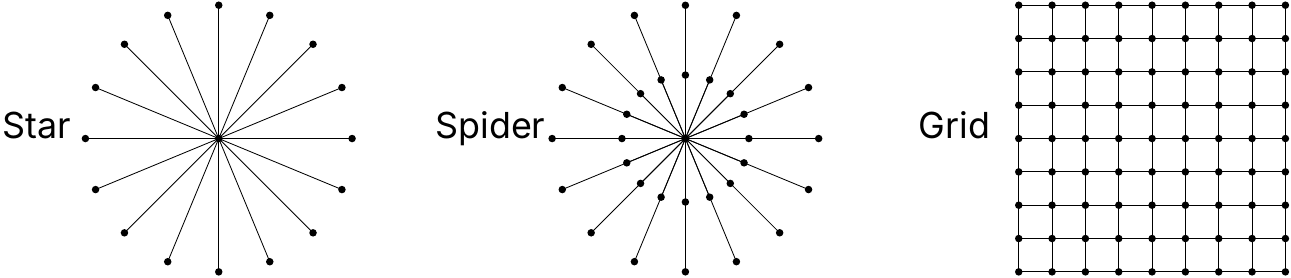}
            \caption{Illustration of the Start, Spider, and Grid graphs.}
            \label{fig:TBP1}
        \end{figure} 

	In this paper we introduce a relaxed definition, where instead of hitting all the required shortest paths, it is enough to hit a single approximate shortest path.
	As we will see, this change also makes the universal constant $c$ from \autoref{def:orig-HD} redundant.
        Note that for the new definition the highway dimension is not a number but a function that depends on how well the hubs approximate the shortest path distance.
        This is because choosing more accurate hubs will give better precision, and in an approximation scheme we will want to control the accuracy.

     \begin{definition}[Highway dimension]\label{def:HD}
    A weighted graph $G=(V,E,w)$ has \emph{highway dimension} $h:\mathbb{R}_{\geq 0}\to\mathbb{N}\cup\{\infty\}$ if for every $\eps\geq 0$, $r>0$, and $v\in V$, there exists a subset $H_\eps\subseteq V$ with $|H_\eps|\leq h(\eps)$ vertices, such that the following property of the ball $B_v=B_G(v,(4+8\eps)r)$ holds:\\
    For every pair of vertices $u,z\in B_v$ such that $d_{G}(u,z)>r$ and $d_{G[B_v]}(u,z)\leq (1+\eps)d_{G}(u,z)$, there is a $u-z$ path $P_{u,z}$ in $G[B_v]$ of length $w(P_{u,z})\le (1+\eps)\cdot d_{G}(u,z)$ intersecting $H_\eps$, i.e., $V(P_{u,z})\cap H_\eps\ne\emptyset$.
    \end{definition}

    Some comments are in order.

    \paragraph{Computational aspects.}\label{rem:HittingSetFind}
        In general, given a center $v\in V$, $r>0$, and $\eps\ge0$, we do not know of any algorithm that finds the minimum set $H_\eps$ satisfying the requirement in \autoref{def:HD} in polynomial time. In fact, we do not know of any approximation algorithm beyond $O(\log n)$ (that works for general instances of hitting set / set cover). This is in contrast to \autoref{def:orig-HD} where an $O(\log h)$ approximation is known \cite{ADFGW11}, assuming that all shortest paths are unique.
        Throughout this paper, when we say that a graph has highway dimension $h:\mathbb{R}_{\geq 0}\to\mathbb{N}\cup\{\infty\}$ we assume that it comes with a polynomial time algorithm that computes sets $H_\eps$ of size $h(\eps)$. Note that $H_\eps$ is not assumed to be of minimum possible size. In particular, if there always exists a hitting set of size $h(\eps)$, but efficiently we can only find a hitting set of size $h'(\eps)$, we will say that the highway dimension is $h'(\eps)$.

    \paragraph{Relation to old definition.}
        Note that any graph $G$ of highway dimension $h$ according to \autoref{def:orig-HD} for $c\geq 4$ has highway dimension $h(\eps)=h$ according to \autoref{def:HD} for any $\eps\in[0,\frac{c-4}{8}]$, since for these values of~$\eps$ the balls in \autoref{def:HD} have radius at most $cr$, which matches the radius of the balls in \autoref{def:orig-HD}.\footnote{This is correct up to a small subtlety: \autoref{def:orig-HD} guarantees to hit all the $u-z$ shortest paths lying inside $B_v$ (potentially an exponential number w.r.t. the number of vertices $|V|$), while \autoref{def:HD} guarantees only to hit a single $u-z$ shortest path. As it is usually possible to assume unique shortest paths (by slightly perturbing weights by infinitesimal amounts, or just breaking ties consistently), we will ignore this subtlety.}
        For larger values of $\eps$ the value of $h(\eps)$ might be unbounded for such a graph $G$. For our later results such as the PTAS for \TSP, we only need the value of $h(\eps)$ to be bounded for a fixed value $\eps>0$.
        In particular, to get a PTAS it is enough that $h(\eps)$ is bounded for all $\eps\in(0,\gamma]$, for an arbitrary constant $\gamma$ (even if $h(0)=\infty$).
        This also implies a PTAS for \TSP for any given graph~$G$ with highway dimension $h$ w.r.t.\ \autoref{def:orig-HD} where $c$ is strictly greater than $4$, since then the graph has highway dimension $h(\eps)= h$ for $\eps\in[0,\frac{c-4}{8}]$ w.r.t. \autoref{def:HD}.
        This reflects the observation of Feldmann \etal \cite{FFKP18} that choosing $c$ strictly greater than~$4$ gives additional properties that can be exploited algorithmically compared to Abraham \emph{et al.}'s \cite{AFGW10} definition where~$c=4$.

     \paragraph{Highway dimension of a metric space.}
     A major advantage of \autoref{def:HD} is that is can be applied to metric spaces (even on infinite continuous spaces), and generalizes doubling dimension (see \autoref{obs:doublingToHighway}).
    		Given a metric space $(X,d_X)$, we can think of it as a complete graph  $G=(X,E={\binom{X}{2}},w)$ where the weight of an edge $xy$ is equal to the metric distance between its endpoints, i.e., $w(xy)=d_X(x,y)$. We will use \autoref{def:HD} w.r.t.\ metric spaces. Note that here, the induced graph $G[B_v]$ is simply the metric space $X$ restricted to the points in the ball $B_v=B_X(v,(4+8\eps)r)$. In particular, for every pair of points $u,z\in B_v$, $d_{G[B_v]}(u,z)=w(uz)=d_X(u,z)$. Hence, due to the triangle inequality, the property  from \autoref{def:HD} in metric spaces is equivalent to the following:
    		For every pair of points $u,z\in B_v$ such that $d_{X}(u,z)>r$, there is a point $x\in H_\eps$ such that $d_X(u,x)+d_X(x,z)\leq(1+\eps)d_X(u,z)$.
            In this sense this definition naturally extends to metric spaces and can even be used on infinite metric spaces (such as the Euclidean space $(\R^d,\|\cdot\|_2)$~).

    	Note that it is important that our \autoref{def:HD} is w.r.t.\ graphs and not metric spaces in order to generalize \autoref{def:orig-HD}. This is as balls in graphs are not \enquote{convex}. That is, it is possible that $u,z\in B_v$, while their shortest path does not lie inside $B_v$. In particular, it is possible that a graph $G=(V,E,w)$ will have small highway dimension, while the respective shortest path metric will not (see \autoref{fig:DuoStar}).%

     \begin{cfigure}%[t]
        \centering{\includegraphics[width=.9\textwidth]{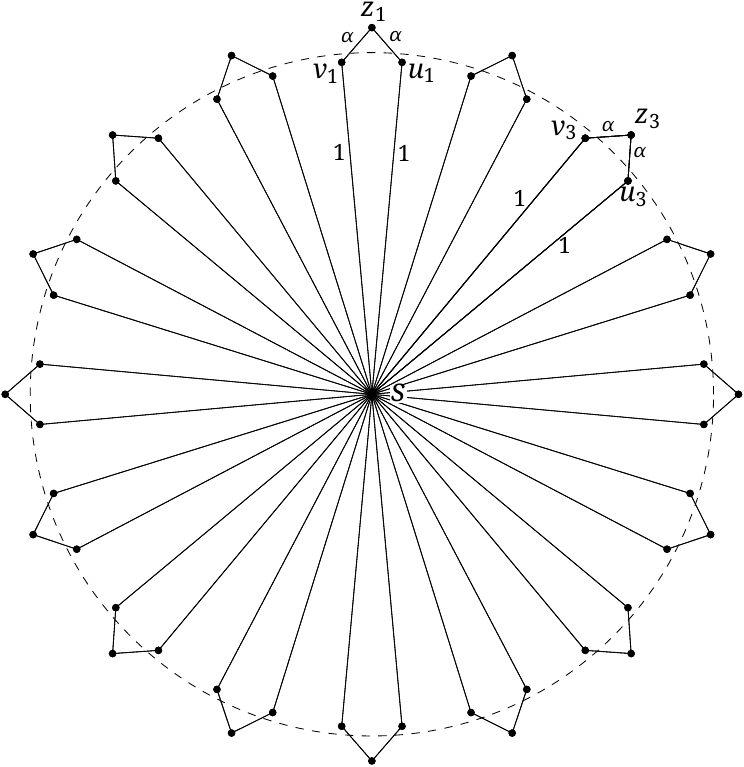}}
        \caption{%
        An example of a graph $G=(V,E,w)$ with small highway dimension, while it's respective shortest path metric has large highway dimension.
        Fix $\eps>0$. The graph $G$ in the figure consist of a center vertex $s$, which is connected to $n$ pairs of vertices $v_1,u_1,v_2,u_2,\dots,v_n,u_n$ by edges of weight $1$. In addition, each such pair $v_i,u_i$ is connected to a vertex $z_i$ with edges of weight $\alpha=\frac{1}{7+16\eps}$. The graph $G$, for $\eps$ has highway dimension $h(\eps)=1$. Indeed, fix $r=\frac{1}{4+8\eps}$, and consider the ball $B_v=B_G(s,(4+8\eps)\cdot r)=B_G(s,1)$. The induced graph $G[B_v]$ is simply a star graph with $2n$ leaves, all at pairwise distances $2$. In particular $H_\eps=\{s\}$ hits all the (exact) shortest paths.
         Next fix $r'=\frac{1+\alpha}{4+8\eps}=\frac{1+\frac{1}{7+16\eps}}{4+8\eps}=\frac{2}{7+16\eps}=2\alpha$, and consider the ball $B'_v=B_G(s,(4+8\eps)\cdot r')=B_G(s,1+\alpha)$. Note that the distance from $v_i$ to $u_i$ is $2\alpha=r'$. It follows that $H_\eps=\{s\}$ is still a valid hitting set, as it hits all the shortest paths of length strictly greater than $r$. The other balls have small hitting sets as well.\\
        Next consider the shortest path metric $(X=V,d_G)$ of $G$, and the ball $B_v=B_G(s,(4+8\eps)\cdot r)=B_X(s,1)$. For every pair $v_i,u_i$, the distance is $d_X(v_i,u_i)=2\alpha=\frac{2}{7+16\eps}>\frac{1}{4+8\eps}=r$, and thus $H_\eps$ must contain a point from $\{u_i,v_i\}$. It follows that the highway dimension is at least $h(\eps)\ge n$. It is also interesting to note that the doubling dimension of the shortest path metric is $\Theta(\log n)$.
        }\label{fig:DuoStar}
    \end{cfigure}

    While the \autoref{def:orig-HD} of highway dimension is incomparable to doubling dimension$^{\ref{foot:doubling}}$,
    our \autoref{def:HD} also encompasses doubling metrics. This follows using the packing property (\autoref{lem:dd-vs-nets}). See \autoref{sec:prelims} for formal definitions and proof.

\begin{restatable}[]{observation}{HwayImpliesDoubling}
\label{obs:doublingToHighway}\RestateRemark
    	Consider a metric $(X,d_X)$ of doubling dimension $d$.
    	Then $X$ has highway dimension at most $h(\eps)=\begin{cases}
    		(64+\frac{32}{\eps})^{d} & \eps>0\\
    		\infty & \eps=0
    	\end{cases}$~.
\end{restatable}

\subsection{Our Results}
Our main theorem is a PTAS for \TSP in metrics/graphs with bounded highway dimension.
In fact we prove our result for the Subset \TSP problem, which is a generalization of \TSP (see \autoref{def:subsetTSP}).
We remark that, in contrast to doubling metrics, the highway dimension is not bounded when taking subgraphs, and thus solving Subset \TSP is not trivial given an algorithm for \TSP in our case.

Note that due to \autoref{obs:doublingToHighway}, the following theorem generalizes results by Bartal \etal\cite{BGK16} and Banerjee \etal\cite{banerjee2024novel}, who gave a PTASes for \TSP in low doubling metrics.
Furthermore, our result also improves and generalizes over the QPTAS by Feldmann \etal \cite{FFKP18} for \TSP in graphs of low highway dimension according to \autoref{def:orig-HD}.

 \begin{restatable}[]{theorem}{TSPthm}
\label{thm:TSPHighway}\RestateRemark
Consider a weighted graph $G=(V,E,w)$ with $n$ vertices and highway dimension $h:\mathbb{R}_{\geq 0}\to\mathbb{N}\cup\{\infty\}$, then for a given terminal set $K\subseteq V$, and $\eps\in(0,\frac16]$, there is an algorithm that computes a $(1+\eps)$-approximation to the Subset \TSP problem in time $2^{2^{O\left(\frac{1}{\eps}\log^2\frac{h(\eps)}{\eps}\right)}}n^{O(1)}$.
\end{restatable}

We then turn to constructing a metric toolkit for graphs/metrics with bounded highway dimension. Our (and previous) results are summarized in \autoref{fig:tableMetricToolkit}. We refer to \autoref{sec:MetricToolkit} for definitions and background.
First we construct strong padded decomposition with padding parameter $O(\log h(\eps))$ (\autoref{thm:padded}).
Roughly, this means that for any $\Delta>0$, we can sample a partition of the graph vertices into clusters of strong diameter $\le \Delta$ such that every ball of radius $\Omega(\frac{\Delta}{\log h(\eps)})$ is fully contained in a single cluster with probability $\frac12$.
We then use this padded decomposition to construct an embedding into $\ell_2$ with distortion $O(\sqrt{\log h(\eps)\cdot\log n})$ (\autoref{cor:EmbeddingLp}), to construct $O(\log h(\eps))$-approximation algorithms for the \emph{\ZEX} problem (\autoref{cor:0Extension}), and the Lipschitz extension problem (\autoref{cor:LipschitzExtension}), and finally to construct a flow sparsifier with quality $O(\log h(\eps))$.

Next, we construct an $(8,O(h(\eps)^2))$-sparse cover (\autoref{thm:SparseCover}), which means that we can construct a collection of (overlapping) clusters of strong diameter $\Delta$ such that every vertex belongs to $O(h(\eps)^2)$ clusters, and every ball of radius $\frac\Delta8$ is fully contained in some cluster.
We then use an alternative version of this sparse cover (called spare partition cover; \autoref{thm:SparsePartitionCover}) to construct an $O(\eps^{-2}\cdot h(\eps)^4\cdot\log n)$-approximation algorithm for the oblivious buy-at-bulk problem (\autoref{cor:BuyAtBulk}), and for metric embedding with distortion $O(\frac1\eps)$ into  $O(h(\eps)^2\cdot \log\frac n\eps \cdot\log\frac1\eps)$ dimensional~$\ell_\infty$ (\autoref{cor:EmbeddingLinfty}).

Finally, we construct a tree cover, that is a collection of $h(\frac\eps2)\cdot\tilde{O}(\eps^{-1})$ dominating trees such that every pair has distortion only $1+\eps$ in one of these trees (\autoref{thm:HighwayTreeCover}). We use this tree cover to construct a succinct and efficient distance oracle (\autoref{cor:DistanceOracle}), and a reliable spanner (\autoref{cor:ReliableSpanner}), which can withstand massive failures of vertices.

\begin{table}
\begin{multicols}{2}

\scalebox{0.9}{\begin{tabular}{|l|l|l|}
\multicolumn{3}{c}{{\Large Padded Decomposition}}                     \\\hline
Space              & Padding Parm. & Reference           \\\hline
General            & $O(\log n)$       & \cite{Bar96}               \\\hline
$K_r$-minor   free & $O(\log r)$            & \cite{CF25} \\\hline
Doubling $d$       & $O(d)$            & \cite{GKL03,Fil19padded}   \\\hline
Highway            & $O(\ln h(\eps))$  & \autoref{thm:padded} \\\hline
\end{tabular}}

\vspace{5pt}

\scalebox{0.9}{\begin{tabular}{|l|l|l|l|}
\multicolumn{4}{c}{{\Large Sparse Cover}}                                         \\\hline
Space            & Padding & Sparsity              & Reference           \\\hline
General          & $4k-2$  & $2k\cdot n^{\frac1k}$ & \cite{AP90}                \\\hline
$K_r$-minor free & $8+\eps$  & $O(r^4/\eps^2)$              & \cite{CF25} \\\hline
Doubling $d$     & $O(d)$  & $\tilde{O}(d)$        & \cite{Fil19padded}   \\\hline
Highway          & $8$     & $O(h(\eps)^2)$        & \autoref{thm:SparseCover}   \\\hline
\end{tabular}}

\phantom{a}\\\phantom{a}\\

\scalebox{0.9}{\begin{tabular}{|l|l|l|l|}
\multicolumn{4}{c}{{\Large Tree Cover}}                                                                      \\\hline
Space            & Stretch  & \# Trees                                       & Reference            \\\hline
General          & $O(k)$   & $O(n^{\frac1k}\cdot k)$                        & \cite{MN07}                 \\\hline
Planar           & $1+\eps$ & $\eps^{-3}\cdot\log(\frac1\eps)$               & \cite{CCLMST23Planar}       \\\hline
$K_r$-minor free & $1+\eps$ & $2^{r^{O(r)}/\eps}$                            & \cite{CCLMST24}             \\\hline
Euclidean $\R^d$ & $1+\eps$ & $\tilde{O}(\eps^{-\frac{d+1}{2}})$ & \cite{CCLMST24Euclidean}    \\\hline
Doubling $d$     & $1+\eps$ & $\eps^{-O(d)}$                                 & \cite{BFN22}                \\\hline
Highway          & $1+\eps$ & $h(\frac\eps2)\cdot\tilde{O}(\eps^{-1})$                                 & \autoref{thm:HighwayTreeCover}\\\hline
\end{tabular}}
\end{multicols}
    \caption{Summary of new and previous results on the metric toolkit.}
    \label{fig:tableMetricToolkit}
\end{table}

\section{Related work}\label{sec:Related}
Context and related work for the metric toolkit is given in \autoref{sec:MetricToolkit} where they are introduced. Here we provide additional background on highway dimension and TSP.

The metric TSP problem in general graphs is APX-hard \cite{PY93}. Christofides \cite{Chr76} provided a celebrated $\frac32$ approximation, which recently was slightly improved \cite{KKG21}.
The main bulk of the research focused on constructing approximation schemes for more structured metric spaces. Specifically, there is an efficient PTAS (EPTAS) for planar graphs \cite{GKP95,AGKKW98,Klein05,Klein08}, EPTAS for low dimensional Euclidean space \cite{Arora98,Mitchell99,BartalG13}, EPTAS for metrics with low doubling dimension \cite{BGK16,Gottlieb15}, and EPTAS for fixed minor free graphs \cite{Grigni00,BLW17}.
The running time for metrics of doubling dimension $d$ in \cite{BGK16} is roughly $n^{(2^{d}/\eps)^{O(d)}}$.
However, in the original arXiv version~1 (from~2011) \cite{BGK16arxivVersion1}, in Theorem~4.1 it is stated that if there was a light spanner for doubling metrics, then a better running time could be achieved.
Indeed, since \cite{BGK16arxivVersion1} light spanners for doubling metrics where discovered \cite{Gottlieb15,FS20,BLW19}.
By plugging in the spanner from \cite{BLW19} with lightness $\eps^{-O(d)}$ into Theorem~4.1 in \cite{BGK16arxivVersion1}, a runtime of $n\cdot(\log n)^{\eps^{-O(d)}\cdot2^{O(d^{2})}\cdot\log\log n}=n\cdot(2^{\eps^{-O(d)}\cdot2^{O(d^{2})}}+n^{o(1)})$ can be achieved. However a recent result of Banerjee \etal \cite{banerjee2024novel} gives a faster approximation scheme with runtime $2^{(d/\eps)^{O(d)}}n+(1/\eps)^{O(d)}n\log n$, as summarized in \autoref{thm:TSP-dd}.

The subset TSP version, which also studied in this paper (see \autoref{def:subsetTSP}) is especially interesting in planar graph and fixed minor free graphs (as a metric over a subset of the vertices is not planar/minor free). Here there is a PTAS for subset TSP in planar graph \cite{Le20} and fixed minor free graphs \cite{CFKL20}.

The highway dimension was defined by Abraham \etal \cite{AFGW10} and several other definitions followed (see \cite{Blum19} for an overview and comparisons).
Abraham \etal showed that certain shortest-path heuristics are provably more efficient than Dijkstra's algorithm if the input has bounded highway dimension.
In \cite{ADFGW11}, they showed that set systems given by unique shortest paths have VC-dimension~2.
This means that slightly perturbing the edge weights and thus obtaining unique shortest paths, it is possible to approximate the highway dimension within an $O(\log h)$ factor.
Blum \etal \cite{DBLP:conf/iwpec/0001DF0Z22} showed that it is unlikely that there is an FPT algorithm to compute the highway dimension (although the question was not fully answered).

Feldmann \etal \cite{FFKP18} studied NP-hard optimisation problems such as \TSP, Steiner Tree, $k$-Median, and Facility Location in graphs of bounded highway dimension.
Specifically, they showed that every weighted graph with highway dimension $h$ (under \autoref{def:orig-HD}) and aspect ratio$^{\ref{foot:aspect}}$ $\Phi$ stochastically embeds into a distribution over graphs with treewidth $(\log\Phi)^{O(\log^{2}(\frac{h}{\eps\lambda})/\lambda)}$, and expected distortion $1+\eps$. A similar theorem for doubling metrics was previously known \cite{Talwar04}.
This implies QPTASs for the above mentioned problems.
For clustering problems such as $k$-Median, $k$-Means, and Facility Location, Feldmann and Saulpic \cite{FS21} gave a PTAS for graphs of bounded highway dimension.
For the related $k$-Center and Vehicle Routing problems, Becker \etal \cite{BKS18} gave a PTAS.
For the $k$-Center problem, Feldmann and Marx \cite{DBLP:journals/algorithmica/FeldmannM20} showed that the problem is W[1]-hard even on planar graphs of constant doubling dimension, if the highway dimension and $k$ are the parameters.

\section{Preliminaries}\label{sec:prelims}
The $\tilde{O}$ notation hides poly-logarithmic factors, that is $\tilde{O}(g)=O(g)\cdot\polylog(g)$.

    \paragraph*{Doubling dimension.}
	A metric space $(X,d_X)$ has \emph{doubling dimension} $d$, if every ball $B_X(v,r)$ of radius $r$ can be covered by at most $2^d$ balls of radius $\frac r2$ \cite{GKL03}. That is, there is a set $A$ of at most $2^d$ points such that $B_X(v,r)\subseteq\bigcup_{u\in A}B_X(u,\frac r2)$.

    \begin{definition}\label{def:net}
        A \emph{$\delta$-net} of a metric $(X,d_X)$ is a subset $N\subseteq X$ such that
        \begin{enumerate}
            \item $N$ is a \emph{$\delta$-packing}, i.e., $d_X(p,q)>\delta$ for all distinct $p,q\in N$, and
            \item $N$ is a \emph{$\delta$-covering}, i.e., for every $u\in X$ there is a $p\in N$ with $d_X(u,p)\leq\delta$.
        \end{enumerate}
    \end{definition}
        A $\delta$-net can be computed greedily in polynomial time, e.g., as a prefix of a Gonzales order of the metric. Furthermore, a net hierarchy $N_0\supseteq N_1 \supseteq \dots \supseteq N_i\supseteq\dots$ where each $N_i$ is a $2^i$-net of the metric can also be computed in this way, i.e., each $N_i$ is a prefix of a Gonzales order of $N_{i-1}$.
        The most crucial property of doubling metrics, that formalized as the packing property bellow, is the fact that they contain sparse nets.
        The packing property can be proven by repeatedly applying the definition of doubling dimension \cite{GKL03}.
        \begin{lemma}[Packing Property]\label{lem:dd-vs-nets}
        Let $X$ be a metric, and $\alpha(N)$ denote the aspect ratio\footnote{The aspect ratio of a metric space is $\frac{\max_{x,y\in X}d_X(x,y)}{\min_{x,y\in X}d_X(x,y)}$  the ratio between the maximal and minimal distances.\label{foot:aspect}} of $N\subseteq X$.
        If $X$ has doubling dimension $d$, then every $\delta$-packing $N$ of $X$ has size $|N|\leq 2^{d\lceil\log_2\alpha(N)\rceil}$.
        \end{lemma}

    We now prove that doubling dimension implies bounded highway dimension (\autoref{obs:doublingToHighway} restated for convenience).

\HwayImpliesDoubling*

    \begin{proof}
		Fix $v\in V$, $r>0$, and $\eps\geq 0$. If $\eps=0$ there is nothing to prove, so we will assume $\eps>0$.
		Fix $B_v=B_X(v,(4+8\eps)r)$, and let $H_\eps$ be a $\frac{\eps}{2}r$-net of $B_v$. That is, a set of points at pairwise distance strictly greater than $\frac{\eps}{2}r$, and such that every point in $B_v$ has a net point at distance at most $\frac{\eps}{2}r$.
		Given vertices $u,z\in B_v$ with $d_X(u,z)>r$, there is a net point $x\in H_\eps$ with $d_X(u,x)\leq \frac{\eps}{2}r< \frac{\eps}{2} d_X(u,z)$, which by the triangle inequality also means that $d_{X}(u,x)+d_{X}(x,z)\leq d_{X}(u,z)+2\cdot d_{X}(u,x)\leq(1+\eps)\cdot d_{X}(u,z)$. Thus $H_\eps$ is a valid hitting set for \autoref{def:HD}.
		Note that $\alpha(H_\eps)<\frac{\diam(B_v)}{\frac{\eps}{2}\cdot r}\le 32+\frac{16}{\eps}$,
		and thus by \autoref{lem:dd-vs-nets} we have $|H_{\eps}|\leq2^{d\lceil\log_{2}\alpha(H_{\eps})\rceil}\leq2^{d\cdot\log_{2}\alpha(H_{\eps})+d}\le\left(32+\frac{16}{\eps}\right)^{d}\cdot2^{d}=\left(64+\frac{32}{\eps}\right)^{d}$.
	\end{proof}
    
	\paragraph*{Graphs.}
We consider connected undirected graphs $G=(V,E)$ with edge weights
	$w: E \to \R_{\ge 0}$. We say that vertices $v,u$ are neighbors if $\{v,u\}\in E$. Let $d_{G}$ denote the shortest path metric in
	$G$.
	$B_G(v,r)=\{u\in V\mid d_G(v,u)\le r\}$ is the ball of radius $r$ around $v$. For a vertex $v\in
	V$ and a subset $A\subseteq V$, let $d_{G}(x,A):=\min_{a\in A}d_G(x,a)$,
	where $d_{G}(x,\emptyset)= \infty$. For a subset of vertices
	$A\subseteq V$, let $G[A]$ denote the induced graph on $A$,
	and let $G\setminus A := G[V\setminus A]$.
	
	The \emph{diameter} of a graph $G$ is $\diam(G)=\max_{v,u\in V}d_G(v,u)$, i.e. the maximal distance between a pair of vertices.
	Given a subset $A\subseteq V$, the \emph{weak}-diameter of $A$ is $\diam_G(A)=\max_{v,u\in A}d_G(v,u)$, i.e. the maximal distance between a pair of vertices in $A$, w.r.t. to $d_G$. The \emph{strong}-diameter of $A$ is $\diam(G[A])$, the diameter of the graph induced by $A$. For illustration, in the figure to the right, consider the lower cluster encircled by a dashed red line. The weak-diameter of the cluster is $4$ (as $d_G(v,z)=4$) while the strong diameter is $6$ (as $d_{G[A]}(v,u)=6$).

    \section{Basic Properties}\label{sec:BasicProperties}
    In this section we adapt the toolkit for highway dimension created in previous works to our new definition.
    As already shown by Abraham \etal \cite{AFGW10}, low highway dimension implies sparse shortest path covers, as defined next. Note that here we switch to (shortest-path) metrics instead of graphs (this is as the topology of the graph has no role here).

	\begin{definition}[shortest path cover]\label{def:SPC}
        For a metric $(X,d_X)$, $r>0$, and $\eps\geq 0$,
        an \emph{$(r,\eps)$-shortest path cover} \mbox{$\SPC\subseteq X$} is a set of \emph{hubs}
        such that for any pair $u,z\in X$ with $d_X(u,z)\in (r,(2+\eps)r]$, there is a hub $x\in\SPC$
        for which $d_X(u,x)+d_X(x,z)\leq(1+\eps)d_X(u,z)$.
        A shortest path cover $\SPC$ is called \emph{locally $s$-sparse} for scale~$r$,
        if no ball of radius $(2+4\eps)r$ contains more than~$s$ vertices from~$\SPC$.
	\end{definition}

        A proof by Abraham \etal \cite{AFGW10} generalized to \autoref{def:HD} shows that any graph of highway dimension $h:\mathbb{R}_{\geq 0}\to\mathbb{N}\cup\{\infty\}$ has locally $h(\eps)$-sparse $(r,\eps)$-shortest path covers for any $r>0$ and $\eps\in[0,1]$.
        Indeed, if we can compute hitting sets for approximate shortest paths in polynomial time, as discussed in \autoref{rem:HittingSetFind}, then we may even compute shortest path covers with this sparseness guarantee in polynomial time.
        This follows as the proof is constructive.

    \begin{restatable}[]{lemma}{SPCLemma}
\label{lem:sparse-SPC}\RestateRemark
Given a graph $G=(V,E)$, and parameters $r>0$, $\eps\in[0,1]$, there is an $(r,\eps)$-shortest path cover $\SPC$, such that if $G$ has highway dimension $h:\mathbb{R}_{\geq 0}\to\mathbb{N}\cup\{\infty\}$, then $\SPC$ is locally $h(\eps)$-sparse.
\end{restatable}

        \begin{proof}
        	We present a local search procedure to compute the shortest path cover, which starts off with an arbitrary $(r,\eps)$-shortest path cover $\SPC$, e.g., the full vertex set $V$, and then repeatedly
        	finds the densest ball of radius $(2+4\eps)r$ and tries to sparsify it. The algorithm halts once the sparsification attempt fails.

        	In more detail, initially $\SPC\leftarrow V$. In every iteration the algorithm finds the vertex $v\in V$ maximizing $|B_v\cap\SPC|$ where $B_v=B_G(v,(2+4\eps)r)$. The algorithm then constructs a set $H$ such that $\SPC'=(\SPC\setminus B_v)\cup H$ is still an $(r,\eps)$-shortest path cover. If $|\SPC'|\ge|\SPC|$, the algorithm halts and returns $\SPC$. Otherwise, it continues for another iteration after setting $\SPC\leftarrow \SPC'$.
        	Clearly, the returned set is an $(r,\eps)$-shortest path cover. Furthermore, the algorithm halts after at most $n$ iterations, since the size of the shortest path cover is decreased in each iteration.

        	The algorithm obtains the set $H$ as a hitting set to a set system $\mathcal{H}\subseteq 2^V$ over the vertex set~$V$.
        	Specifically, a vertex pair $u,z\in V$ for which (1) $d_G(u,z)\in(r,(2+\eps)r]$, and (2) there exists a hub $x\in B_v\cap\SPC$ for which $d_G(u,x)+d_G(x,z)\leq(1+\eps)d_G(u,z)$, is called a $B_v$-affected pair. For each $B_v$-affected pair $u,z$, the instance~$\mathcal{H}$ contains a set $S_{u,z}=\{y\in V\mid d_G(u,y)+d_G(y,z)\leq(1+\eps)d_G(u,z)\}$, i.e., it consists of all the potential hubs for $u,z$. The set $H$ is then obtained as a minimum sized hitting set to $\mathcal{H}$.
        	Clearly, $\SPC'=(\SPC\setminus B_v)\cup H$ is still an $(r,\eps)$-shortest path cover.

        	Suppose for the sake of contradiction that the returned set $\SPC$ is not locally $h(\eps)$-sparse. That is, there is a vertex $v'\in V$ such that $|B_{v'}\cap\SPC|> h(\eps)$.
        	In the last iteration of the algorithm, it chose a vertex $v\in V$ that maximizes $|B_{v}\cap\SPC|\ge |B_{v'}\cap\SPC|>h(\eps)$.
        	Consider a $B_v$-affected pair~$u,z$.
            That is, there is a hub $x\in B_{v}\cap\SPC$ such that $d_G(u,x)+d_G(x,z)\leq(1+\eps)d_G(u,z)$.
         Let $Q_{u,z}$ be any $u-z$ path of length at most $(1+\eps)d_G(u,z)$. Denote by $b$ the vertex of~$Q_{u,z}$ that is farthest from the center $v$ of the ball $B_v$. We can bound the distance between $b$ and the hub $x\in B_v$ as follows:
     	    \begin{align*}
        		d_G(x,b) &\leq \min\{d_G(x,u)+d_X(u,b),\; d_G(x,z)+d_G(z,b)\} \\
        		&\leq \textstyle\frac{1}{2}(d_G(u,x)+d_G(x,z)+d_G(u,b)+d_G(b,z))\quad\leq\quad (1+\eps)d_G(u,z)~.
        	\end{align*}
        	Hence the distance from $v$ to $b$ is $d_G(v,b)\leq d_G(v,x)+d_G(x,b)\leq (2+4\eps)r +(1+\eps)(2+\eps)r\leq (4+8\eps)r$, given that $x$ lies in the ball of radius $(2+4\eps)r$ centered at $v$, $d_G(u,z)\leq (2+\eps)r$, and $\eps^2\leq\eps$ as~$\eps\leq 1$. As~$b$ is the vertex of $Q_{u,z}$ farthest from $v$, the entire path $Q_{u,z}$ lies within the ball $B'_v=B_G(v,(4+8\eps)r)$. In particular, $u,z\in B'_v$, $d_G(u,z)>r$, and $d_{G[B'_v]}(u,z)\leq(1+\eps)d_G(u,z)$.
        	As we assumed that $G$ has highway dimension $h:\mathbb{R}_{\geq 0}\to\mathbb{N}\cup\{\infty\}$,
        	\autoref{def:HD} implies that there is a set $H_\eps$ of size at most $h(\eps)$ intersecting some $u-z$~path~$P_{u,z}$ in~$G[B'_v]$ of length at most $(1+\eps)d_G(u,z)$ for every $B_v$-affected pair $u,z$. Thus this set $H_\eps$ is a hitting set for the instance $\mathcal{H}$. Since our algorithm computed a solution $H$ of minimum size, we have $|H|\leq |H_\eps|\leq h(\eps)$. But this implies that
        	$|\SPC'|=\left|(\SPC\setminus B_v)\cup H\right|\leq |\SPC|-|B_{v'}\cap\SPC|+|H|<|\SPC|$, a contradiction to the fact that the algorithm halted.
        \end{proof}

        We will use \emph{towns} and \emph{sprawl}, for which our definition is adapted from  \cite{FFKP18}.

\begin{definition}[towns and sprawl]\label{def:towns}
	Given an $(r,\eps)$-shortest path cover $\SPC$ for scale $r>0$ and $\eps\geq 0$,
	for any node~$v\in X$ such that $d_X(v,\SPC) > (2+\eps)\cdot r$, we call the set
	$T=\{u \in X \mid d_X(u,v) \leq r\}$ a \emph{town} for scale $r$.
	The \emph{sprawl} for scale $r$ (denoted $\cS$) is the set of all vertices that are not in towns.
\end{definition}

        Similarly to \cite{FFKP18}, the towns form clusters. We reprove this fact using our new definition of the highway dimension.

\begin{lemma}%
	\label{lem:townproperties}
	Let $T$ be a town for an $(r,\eps)$-shortest path cover $\SPC$ for some $r>0$ and~$\eps\geq 0$.
	The following properties hold:
	\begin{enumerate}
		\item Small diameter: For every town $T$, $\diam(T) \le r$.
		\item Separation: For every town $T$, $d_X(T,V\setminus T) > r$.
		\item Sprawl: For any node $v\in\cS$ of the sprawl of scale $r$, $d_X(v,\SPC)\leq  (2+\eps)\cdot r$.
	\end{enumerate}
\end{lemma}
\begin{proof}
	The third property (sprawl) holds immediately from the definition of the towns.
	Next we prove the first property regarding the bound on the diameter of a town. Consider a town $T=\{u \in X \mid d_X(u,v) \leq r\}$, and suppose for the sake of contradiction that there are vertices $u,z\in T$ such that
	$d_X(u,z)>r$.
	By the triangle inequality, $d_X(u,z)\le d_X(u,v)+d_X(v,z)\le 2r$. Thus $d_X(u,z)\in(r,(2+\eps)r]$, and hence there is a hub $x\in\SPC$ that lies approximately between $u$ and~$z$,
	i.e., $d_X(u,x)+d_X(x,z)\leq(1+\eps)d_X(u,z)$.
	Assume w.l.o.g.\ that $x$ is closer to $z$ than to~$u$, so that
	$d_X(x,z)\leq\frac{1+\eps}{2}\cdot d_X(u,z)\leq (1+\eps)r$.
	But then, $d_X(x,v)\leq d_X(x,z)+d_X(z,v)\leq (2+\eps)\cdot r$,
	which contradicts $d_X(v,\SPC)> (2+\eps)\cdot r$.

 Finally, to prove the separation property, consider a town $T=\{u \in X \mid d_X(u,v) \leq r\}$, and suppose for the sake of contradiction that there are vertices $u\in T$, $z\notin T$ such that $d_X(u,z)\le r$.
	By the triangle inequality, $d_X(v,z)\le d_X(v,u)+d_X(u,z)\le 2r$.
	From the other hand, as $z\notin T$, $d_X(v,z)>r$, and thus $d_X(v,z)\in(r,(2+\eps)r]$. In particular there is a hub $x\in\SPC$ such that $d_X(v,x)+d_X(x,z)\leq (1+\eps)d_X(v,z)$.
    It follows that
    \[
d_{X}(v,x)\leq(1+\eps)d_{X}(v,z)-d_{X}(x,z)\overset{(*)}{\le}(1+\eps)d_{X}(v,z)-\left(d_{X}(v,x)-d_{X}(v,z)\right)~,
    \]
    where the inequality $^{(*)}$ follows by the triangle inequality. Rearranging the terms we obtain $d_{X}(v,x)\le\frac{2+\eps}{2}\cdot d_{X}(v,z)\le(1+\frac{\eps}{2})\cdot2r=(2+\eps)\cdot r$, a contradiction.
\end{proof}

        Later we will want to cover the sprawl in a ball $B_v=B_X(v,(2+4\eps)r)$ of radius $(2+4\eps)r$
        by a bounded number of balls of smaller radius, similar as for the doubling dimension:
        the high-level idea is to handle the sprawl similar to low doubling metrics,
        while handling the towns separately by exploiting the fact that they
        form clusters in the sense of \autoref{lem:townproperties}.
        Since any vertex of the sprawl is at distance at most $(2+\eps)r$ from a hub of an $(r,\eps)$-shortest path cover,
        an immediate idea is to cover the sprawl in the ball $B_v$ using balls of radius $(2+\eps)r$ around the hubs.
        Note that the latter balls have smaller radius than~$B_v$ if $\eps>0$.
        To bound the number of these balls of smaller radius, it is tempting to rely on the local sparseness of the shortest path cover.
        However, this only bounds the number of hubs inside~$B_v$, while it may be necessary to use a ball of radius $(2+\eps)r$
        around a hub outside of $B_v$ to cover parts of the sprawl in $B_v$.
        We will therefore prove that in fact there is a bounded number of hubs at distance~$(2+\eps)r$ from any ball of radius $(2+4\eps)r$,
        if the shortest path cover is inclusion-wise minimal. This was shown for \autoref{def:orig-HD} in~\cite{FFKP18}, and
        the following lemma gives the corresponding bound for \autoref{def:HD}. We remark that it does not bound the number of hubs inside a ball of
        radius $(2+\eps)r+(2+4\eps)r=(4+5\eps)r$, and in fact it is possible to construct examples where the number of hubs in such balls is unbounded (cf.~\cite{FFKP18}).
        Also note that the graph $G$ in the following lemma coincides with $X$ if we are considering a finite metric as input (in which case $G$ is simply a complete graph).
        As before, we may set $s=h(\eps)$ in the following two lemmas, if we can compute hitting sets for approximate shortest paths in polynomial time (cf.~\autoref{rem:HittingSetFind}).

 \begin{restatable}[]{lemma}{HubsBallAroundBall}
\label{lem:hub_bound}\RestateRemark
Let $\eps\in[0,1]$ and $r>0$, and let $\SPC$ be an $(r,\eps)$-shortest path cover in the shortest-path metric $X$ of a graph $G$ with highway dimension~$h:\mathbb{R}_{\geq 0}\to\mathbb{N}\cup\{\infty\}$.
	Suppose that $\SPC$ is inclusion-wise minimal and locally $s$-sparse (that is, no ball of radius $(2+4\eps)r$ contains more than~$s$ vertices from~$\SPC$).
    Then for every $v\in X$, there are at most $3s\cdot h(\eps)$ hubs $x \in \SPC$ at distance at most $(2+\eps)r$ from the ball $B_v=B_X(v,(2+4\eps)r)$.
    Formally: $\left|\left\{x\in\SPC\mid d_X\left(x,B_v\right)\le (2+\eps)r\right\}\right|\le3s\cdot h(\eps)$.
\end{restatable}

\begin{proof}%
\autoref{fig:hub-bound} follows the proof of the lemma. Set $W = \{ x \in\SPC \mid d_X(x,B_v) \leq (2+\eps)r\}$
to be the set of hubs near the ball $B_v=B_X(v,(2+4\eps)r)$, whose size we want to bound.
As a starting point, observe that by inclusion-wise minimality of~$\SPC$ and \autoref{def:SPC},
for each $x \in W\subseteq \SPC$ there must be a pair of points $u_x,z_x\in X$ with $d_X(u_x,z_x)\in (r,(2+\eps)r]$ for which
$d_X(u_x,x)+d_X(x,z_x)\leq(1+\eps)d_X(u_x,z_x)$, and accordingly,
there is a $u_x-z_x$ path $Q_x$ through~$x$ of length at most~$(1+\eps)d_X(u,z)$.
Note that for $x\in\SPC$ there might be many different such paths. Pick $Q_x$ to be a path with the properties above that minimizes the distance of its farthest vertex from $v$, i.e., $\max_{b\in V(Q_x)}d_X(b,v)$ is minimized.
Our approach is to bound the points in $W$ based on such paths using local sparsity of $\SPC$ together with the bound on the highway dimension of \autoref{def:HD}.
Indeed, using these properties, in \autoref{clm:balls_aux} below we prove that the number of hubs $x$ in \SPC~whose corresponding path $Q_x$ lies inside $B_X(v,(4+8\eps)r)$ is bounded by~$s\cdot h(\eps)$.
However, there is an issue: there are hubs $x\in W$ whose corresponding path $Q_x$ does not lie inside $B_X(v,(4+8\eps)r)$.
To address this issue, in \autoref{clm:balls_aux} below we will have a general parameter $\ell$ replacing $r$. Then a careful case analysis (where we will take smaller and larger balls around $v$)  will prove the lemma.

\begin{claim}\label{clm:balls_aux}
Consider a value $\ell \in [0,(2+\eps)r)$, and %
let $L\subseteq \SPC$ be a subset such that for each $x\in L$ there exists a path $P_x$ (with endpoints $u_x$ and $z_x$) passing through $x$ for which
\begin{enumerate}[(a)]
    \item $P_x$ lies in the ball $B'_v=B_X(v,(4+8\eps)\ell)$,
    \item $d_X(u_x,z_x)\in(\ell, (2+\eps)r]$, and
    \item the length of $P_x$ is at most $(1+\eps)\cdot d_X(u_x,z_x)$.
\end{enumerate}
Then $|L|\leq s\cdot h(\eps)$.
\end{claim}

\begin{figure}%[][t]
	\includegraphics[width=.9\textwidth]{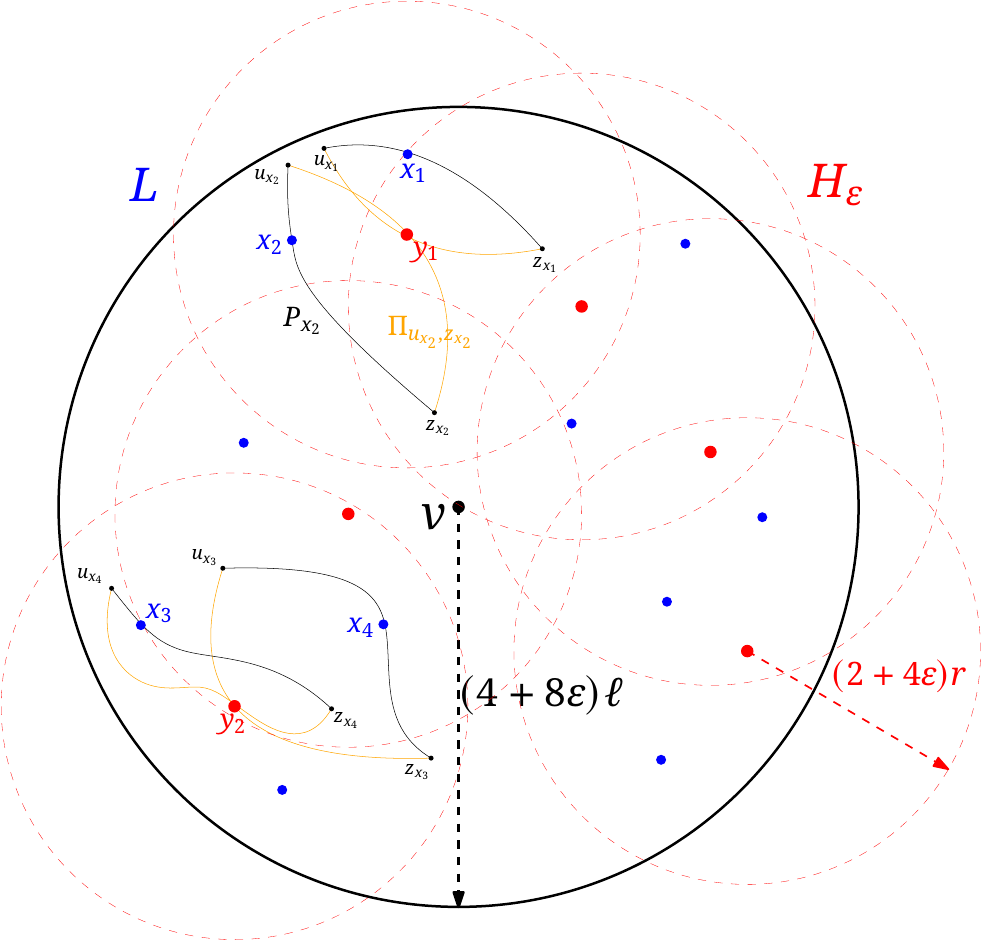}
    \caption{
		Illustration of the proof of \autoref{clm:balls_aux}. The ball $B'_v$ of radius $(4+8\eps)\ell$ around the vertex $v$ is in the center. The set $L\subseteq \SPC$ is illustrated by blue points. Each point $x\in L$ lies on a corresponding approximate shortest path $P_x$ between the endpoints $u_x,z_x$ (illustrated in black).
		Using \autoref{def:HD} we find a hub set $H_\eps$ (illustrated in red) for the ball $B'_v$. For every pair $u_x,z_x$ there is an approximate shortest path $\Pi_{u_x,z_x}$ going through a hub $y\in H_\eps$ (illustrated in orange).  It holds that the distance between $x\in L$ to the corresponding hub $y\in H_\eps$ is at most $d_X(x,y)\le (2+4\eps)r$. In particular the balls of radius $(2+4\eps)r$ around $H_\eps$ (illustrated in dashed red) cover all the points in~$L$. Due to the local sparsity of $\SPC$, it follows that $|L|\le s\cdot |H_\eps|$.}
		\label{fig:ClmGeneralEll}
\end{figure}
\begin{proof}
\autoref{fig:ClmGeneralEll} follows the proof of \autoref{clm:balls_aux}.
Consider any $x \in L$ and the endpoints $u_x$ and $z_x$ of the corresponding path $P_x$.
By property~(a), $u_x$ and $z_x$ lie inside the ball $B'_v=B_X(v,(4+8\eps)\ell)$, and combining properties~(a) and~(c) we get
$d_{G[B'_v]}(u_x,z_x)\leq(1+\eps)d_X(u,z)$.
Also, by property~(b), $d_X(u_x,z_x)>\ell$.
Thus \autoref{def:HD} (applied to variable~$\ell$) implies that there is a hub set
$H_\eps$ with $|H_\eps|\leq h(\eps)$ and such that for every $x\in L$, there is a $u_x-z_x$ path $P'_{u_x,z_x}$ in $G[B'_v]$ of length at most $(1+\eps)d_X(u_x,z_x)$
that is hit by some $y\in H_\eps$ (note that $P'_{u_x,z_x}$ may be different from~$P_x$).
By property~(c), the distance between $x$ and $y$ can be bounded as
\begin{align*}
d_X(x,y) &\leq \min\{d_X(x,u_x)+d_X(u_x,y),\; d_X(x,z_x)+d_X(z_x,y)\} \\
&\leq \textstyle\frac{1}{2}(d_X(u_x,x)+d_X(x,z_x)+d_X(u_x,y)+d_X(y,z_x))\\
&\leq (1+\eps)d_X(u_x,z_x)\\
&\leq (1+\eps)(2+\eps)r~,
\end{align*}
where the last inequality follows from property~(b).
Given that $\eps\in[0,1]$, we get $\eps^2\leq \eps$, and thus we obtain $d_X(x,y)\leq (2+4\eps)r$ from the above inequality.
This implies that $L$ can be covered by at most~$h(\eps)$ balls of
  radius~$(2+4\eps)r$ centered at the vertices in~$H_\eps$. The set $\SPC$, and with
  that also~$L\subseteq \SPC$, is locally $s$-sparse, so each of these balls
  contains at most $s$ hubs, from which we obtain $|L| \leq s\cdot h(\eps)$, as required.
\end{proof}

Fix $\ell_1=r$ and set $B_1=B_X(v,(4+8\eps)r)$. Let $L_1\subseteq W\subseteq\SPC$ be all the hubs $x$ whose corresponding path $Q_x$ lies inside the ball $B_1$. Note that the set $L_1$ satisfies requirements (a), (b), and~(c) from \autoref{clm:balls_aux} (w.r.t.\ $\ell_1$), and hence $|L_1|\le s\cdot h(\eps)$.
However, there might be hubs $x\in L$ whose respective path~$Q_x$ does not lie inside the ball $B_1$, and correspondingly $x\notin L_1$.

Our next goal is to bound the number of such hubs in $W\setminus L_1$.
Consider any $x\in W\setminus L_1$, and let~$b_x$ be a vertex on path $Q_x$ of maximum distance from the center point~$v$ (recall that $Q_x$ is chosen to minimize $d_X(v,b_x)$).
Note that $b_x\notin B_1$, given $x\notin L_1$.
In addition, let $R_x$ be the shortest path between $x$ and its closest vertex~$a_x$ in the ball $B_v=B_X(v,(2+4\eps)r)$.
Note that $R_x$ (trivially) is a path through $x$ of length at most $1+\eps$ times the distance between its endpoints.
Furthermore, as~$x\in W$, $R_x$ is of length at most~$(2+\eps)r$.
Therefore, if $R_x$ was of length more than $r$, then $R_x$ would be a candidate for the choice of $Q_x$.
However, $R_x$ lies inside the ball $B_1=B_X(v,(4+8\eps)r)$, as $a_x\in B_v$ and $w(R_x)\leq (2+\eps)r$, which means that its farthest vertex from $v$ is closer to $v$ than~$b_x\notin B_1$.
In conclusion, $R_x$ must be of length at most~$r$.

\begin{figure}%[t]
	\centering{\includegraphics[width=0.9\textwidth]{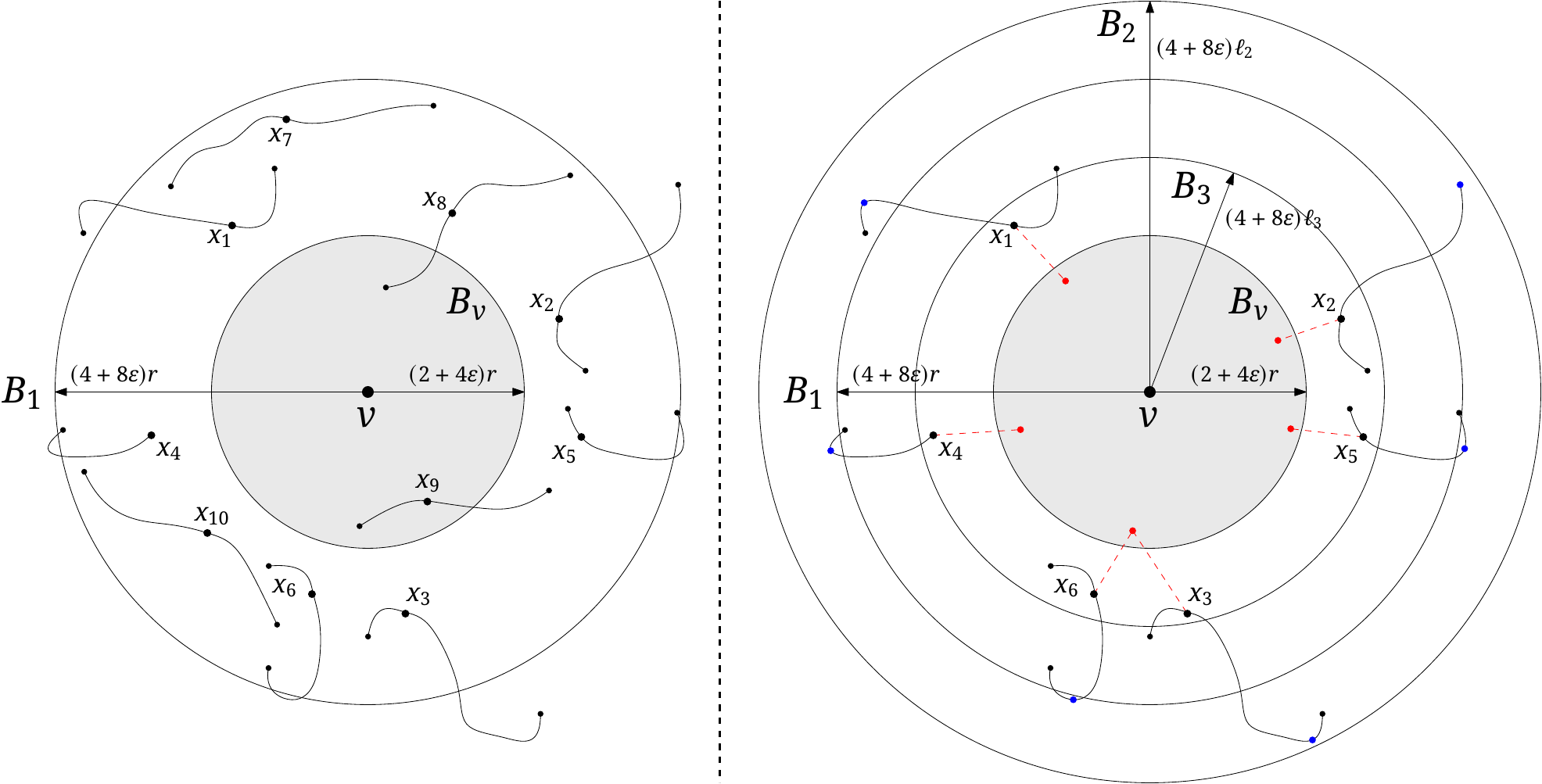}}
	\caption{
		Illustration of the proof of \autoref{lem:hub_bound}. On the left is illustrated a set of
		hubs $W=\left\{ x_{1},\dots,x_{10}\right\} $ at distance $(2+\eps)r$
		from the ball $B_{v}=B_{X}(v,(2+4\varepsilon)r)$. Every hub $x\in W$
		has a corresponding approximate shortest path $Q_{x}$ of length in
		$(r,(2+\eps)r]$, which is chosen to minimize the distance $d_{X}(v,Q_{x})$ to the
		center. $L_{1}=\left\{ x_{7},x_{8},x_{9},x_{10}\right\} \subseteq W$
		is the set of hubs whose corresponding path $Q_{x}$ lies inside $B_{1}=B_{X}(v,(4+8\eps)r)$.
		By \autoref{clm:balls_aux}, $|L_{1}|\le s\cdot h(\eps)$.\\
		On the right is illustrated the set of remaining hubs $W\setminus L_{1}$.
		For a hub $x\in W\setminus L_{1}$, let $a_{x}$ (colored in red) be
		the closest point in $B_{v}=B_{X}(v,(2+4\eps)r)$ to $x$. It holds
		that $d_{X}(x,a_{x})<r$ (as otherwise, by minimality of $d_{X}(v,Q_{x})$, $Q_{x}$
		would have been chosen as the path from $a_{x}$ to $x$). Let $b_{x}$
		be the farthest point on $Q_{x}$ from~$x$ (colored in blue). As
		$d_{X}(x,b_{x})\le w(Q_{x})\le(2+\eps)r$, it follows that $d_{X}(v,b_{x})<(5+8\eps)r$.
		Set $\ell_{2}=(\frac{5+8\eps}{4+8\eps})r$ and $B_{2}=B_{X}(v,(4+8\eps)\ell_{2})$.
		Let $L_{2}=\left\{ x_{1},x_{2},x_{3}\right\} \subseteq W\setminus L_{1}$
		be all the hubs $x$ such that the distance between the corresponding
		endpoints of $Q_{x}$ is at least $\ell_{2}$. By \autoref{clm:balls_aux}, $|L_{2}|\le s\cdot h(\eps)$. \\
		Finally, set $\ell_{3}=(2+4\eps)r-(1+\eps)\ell_{2}$. Using the triangle
		inequality, for $x\in W\setminus(L_{1}\cup L_{2})$ it holds that
		$d_{X}(x,a_{x})\ge\ell_{3}$, while the entire path from $x$ to $a_{x}$
		lies in the ball $B_{3}=B_{X}(v,(4+8\eps)\ell_{3})$. By \autoref{clm:balls_aux}, $|W\setminus(L_{1}\cup L_{2})|\le s\cdot h(\eps)$.
	}\label{fig:hub-bound}
\end{figure}

Given that the length of~$Q_x$ is at most $(1+\eps)d_X(u_x,z_x)$ for the endpoints $u_x,z_x$ of $Q_x$, while~$Q_x$ passes through $x$ and~$b_x$,
we can conclude that $d_{X}(x,b_{x})\leq(1+\eps)d_{X}(u_{x},z_{x})\leq(1+\eps)(2+\eps)r\leq(2+4\eps)r$
 (where we used  $\eps\in[0,1]$).
Hence, using that $a_x$ lies in the ball $B_v$ of radius $(2+4\eps)r$ and $d_X(a_x,x)\leq r$, the distance from $v$ to~$b_x$ is bounded by
\[
\textstyle
d_X(v,b_x)\leq
d_X(v,a_x)+d_X(a_x,x)+d_X(x,b_x)\leq (2+4\eps)r+r+(2+4\eps)r = (5+8\eps)r~.
\]

Set $\ell_2=(\frac{5+8\eps}{4+8\eps})r$ and $B_2=B_X(v,(4+8\eps)\ell_2)=B_X(v,(5+8\eps)r)$. As $b_x$ is the vertex of $Q_x$ farthest from $v$, the path $Q_x$ lies in $B_2$.
Let $L_2\subseteq W\setminus L_1$ contain all hubs $x$ for which $d_X(u_x,z_x)>\ell_2$.
Note that all hubs in $L_2$ fulfill the requirements of \autoref{clm:balls_aux} (w.r.t.\ $\ell_2$). It follows that $|L_2|\le s\cdot h(\eps)$.

Finally, let us consider the hubs $x \in W\setminus (L_1 \cup L_2)$, for which
the shortest path $R_x$ between $x$ and the closest vertex $a_x$ in $B_v$ has
length at most $r$, and the distance between the endpoints $u_x,z_x$ of $Q_x$ lies in the interval $(r,\ell_2]$.
The distance from $x$ to $b_x$ (using a sub-path of $Q_x$) is at most $d_X(x,b_x)\leq(1+\eps)d_X(u_x,z_x)\leq (1+\eps)\ell_2$.
At the same time, the distance between $a_x$ and $b_x$ is more than $(2+4\eps)r$, since $a_x$ lies in the ball $B_v$ of radius $(2+4\eps)r$,
while $b_x$ is not contained in the ball~$B_1$ of radius $(4+8\eps)r$ (as $x\notin L_1$).
By the triangle inequality, it follows that %
\[\textstyle
d_X(a_x,x) \geq d_X(a_x,b_x)-d_X(b_x,x) > (2+4\eps)r - (1+\eps)\ell_2.
\]
Fix $\ell_3= (2+4\eps)r - (1+\eps)\ell_2$ and $B_3=B_X(v,(4+8\eps)\ell_3)$. Then for $x \in W\setminus (L_1 \cup L_2)$ it holds that $d_X(a_x,x)\in(\ell_3,r]\subseteq (\ell_3,(2+\eps)r]$.
Note that the radius of the ball $B_3$ is $(4+8\eps)\cdot\ell_{3}=(4+8\eps)\cdot\left((2+4\eps)r-(1+\eps)\frac{5+8\eps}{4+8\eps}\right)=\left(3+19\eps+24\eps^{2}\right) r$.
We argue that the shortest path $R_x$ from $a_x$ to $x$ fulfills the requirements of \autoref{clm:balls_aux} (w.r.t.\ $\ell_3$). Properties (b) and (c) are clear, and it remains to prove that property (a) holds, that is: $B_3$ contains $R_x$. Recall that $R_x$ is a shortest path from $x$ to $a_x$. In particular, for every $y\in R_x$ it holds that\footnote{The astute reader will notice that there is some slack in this calculation, and thus slightly smaller constants can be chosen in \autoref{def:HD} to make this proof work. However, the proof of \autoref{lem:sparse-SPC} requires the chosen constants.}
\[
d_X(v,y) \leq d_X(v,a_x)+d_X(a_x,y)\leq (2+4\eps)r+ r=(3+4\eps)r < \left(3+19\eps+24\eps^{2}\right) r= (4+8\eps)\ell_3~,
\]
implying that $R_x$ lies in $B_3$.
By \autoref{clm:balls_aux}, for $L_3=W\setminus (L_1 \cup L_2)$, it holds that $|L_3|\leq s\cdot h(\eps)$.

In conclusion, we have that $W \subseteq L_1 \cup
L_2 \cup L_3$, and hence $W$ contains at most $3s\cdot h(\eps)$ hubs.
\end{proof}

The following lemma is similar to \autoref{lem:hub_bound}, however it bounds the number of hubs in a real ball (and not a ball around ball), alas with smaller radius. Previous literature on highway dimension did not contain an analog of this lemma. \autoref{lem:hub_bound_2.8Ball} will be useful in constructing our metric toolkit (in \autoref{sec:MetricToolkit}).

\begin{restatable}[]{lemma}{HubsAlphaBall}
\label{lem:hub_bound_2.8Ball}\RestateRemark
	Let $\eps\in[0,1]$ and $r>0$, and let $\SPC$ be an $(r,\eps)$-shortest path cover in the shortest-path metric $X$ of a graph $G$ with highway dimension~$h:\mathbb{R}_{\geq 0}\to\mathbb{N}\cup\{\infty\}$.
	Suppose that $\SPC$ is inclusion-wise minimal and locally $s$-sparse (that is, no ball of radius $(2+4\eps)r$ contains more than~$s$ vertices from~$\SPC$).
    Then for $\alpha=2.8+6\eps$, every $v\in X$, $\left|B_{X}(v,\alpha\cdot r)\cap\SPC~\right|\le 2s\cdot h(\eps)$.
\end{restatable}

\begin{proof}
	The proof follows similar lines to the proof of \autoref{lem:hub_bound}, and in particular makes use of \autoref{clm:balls_aux}.
	Set $W =B_X(v,\alpha\cdot r)\cap\SPC$.
	By the inclusion-wise minimality of~$\SPC$ and \autoref{def:SPC},
	for each $x \in W\subseteq \SPC$ there must be a pair of points $u_x,z_x\in X$ with $d_X(u_x,z_x)\in (r,(2+\eps)r]$ for which
	$d_X(u_x,x)+d_X(x,z_x)\leq(1+\eps)d_X(u_x,z_x)$, and accordingly,
	there is a $u_x-z_x$ path $Q_x$ through~$x$ of length at most~$(1+\eps)d_X(u,z)$.

	Fix $\ell_1=r$ and set $B_1=B_X(v,(4+8\eps)r)$. Let $L_1\subseteq W\subseteq\SPC$ be all the hubs $x$ whose corresponding path $Q_x$ lies inside the ball $B_1$. The set $L_1$ satisfies requirements (a), (b), and~(c) from \autoref{clm:balls_aux} (w.r.t.\ $\ell_1$), and hence $|L_1|\le s\cdot h(\eps)$.
	However, there might be hubs $x\in L$ whose respective path~$Q_x$ does not lie inside the ball $B_1$, and correspondingly $x\notin L_1$.
	Consider such $x\in W\setminus L_1$, and let~$b_x$ be a vertex on path $Q_x$ of maximum distance from the center point~$v$.
	Note that $b_x\notin B_1$, given $x\notin L_1$.
	As the length of~$Q_x$ is at most $(1+\eps)d_X(u_x,z_x)$ for the endpoints $u_x,z_x$ of $Q_x$, while~$Q_x$ passes through $x$ and~$b_x$,
	we can conclude that $d_{X}(x,b_{x})\leq(1+\eps)d_{X}(u_{x},z_{x})\leq(1+\eps)(2+\eps)r\leq(2+4\eps)r$
	(where we used  $\eps\in[0,1]$).
	Hence, as $x\in W\subseteq B_X(v,\alpha\cdot r)$, the distance from $v$ to~$b_x$ is bounded by
	\[
	d_{X}(v,b_{x})\leq d_{X}(v,x)+d_{X}(x,b_{x})\leq(\alpha+2+4\eps)r~.
	\]

	Set $\ell_2=\frac{\alpha+2+4\eps}{4+8\eps}\cdot r$ and $B_2=B_X(v,(4+8\eps)\ell_2)=B_X(v,(\alpha+2+4\eps)r)$. As $b_x$ is the vertex of $Q_x$ farthest from $v$, the path $Q_x$ lies in $B_2$.
	On the other hand, as $x\notin L_1$, $Q_x$ is not contained in $B_1$, i.e., $d_X(v,b_x)>(4+8\eps)r$. It follows that
	\[
	d_{X}(u_{x},z_{x})\ge\frac{1}{1+\eps}\cdot d_{X}(x,b_{x})\ge\frac{1}{1+\eps}\cdot\left(d_{X}(v,b_{x})-d_{X}(x,v)\right)\ge\frac{1}{1+\eps}\cdot\left(4+8\eps-\alpha\right)\cdot r\overset{(*)}{>}\ell_{2}~,
	\]
	where the inequality $^{(*)}$ holds iff
 $\alpha<\frac{\left(4+8\eps\right)^{2}-\left(2+4\eps\right)\left(1+\eps\right)}{5+9\eps}=\frac{14+58\eps+60\epsilon^{2}}{5+9\eps}=\frac{14}{5}+6\eps+\frac{2.8\eps+6\epsilon^{2}}{5+9\eps}$,
 and indeed we've set $\alpha=2.8+6\eps$.
	It follows that the hubs in $W\setminus L_1$ satisfy  requirements (a), (b), and~(c) from \autoref{clm:balls_aux} (w.r.t.\ $\ell_2$), and hence $|W\setminus L_1|\le s\cdot h(\eps)$. The lemma now follows.
\end{proof}

\section{PTAS for the Subset Traveling Salesman problem}\label{sec:TSP}

\paragraph*{High level overview.}
Our basic approach is somewhat similar to the algorithm of \cite{BGK16} for doubling dimension (this idea was later generalized in \cite{banerjee2024novel}).
Roughly speaking, \cite{BGK16} look for a \enquote{dense} ball $B_X(v,r)$, meaning that the cost of the solution inside the ball is several orders of magnitude more expensive than the radius~$r$. The weight of the MST in the ball is used as a rough estimate for the cost of the solution in the ball, and a ball is called dense if this estimate is very high.
If there is no such ball, then a dynamic program can find a near-optimal solution. Otherwise, let $B_X(v,r)$ be a dense ball of minimum radius.
The idea in \cite{BGK16} is to then solve \TSP on $X\setminus B_X(v,r)$ recursively, and \TSP on $B_X(v,r)$ by exploiting the structure implied by the minimally of the ball. The crux is that stitching the two solutions together is at cost proportional to $r$ (and the doubling dimension), while the cost of the solution in the ball $B_X(v,r)$ is so high, that it can easily absorb the stitching cost.

Our algorithm for bounded highway dimension uses a similar approach, but the details differ significantly. We believe that this approach could also be used for designing approximation algorithms for other problems for metric spaces with low highway dimension. Very roughly speaking: we are looking for the ball of minimum radius with high doubling dimension. If there is no such ball, a solution for doubling metrics can be applied. Otherwise, we show that this ball intersects a huge number of towns. In each town we can compute a near-optimal solution (by minimality of the ball each town is doubling).
We compute the solution for the metric without the towns recursively, and then we stitch the two solutions together via the hubs of a shortest path cover. As the number of towns is very large, while the number of hubs is small, the solutions for the towns can easily absorb the cost of the stitching.
In order to compare the computed solution to the optimum, we show that we can \enquote{massage} the optimum solution so that all the outgoing edges from the towns will be into hubs from the shortest path cover.

We continue with a more detailed description.
For every $i\in\mathbb{N}$, fix $r_i=(1+\sigma)^i$ for a carefully chosen value~$\sigma>0$. For each level $i$ we construct a corresponding shortest path cover $H_i$. We \enquote{massage} the covers so that they will be hierarchical $H_0\supseteq H_1\supseteq H_2\supseteq\dots$, while slightly losing in the sparsity (see Lemmas \ref{lem:H'-is-SPC}, \ref{lem:HubPacking}, and \ref{lem:HubHierarchy}).
Let $i$ be the minimum level such that there is a ball $B_X(v,(2+4\eps)\cdot r_i)$ that intersect many $r_i$-level towns. Note that the town centers, and the hubs constitute a $(2+\eps)\cdot r_i$-net due to \autoref{lem:townproperties}.
The covering radius $(2+\eps)\cdot r_i$ of the net is only slightly smaller than the ball radius of $(2+4\eps)\cdot r_i$, but still enough to prove that if there is no ball intersecting many towns (for any~$i$), then the entire space has small doubling dimension, and thus an approximation scheme for doubling dimension, as found in~\cite{banerjee2024novel}, can be applied.
Otherwise, consider such a ball $B_v=B_X(v,(2+4\eps)\cdot r_i)$ of minimal level $i$ that intersects at least $q$ \, $r_i$-towns for some given threshold~$q$. Note that each such town is contained in a smaller ball, and thus the induced metric on the town is doubling.

Next, we will consider a family of more structured solutions called \emph{hub-net-respecting} (see \autoref{def:hub-net-resp}).
This is similar to net-respecting solutions as used in \cite{BGK16} (but also in other previous work on \TSP and related problems, e.g.~\cite{Talwar04}).
In addition to the hierarchy of hubs $H_0\supseteq H_1\supseteq H_2\supseteq\dots$, we construct a hierarchy of nets $N_0\supseteq N_1\supseteq N_2\supseteq\dots$, where $N_j$ is an $\eps\cdot r_j$-net.
Roughly, a solution is called hub-net-respecting if every edge of weight about $r_j$ is between a hub in $H_j$ to a net point in $N_j$.
Using the properties of hubs, and the density of the nets, one can argue that there is a hub-net-respecting solution of weight at most $(1+O(\eps))$ times the optimum (see \autoref{lem:hub-net-resp}).
Let $O^{HN}$ be an optimal hub-net-respecting tour, and consider its behaviour around the $r_i$-towns intersecting the ball $B_v$.
Every edge of $O^{HN}$ leaving a town must be of weight at least~$r_i$ (as each town is separated from the rest of the metric by \autoref{lem:townproperties}). However, the towns do not contain $r_i$-level hubs. By the structure of the hub-net-respecting tour $O^{HN}$ every edge leaving the town is going to a hub, while the source inside the town is a net point.

By the definition of nets and the diameter of the ball $B_v$, there is at most one town (intersecting~$B_v$) which contains a $\Theta(\frac{r_i}{\eps})$-net point, i.e., of some level $j$ where $\eps r_j=\Theta(\frac{r_i}{\eps})$. We will simply ignore this town (that is, it will be considered as part of the recursively solved part of the metric). As every edge in $O^{HN}$ leaving a town is incident to a net point, it follows that no edges of weight above $\Theta(\frac{r_i}{\eps})$ leaves our towns.
Next, we construct an \emph{interface} $I$ as follows: for every index $i\le j\le i+\Theta(\log_{1+\sigma}\frac1\eps)$, add all the $r_j$-level hubs around $B_v$ to $I$. By the sparsity of our hierarchical hub sets (in particular \autoref{lem:hub_bound}), the size of the interface is rather small.
Now, consider the solution restricted to these towns and their outgoing edges. All the outgoing edges are into the small interface set~$I$, and each town $T$ has small doubling dimension bounded by the threshold~$q$.
We can thus compute a near-optimal solution to each of these towns and connect it to the interface.
Interestingly, we do not rely on dynamic programming in this step, in contrast to many algorithms of this type, e.g., \cite{banerjee2024novel,BGK16,Talwar04}.
Instead we treat the algorithm given by \cite{banerjee2024novel} as a black box to obtain a near-optimal solution for low doubling metrics.
Then we can stitch the tours of the different towns together by paying something proportional to the cost of the minimum spanning tree $\MST(I)$ on the interface. In particular, as we computed a near-optimal tour for each town, the optimal hub-net-respecting tour will differ from our tour again by something proportional to $\MST(I)$. However, $|I|$ is small, while the number of towns is much larger.
Thus by the separation property of the towns (\autoref{lem:townproperties}), the cost of $\MST(I)$ can be absorbed by the cost of connecting the towns to the remaining metric many times over.
Finally, we recursively compute a solution for the rest of the metric, and stitch it together with our solution for the towns (again paying just order $\MST(I)$).

\paragraph*{Formal proof.} We are given a metric $(X,d_{X})$ of highway dimension~$h:\mathbb{R}_{\geq 0}\to\mathbb{N}\cup\{\infty\}$ and we want to compute a near-optimal solution to the \TSP problem.
We will actually solve a more general problem called Subset \TSP.

\begin{definition}[Subset \TSP]\label{def:subsetTSP}
Given a metric space $(X,d_{X})$ and a subset of
terminals $K\subseteq X$ a \emph{Subset \TSP tour} is a closed walk $P$, which visits
each point in $K$ at least once. Note that $P$ is not necessarily
simple (i.e., it can visit points multiple times), and it is allowed to
visit or skip non-terminal (Steiner) points.
The goal of the Subset \TSP problem is to find a Subset \TSP tour $P$ that minimizes the cost $w(P)$.
\end{definition}

In order to compute solutions in towns for which the doubling dimension is bounded, we will use the following result for the special case of \TSP where $K=V$.

\begin{theorem}[{\cite{banerjee2024novel}}]\label{thm:TSP-dd}
For a metric with $n$ points and doubling dimension $d$, a $(1+\eps)$-approximate \TSP tour can be computed in $2^{(d/\eps)^{O(d)}}n+(1/\eps)^{O(d)}n\log n$ time.
\end{theorem}

Throughout this section we fix a value~$\eps\in(0,\frac{1}{6}]$, which will determine the approximation ratio, meaning that our algorithm will compute a $(1+O(\eps))$-approximation.\\
For every~$i\in\mathbb{N}_0$,\footnote{Here, $\mathbb{N}_0$ denotes all natural numbers including 0.} we define $r_i=(1+\sigma)^i$ where $\sigma=\frac{\eps}{4+3\eps}$ and let $\SPC_i$ be an $(r_i,\eps)$-shortest path cover.
Note that the length intervals $(r_i,(2+\eps)r_i]$ for the distances covered by $\SPC_i$ overlap. %
Assume that the minimum distance in the input metric $X$ is slightly larger than $r_0=1$, so that every distance in the metric is covered by some~$\SPC_i$, as $i\geq 0$ (otherwise we can scale accordingly).
For any closed walk of a metric we call a pair of consecutive nodes $u$ and $x$ a \emph{connection} (similar to an edge, but it might correspond to a shortest path between the endpoints in some underlying graph).

\subsection{Hierarchical hubs and nets}

In order to define towns and sprawl that give the structure of sub-instances we want to solve, we need the hub sets of shortest path covers.
However, we will need a hub set that forms a hierarchy, i.e., any hub on a level $i$ is also a hub on a lower level.
Unfortunately, the shortest path covers $\SPC_i$ given by \autoref{def:SPC} do not provide this property.
We therefore define a new hub set $H_i$ for each $i\in\mathbb{N}_0$ based on $\SPC_i$, such that $H_0\supseteq H_1\supseteq\ldots\supseteq H_L$, where $L=\lceil\log_{1+\sigma}(\diam(X))\rceil$ is the highest level of the given metric~$X$.
To obtain each~$H_i$, we begin by constructing an auxiliary hub set $H'_i$ and then taking $H_i=\bigcup_{j \geq i}H'_j$ to enforce a hierarchy.
Each~$H'_i$ is constructed in a top-down manner as follows.
On the highest level $L$, the inclusion-wise minimal shortest path cover~$\SPC_L$ is empty, since there are no vertices at distance more than~$r_L$.
Accordingly we let~$H'_{L}=\emptyset$.
When constructing $H'_{i}$ for $i<L$ (after we already constructed $H'_{j}$ for~$j\geq i+1$), we examine each point in $\SPC_i$ iteratively: when we consider a hub~$x\in\SPC_i$,
\begin{enumerate}
    \item if there is already a hub $y\in H'_i$ at distance at most $\frac{\eps}{4} r_i$ from $x$ we will ignore $x$,
    \item otherwise, if there is a hub $y\in H'_{j}$ for $j>i$ such that $d_{X}(x,y)\leq\frac{\eps}{4}r_i$, then we will add $y$ to~$H'_i$,
    \item otherwise (there is no such $y$), we will add $x$ to $H'_i$.
\end{enumerate}

Before finalizing the construction of the hub sets, let us prove the following simple fact.

\begin{lemma}\label{lem:H'-is-SPC}
For every $i$, $H'_i$ is an $(r_i,\frac{3}{2}\eps)$-shortest path cover.
\end{lemma}
\begin{proof}
    Consider a pair $u,z\in X$ with $d_X(u,z)\in (r_i,(2+\eps)r_i]$.
    By \autoref{def:SPC} there is a hub $x\in\SPC_i$ with $d_X(u,x)+d_X(x,z)\leq(1+\eps)d_X(u,z)$.
    From the construction of $H'_i$, there is some hub $y\in H'_i$ at distance at most $\frac{\eps}{4}r_{i} < \frac{\eps}{4}d_X(u,z)$ from $x$.
    Thus using the triangle inequality, we have
    \[\textstyle
    d_X(u,y)+d_X(y,z)\leq d_X(u,x)+2d_X(x,y)+d_X(x,z)\leq(1+\frac{3}{2}\eps)d_X(u,z)~,
    \]
    i.e., $H'_i$ is an $(r_i,\frac{3}{2}\eps)$-shortest path cover.
\end{proof}

Since $H'_i$ is a shortest path cover, we may make it inclusion-wise minimal by removing any redundant hubs.
After this step we finally let~$H_i=\bigcup_{j\geq i} H'_j$ (note that we cannot make $H_i$ inclusion-wise minimal without losing the hierarchy property).
As a result we get the following crucial packing property (cf.~\autoref{def:net}).

\begin{lemma}[Hub Packing]\label{lem:HubPacking}
    For each $i\in\mathbb{N}_0$, the hub set $H_i$ is an $\frac{\eps}{4}r_{i}$-packing (that is, $\forall x,y\in H_i$, $d_X(x,y)>\frac\eps4 r_i$).
\end{lemma}
\begin{proof}
    The packing property follows by induction: on the highest level, $H_L=\emptyset$ trivially is an $\frac{\eps}{4}r_L$-packing.
    For any $i<L$ consider the set $H_{i}$, which is a superset of $H_{i+1}$.
    By induction, $H_{i+1}$ is an $\frac{\eps}{4}r_{i+1}$-packing and thus also an $\frac{\eps}{4}r_i$-packing, as~$r_i<r_{i+1}$.
    Any hub $x\in H_i$ that is not already in~$H_{i+1}$, lies in $H'_i$ and was only added if there was no other hub of $H'_j$ for $j\geq i$ at distance at most~$\frac{\eps}{4}r_{i}$.
    Given that $H_i=\bigcup_{j\geq i} H'_j$, this means that $H_{i}$ is an $\frac{\eps}{4}r_{i}$-packing.
\end{proof}

From this lemma we derive the following properties for the hierarchical hub sets $H_i$, which are similar to those for~$\SPC_i$ given by \autoref{def:SPC} and \autoref{lem:hub_bound} where $s=h(\eps)$.

\begin{lemma}\label{lem:HubHierarchy}
If $\SPC_i$ is locally $h(\eps)$-sparse for every $i$, then $H_i$ is an $(r_i,\frac{3}{2}\eps)$-shortest path cover that is locally $O(\frac{1}{\eps}\log(\frac{1}{\eps})\cdot h(\eps))$-sparse.
Furthermore, for each ball $B_v=B_X(v,(2+4\eps)r_i)$ the number of hubs in $H_i$ at distance at most $(2+\eps)r_i$ from $B_v$ is at most $O((\frac{1}{\eps}\log(\frac{1}{\eps})\cdot h(\eps))^2)$.
\end{lemma}
\begin{proof}
   By the  definition of $H_i$, and  \autoref{lem:H'-is-SPC}, $H_i$ is an $(r_i,\frac{3}{2}\eps)$-shortest path cover.
    To bound the local sparsity of $H_i$, consider any ball $B_v=B_X(v,(2+4\eps)r_i)$.
    Since~$H_i$ consists of subsets of hubs from $\SPC_\ell$ for $\ell\geq i$, $H_i\cap B_v$ contains at most $h(\eps)$ hubs of~$\SPC_\ell$ for any~$\ell\geq i$, given that $\SPC_\ell$ is locally $h(\eps)$-sparse.
    At the same time, for any $\ell\geq i$ the hubs of~$\SPC_\ell$ chosen for~$H_i$ lie in $H_{\ell}$, as these hub sets form a hierarchy.
    By \autoref{lem:HubPacking}, $H_{\ell}$ is an $\frac{\eps}{4}r_{\ell}$-packing.
    Letting~$j=i+\lceil\log_{1+\sigma}(\frac{16}{\eps}+32)\rceil$, we get that $\frac{\eps}{4}r_{j}=\frac{\eps}{4}(1+\sigma)^{j}\ge\frac{\eps}{4}\cdot r_{i}\cdot(\frac{16}{\eps}+32)=2\cdot(2+4\eps)\cdot r_{i}$
     is at least the diameter of the ball~$B_v$, and thus there is at most one hub of $\bigcup_{\ell\geq j}\SPC_\ell$ in $H_i\cap B_v$.
    Noting that $\sigma=\frac{\eps}{4+3\eps}\geq \frac{2\eps}{9}$ (as~$\eps\leq \frac16$), we can upper bound the total number of hubs of $H_i$ in the ball $B_v$ by
    \[\textstyle
    O(\log_{1+\sigma}(\frac{16}{\eps}+32)\cdot h(\eps))= O(\frac{1}{\eps}\log(\frac{1}{\eps})\cdot h(\eps))~.
    \]
    That is, $H_i$ is locally $O(\frac{1}{\eps}\log(\frac{1}{\eps})\cdot h(\eps))$-sparse.

    To bound the number of hubs in $H_i$ close to a ball $B_v$ we would like to apply \autoref{lem:hub_bound}.
    However, $H_i$ is not necessarily inclusion-wise minimal.
    Instead, we may apply \autoref{lem:hub_bound} to each inclusion-wise minimal~$H'_\ell$ for~$\ell\geq i$ that $H_i$ is comprised of.
    That is, since $H'_\ell\subseteq H_i$ is locally $O(\frac{1}{\eps}\log(\frac{1}{\eps})\cdot h(\eps))$-sparse, due to \autoref{lem:hub_bound}, $H'_\ell$ contains $O(\frac{1}{\eps}\log(\frac{1}{\eps})\cdot h(\eps)^2)$ hubs at distance at most $(2+\eps)r_i$ from~$B_v$.
    Moreover, whenever $\ell\geq i+\log_{1+\sigma}(\frac{4+5\eps}{2+4\eps})$, all hubs of $H'_\ell$ at distance at most $(2+\eps)r_i$ from~$B_v$ lie the ball $B_X(v,(2+4\eps)r_\ell)$, since $r_\ell=(1+\frac{\eps}{4+3\eps})^\ell$ and the radius of $B_v$ is $(2+4\eps)r_i$.
    Using the local $O(\frac{1}{\eps}\log(\frac{1}{\eps})\cdot h(\eps))$-sparsity of~$H_\ell$ and the hierarchy property, this means that if $j=i+\lceil\log_{1+\sigma}(\frac{4+5\eps}{2+4\eps})\rceil$, only $O(\frac{1}{\eps}\log(\frac{1}{\eps})\cdot h(\eps))$ hubs of $\bigcup_{\ell=j}^{L}H'_\ell=H_j$ are at distance at most $(2+\eps)r_i$ from~$B_v$.
    Given that $H_i=H_j\cup \bigcup_{\ell=i}^{j-1}H'_\ell$, the total number of hubs in $H_i$ at distance at most $(2+\eps)r_i$ from~$B_v$ can be upper bounded by
    \[\textstyle
    O(\frac{1}{\eps}\log(\frac{1}{\eps})\cdot h(\eps))+ \log_{1+\sigma}\left(\frac{4+5\eps}{2+4\eps}\right)\cdot O(\frac{1}{\eps}\log(\frac{1}{\eps})\cdot h(\eps)^2) =
    O((\frac{1}{\eps}\log(\frac{1}{\eps})\cdot h(\eps))^2)
    \]
    where we used $\log_{1+\sigma}(\frac{4+5\eps}{2+4\eps})= O(\frac{1}{\eps}\log(\frac{1}{\eps}))$, similar to above.
    This concludes the proof.
\end{proof}

In order to compute solutions for the towns and the remaining instance, we define a special more structured type of solution. %
For this, in addition to a hub hierarchy we will also need a net hierarchy:
let $N_0\supseteq N_1\supseteq \ldots \supseteq N_L$ be a hierarchy of nets of the given metric where each $N_i\subseteq X$ is an $\eps r_i$-net of $X$.
We can compute such a hierarchy using the Gonzales order of $X$ \cite{Gon85}.\footnote{The Gonzales order of $X$ \cite{Gon85} is defined as follows: $x_1$ is chosen arbitrary. $x_2\in X$ is the farthest point from $x_1$, and in general $x_{i+1}$ is the farthest point from $\{x_1,\dots,x_i\}$. Observe that if $d_X(x_{i+1},\{x_1,\dots,x_i\})=\delta$, then $\{x_1,\dots,x_i\}$ is a $\delta$-net. In particular, the prefixes constitute a hierarchy of nets.}
Our algorithm will compute a solution that is structured with respect to the hub and net hierarchies.
As a first step, we begin by making a solution \emph{net-respecting}, which means that for every connection $uz$ with $d_X(u,z)\in(r_i,r_{i+1}]$ it holds that $u,z\in N_i$.

\begin{lemma}\label{lem:net-resp}
Given any walk $P$, in polynomial time we can compute a net-respecting walk $P^{N}$ visiting a superset of the nodes of $P$ with cost $w(P^{N})\leq (1+60\eps) w(P)$.
\end{lemma}
\begin{proof}
We transfer the given walk $P$ to a net-respecting walk~$P^{N}$ following an argument similar in spirit to \cite{BGK16}.
This will be done in a greedy manner, where in each iteration we fix the violating connection of maximal length.
Specifically, while there is a connection~$uz$ with $d_X(u,z)\in(r_i,r_{i+1}]$ for which $u$ or $z$ does not belong to $N_i$, let $uz$ be such a connection of maximum length.
Fix $c=16$, and let $u',z'\in N_{i+c}$ be the closest net points from the net $N_{i+c}$ to $u,z$ respectively.
We will replace the connection $uz$ with the three connections $uu'$, $u'z'$, and $z'z$.
First note that $u'z'$ is not a violating connection any longer. Indeed, as $N_{i+c}$ is an $\eps r_{i+c}$-net and $r_i=(1+\sigma)^i$,
\begin{align*}
d_{X}(u',z') & \le d_{X}(u',u)+d_{X}(u,z)+d_{X}(z,z')\le2\eps\cdot r_{i+c}+r_{i+1}\\
 & =\left(1+\sigma+2\eps\cdot(1+\sigma)^{c}\right)\cdot r_{i}\overset{(*)}{<}(1+\sigma)^{c}\cdot r_{i}=r_{i+c}~,
\end{align*}
where the inequality $^{(*)}$ holds as $(1+\sigma)^{c}\cdot(1-2\eps)\ge1+\sigma$, which by itself holds as $(1+\sigma)^{c}\overset{(**)}{>}(1+\sigma)(1+3\eps)\overset{(***)}{\ge}\frac{1+\sigma}{1-2\eps}$.
Here inequality $^{(**)}$ holds as $\sigma>\frac{\eps}{5}$ and thus
$(1+\sigma)^{c-1}>(1+\frac{\eps}{5})^{15}>1+3\eps$, while inequality $^{(***)}$ holds as $\eps\le\frac16$.
As $d_{X}(u',z')\in (r_j,r_{j+1}]$ for some $j<i+c$, and $u',z'\in N_{i+c}\subseteq N_j$, it follows that $u'z'$ is indeed not violating.

We continue this process until no violating connections remain.
In each step, we replace a violating connection $uz$, with a non violating connection $u'z'$, and two potentially violating connections: $u'u$, $z'z$. However, the weight of these two new connections is bounded by: $\eps\cdot(1+\sigma)^{i+c}=\eps\cdot(1+\sigma)^{c}\cdot r_i<\frac{1}{2}\cdot d_{X}(u,z)$. Here in the last inequality we used that $\eps\le\frac16$, and $\sigma\le\frac\eps4\le\frac{1}{24}$.
As there is a minimum pairwise distance in $X$, it follows that this process will eventually halt.

Finally, fix $\gamma=60$. We argue that the weight of the resulting walk $P^N$ is bounded by $1+\gamma\eps$ times the given walk $P$.
Specifically, we will prove that for a connection~$uz$, the total weight of all the edges in the entire process replacing  $uv$ is at most $(1+\gamma\eps)\cdot d_X(u,z)$, the bound on the total length of $P^N$ follows.
The proof is by induction on the length $d_X(u,z)$.
If $uz$ is not a violating connection, there is nothing to prove.
Otherwise, it was replaced by the connections $u'u$, $u'z'$, $z'z$ as above.
The new connection $u'z'$ is not violating and will not be replaced again, while using the induction hypothesis on $u'u$, $z'z$, we conclude that the total weight of connections replacing $uz$ is bounded by
\begin{align*}
d_{X}(u',z')+(1+\gamma\eps)\cdot\left(d_{X}(u',u)+d_{X}(z',z)\right) & \le r_{i+c}+(1+\gamma\eps)\cdot2\cdot\eps\cdot r_{i+c}\\
 & \le\left(1+(1+\gamma\eps)\cdot2\cdot\eps\right)\cdot(1+\sigma)^{c}\cdot r_{i}\\
 & \overset{(*)}{\le}(1+\gamma\eps)\cdot r_{i}\le(1+\gamma\eps)\cdot d_{X}(u,v)~.
\end{align*}
Here in the inequality $^{(*)}$ we used that $(1+\sigma)^{c}\le e^{\sigma c}\le1+2\sigma c\le1+8\eps$, and that $\left(1+(1+\gamma\eps)\cdot2\cdot\eps\right)\cdot(1+8\eps)=1+10\eps+(16+2\gamma)\cdot\eps^{2}+16\gamma\cdot\eps^{3}<1+(13+(\frac{1}{3}+\frac{16}{36})\cdot\gamma)\eps<1+\gamma\eps$.
\end{proof}

We now introduce the required structure of walks, which respects the hubs as well as the nets. %
Roughly, we require that every connection in the walk which is outgoing from a town, will be out of a net point, and into a hub (scaled accordingly).
In the following we denote by $\mathcal{T}_\ell$ the set of towns (as defined in \autoref{def:towns}) w.r.t.\ the hubs $H_\ell$ on level $\ell$ of the hub hierarchy.

\begin{definition}\label{def:hub-net-resp}
A walk $P^{HN}$ is called \emph{hub-net-respecting} (w.r.t.\ the hub hierarchy $\{H_j\}_{j\in\mathbb{N}_0}$ and net hierarchy $\{N_j\}_{j\in\mathbb{N}_0}$) if for any indices $j,\ell\in\mathbb{N}_0$ and town $T\in\mathcal{T}_\ell$, every connection $ux$ of $P^{HN}$ for which $d_X(u,x)\in((1+\frac{3}{4}\eps)r_j,(1+\frac{3}{4}\eps)r_{j+1}]$, $u\in T$, and $x\notin T$ has the property that $u\in N_k$ and $x\in H_k$ where $k=\max\{\ell,j\}$.
\end{definition}

\begin{lemma}\label{lem:hub-net-resp}
For any walk $P$, in polynomial time we can compute a hub-net-respecting walk $P^{HN}$ visiting a superset of the nodes of $P$ with cost $w(P^{HN})\leq (1+77\eps)w(P)$.
\end{lemma}

\begin{proof}
Given a walk $P$, we begin by constructing a net-respecting walk $P^N$ using \autoref{lem:net-resp}.
Next, we construct a walk $P^{HN}$ using the following procedure:
iteratively consider each connection $uz$ of $P^N$ and let $k$ be the index for which $d_X(u,z)\in(r_k,r_{k+1}]$.
As $r_{k+1}=(1+\sigma)r_k<(2+\eps)r_k$, according to \autoref{lem:HubHierarchy} there is a hub $x\in H_k$ such that $d_X(u,x)+d_X(x,z)\leq(1+\frac{3}{2}\eps)d_X(u,z)$.
Replace the connection $uz$ by $ux$ and $xz$.

To prove that $P^{HN}$ is hub-net-respecting, consider a town $T\in\mathcal{T}_\ell$ on some level $\ell$.
Let $ux$ be any connection of $P^{HN}$, which was obtained from a connection $uz$ of $P^N$ with $d_X(u,z)\in(r_k,r_{k+1}]$ for some $j$ according to the above procedure, so that $x\in H_k$ and $u\in N_k$ as~$P^N$ is net-respecting.
Assume that one endpoint of $ux$ lies in $T$ and the other outside of $T$, which by
\autoref{lem:townproperties} implies that $d_X(u,x)> r_\ell$.

%\noindent
%\begin{minipage}{\textwidth} %
%\setlength{\parindent}{1.5em}
%\begin{wrapfigure}{r}{0.2\textwidth}
%	\begin{center}
%		\vspace{-20pt}
%		\includegraphics[width=0.8\linewidth]{fig_edited/HubNetRespectingCases}
%		\vspace{-5pt}
%	\end{center}
%	\vspace{-15pt}
%\end{wrapfigure}
% \begin{sidefigure}
% 		\includegraphics[width=0.8\linewidth]{fig_edited/HubNetRespectingCases}
%     \caption{Illustration of the two cases in the proof of \autoref{lem:hub-net-resp}.}
%     \label{fig:TBP2}
% \end{sidefigure} 
\begin{figure}		\includegraphics[width=0.6\linewidth]{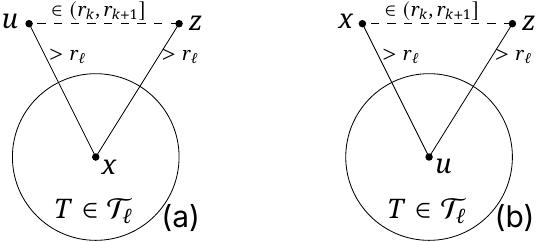}
    \caption{Illustration of the two cases in the proof of \autoref{lem:hub-net-resp}.}
    \label{fig:TBP2}
\end{figure} 

Assume first that $x\in T$ and~$u\notin T$ (\autoref{fig:TBP2}).
Since $x\in H_k$ but $T$ does not contain any hub of $H_\ell$ according to \autoref{def:towns}, the fact that the hubs form a hierarchy implies $\ell\geq k+1$.
This means that~$z$ cannot lie in $T$, as otherwise $d_X(u,z)>r_\ell\geq r_{k+1}$ by \autoref{lem:townproperties}, contradicting $d_X(u,z)\in(r_k,r_{k+1}]$.
Hence~$z\notin T$, which again by \autoref{lem:townproperties} now implies that $d_X(x,z)> r_\ell$.
But then $d_X(u,x)+d_X(x,z)>2r_\ell\geq 2r_{k+1}$ while at the same time $d_X(u,x)+d_X(x,z)\leq(1+\frac{3}{2}\eps)d_X(u,z)\leq (1+\frac{3}{2}\eps)r_{k+1}$, a contradiction (as~$\eps\leq\frac{1}{6}$).

Thus we can conclude that $u\in T$ and $x\notin T$ (illustration (b) on the right).
Now, it cannot be that $z\in T$, since by \autoref{lem:townproperties} this would mean that $2 d_X(u,z)\leq 2r_\ell< d_X(u,x)+d_X(x,z)\leq(1+\frac{3}{2}\eps)d_X(u,z)$, which is a contradiction. %
Therefore, $z\notin T$ and we get $d_X(u,z)>r_\ell$ again by \autoref{lem:townproperties}, so that $k\geq\ell$ as $d_X(u,z)\in(r_k,r_{k+1}]$.
Since the hubs form a hierarchy, this means that $x\in H_k\subseteq H_\ell$.

We established that $u\in N_k$ and $x\in H_k$, but we still need to show that $k\geq\max\{\ell,j\}$ where~$j$ is the index for which $d_X(u,x)\in((1+\frac{3}{4}\eps)r_j,(1+\frac{3}{4}\eps)r_{j+1}]$.
Indeed, this will establish that $P^{HN}$ is hub-net-respecting, since the nets and hubs form a hierarchy and thus $u$ and $x$ will be a net point and a hub, respectively, on level $\max\{\ell,j\}$.
We already know that $k\geq \ell$, and thus it suffices to show that $k\geq j$.
For this we first prove the following useful property, which we will reuse later on.

\begin{claim}\label{clm:dist-u-x}
Given $d_X(u,x)+d_X(x,z)\leq(1+\frac{3}{2}\eps)d_X(u,z)$ we have $d_X(u,x)\leq (1+\frac{3}{4}\eps)d_X(u,z)$.
\end{claim}

\begin{subproof}
By the triangle inequality we get $d_X(x,z)\geq d_X(u,x)-d_X(u,z)$.
Together with the given inequality we obtain
\[\textstyle
d_X(u,x)\leq(1+\frac{3}{2}\eps)d_X(u,z)-d_X(x,z)\leq (1+\frac{3}{2}\eps)d_X(u,z)-(d_X(u,x)-d_X(u,z))~,
\]
which we can solve for $d_X(u,x)$ to get $d_X(u,x)\leq (1+\frac{3}{4}\eps)d_X(u,z)$.
\end{subproof}

By this claim, since $d_X(u,z)\leq r_{k+1}$ we have $d_X(u,x)\leq (1+\frac{3}{4}\eps)r_{k+1}$, which implies $j\leq k$, as required.

Finally, to bound the length of $P^{HN}$, recall that $w(P^N)\leq(1+60\eps)w(P)$ according to \autoref{lem:net-resp}.
Since each connection of $P^{N}$ is replaced by a detour through a hub that is at most a factor of $1+\frac{3}{2}\eps$ longer, the length of $P^{HN}$ can be bounded by $w(P^{HN})\leq (1+\frac{3}{2}\eps)w(P^N)\leq (1+\frac{3}{2}\eps)(1+60\eps)w(P)< (1+77\eps)w(P)$, as $\eps\leq\frac{1}{6}$.
\end{proof}

\subsection{The divide step}

Let $\mathcal{T}_i$ be the set of towns for the shortest path cover~$H_i$, as given by \autoref{def:towns}.
Roughly speaking, the next lemma shows that the metric has bounded doubling dimension if the number of towns of every $\mathcal{T}_i$ intersecting each ball of level~$i$ is bounded.
However, we will focus our attention on the metric induced by the terminal set $K$ lying in some given town~$T\in\mathcal{T}_i$ (recall that we are solving the subset TSP problem, see \autoref{def:subsetTSP}).
The upper bound~$q$ is then on the number of towns for each level~$j$, that contain a terminal and intersect a ball corresponding to level~$j$ inside $T$.
Here we let $q$ be arbitrary, but we will choose a specific value for~$q$ as a function of the highway dimension and $\eps$ for our purposes later on.

\begin{lemma}\label{lem:towns-w-small-dd} 
Let $K$ be a set of terminals in the metric $X$.
Assume that for every $j\in\mathbb{N}_0$ and $v\in X$, the ball $B_v^j=B_X(v,(2+4\eps)r_j)$ intersects at most~$q$ towns of~$\mathcal{T}_j$ containing a terminal, i.e.,
\[
|\{T\in\mathcal{T}_j\mid (T\cap B_v^j)\neq\emptyset\ \land\ (T\cap K)\neq\emptyset\}|\leq q.
\]
Then the metric induced by $K$ has doubling dimension $O\left(\log(q)+\frac{1}{\eps}\log(\frac{h(\eps)}{\eps})\right)$.
\end{lemma}

\begin{proof}
To bound the doubling dimension of $K$ we need to show that for any $\ell>0$ and $v\in K$, the ball $B_{2\ell}=B_X(v,2\ell)\cap K$ of radius $2\ell$ around $v$ restricted to $K$ can be covered by respectively bounded number of balls of radius~$\ell$.
For this, let $j\in\mathbb{N}_0$ be the smallest index for which $B_X(v,2\ell)\subseteq B_v^j$.
Note that this means that $2\ell>(2+4\eps)r_{j-1}>2(1+\sigma)r_{j-1}=2r_j$ as $\sigma=\frac{\eps}{4+3\eps}$, and so~$\ell>r_j$.
By~\autoref{lem:townproperties}, each town in $\mathcal{T}_j$ has diameter at most $r_j$ and can thus be covered by a ball of radius~$\ell$.
To cover the terminals of $B_{2\ell}$ lying in towns of $\mathcal{T}_j$ it thus suffices to use at most~$q$ balls of radius $\ell$, which can be centered at terminals.

We still need to cover any terminals in $B_{2\ell}$ that do not lie in a town of $\mathcal{T}_j$.
By \autoref{def:towns} these terminals belong to the sprawl, and by \autoref{lem:townproperties} the sprawl can be covered by balls of radius~$(2+\eps)r_j$ around hubs of $H_j$.
To cover these terminals in $B_{2\ell}\subseteq B_v^j$, we only need to utilize hubs at distance at most $(2+\eps)r_j$ from $B_v^j$ as ball centers, and by \autoref{lem:HubHierarchy} there are $O((\frac{1}{\eps}\log(\frac{1}{\eps})\cdot h(\eps))^2)$ such hubs.
The radius $(2+\eps)r_j$ of these balls is larger than $\ell$, but since $\eps\in(0,1]$ we have $(2+\eps)r_j=(2+\eps)(1+\sigma)r_{j-1}< (2+4\eps)r_{j-1}$.
For each of these balls we may therefore cover the terminals in the towns of $\mathcal{T}_{j-1}$ using at most $q$ balls of radius $\ell$ centered at terminals.
The terminals in the sprawl for $H_{j-1}$ lying in one of these balls of radius~$(2+4\eps)r_{j-1}$ can in turn be covered by $O((\frac{1}{\eps}\log(\frac{1}{\eps})\cdot h(\eps))^2)$ balls of radius~$(2+4\eps)r_{j-2}$.
We may repeat this until the balls are small enough, i.e., of radius at most $\ell$.
However, note that the balls covering the sprawl are not necessarily centered at terminals, but to bound the doubling dimension of $K$, ultimately we need the balls of half the radius covering $B_{2\ell}$ to have terminal centers.
We therefore repeat this process until we reach a level $j'$ for which~$(2+4\eps)r_{j'}\leq \ell/2$.
At this point we may pick one terminal from each such ball and cover this ball by a ball of radius $2(2+4\eps)r_{j'}\leq\ell$ centered at the chosen terminal.

Given that $j'$ is the largest index for which $(2+4\eps)r_{j'}\leq \ell/2$ while $j$ is the smallest index for which $(2+4\eps)r_j\geq 2\ell$, the number of levels for which we repeat the above process is $\lceil\log_{1+\sigma}(4)\rceil\leq \lceil\frac{10}{\eps}\rceil$ as~$\eps\leq \frac{1}{6}$.
Since on each level we cover the terminals in towns using $q$ balls of radius $\ell$ and we then recurse on $O((\frac{1}{\eps}\log(\frac{1}{\eps})\cdot h(\eps))^2)$ balls of the next level, by a simple recursive formula the total number of balls of radius~$\ell$ needed to cover the terminals of $B_{2\ell}$ is at most
\[
{\textstyle
O((\frac{1}{\eps}\log(\frac{1}{\eps})\cdot h(\eps))^{2\lceil\frac{10}{\eps}\rceil})
}
+\sum_{k=0}^{\lceil\frac{10}{\eps}\rceil-1}
{\textstyle
O((\frac{1}{\eps}\log(\frac{1}{\eps})\cdot h(\eps))^{2k})\cdot q
\ =\ (\frac{1}{\eps}\log(\frac{1}{\eps})\cdot h(\eps))^{O(1/\eps)} \cdot q\ .
}
\]
Thus the doubling dimension is $O\left(\log(q)+\frac{1}{\eps}\log(\frac{h(\eps)}{\eps})\right)$.
\end{proof}

Since $H_L$ is empty, on the highest level $L$, $\cT_L$ is a singleton town which is simply the entire metric~$X$. Hence \autoref{lem:towns-w-small-dd} implies that if no ball $B_X(v,(2+4\eps)r_j)$ on any level~$j$ intersects more than $q$ towns of $\mathcal{T}_j$ containing terminals, then the metric induced by the terminals~$K$ has bounded doubling dimension.
In this case we can invoke \autoref{thm:TSP-dd} to obtain a near-optimal Subset \TSP solution for the metric~$X$, by solving the \TSP instance induced by $K$.
The runtime will be as desired if we carefully choose $q$ as a function of the highway dimension and~$\eps$.

Otherwise, we identify a minimum index $i\in\mathbb{N}_0$ for which more than $q$ towns with terminals intersect the corresponding ball $B_v=B_X(v,(2+4\eps)r_i)$ of this level $i$.
That is, if
\[
\mathcal{U}'=\{T\in\mathcal{T}_i\mid (T\cap B_v)\neq\emptyset\ \land\ (T\cap K)\neq\emptyset\}
\]
then $i$ is the smallest index for which $|\mathcal{U}'|> q$.
Furthermore, we will exclude a town that contains a certain net point of a higher level from this bound.
More concretely, let $y_N\in N_j$ be a net point of level $j=i+\lceil\log_{1+\sigma}(\frac{6}{\eps}+8)\rceil$ lying in the ball $B_X(v,(3+4\eps)r_i)$.
Note that there is at most one such point since $N_j$ is an $\eps r_j$-net and for this value of $j$ we get $\eps r_{j}\ge\eps\cdot(\frac{6}{\eps}+8)\cdot r_{i}=2(3+4\eps)r_{i}$, i.e., the distances between net points of $N_j$ are larger than the diameter of the ball.
Observe that the ball $B_X(v,(3+4\eps)r_i)$ contains $B_v$ but also all towns intersecting the ball~$B_v$, since the latter has radius $(2+4\eps)r_i$ and each town has diameter at most $r_i$ according to \autoref{lem:townproperties}.
Now we let $\mathcal{U}$ be the set of all towns from $\mathcal{U}'$ except the town containing $y_N$ (if any),
so that $|\mathcal{U}|>q-1$.

By minimality of $i$ we may apply \autoref{lem:towns-w-small-dd} to each terminal set contained in a town of $\mathcal{U}$, to conclude that every such town has bounded doubling dimension.
The idea is to solve each of the towns of $\mathcal{U}$ using \autoref{thm:TSP-dd}, and then stitch these solutions together with a recursively computed solution to the remaining instance.
The seam at which these solutions will be stitched together is given by a set $I\subseteq X$, which we call the \emph{interface} and contains all hubs in $H_j$ at distance at most $(2+\eps)r_j$ from the ball~$B_v$, for every $j\in\mathbb{N}_0$ with $i\leq j\leq i+\log_{1+\sigma}(\frac{6}{\eps}+8)$. Formally,
\[\textstyle
\mbox{\emph{interface}\qquad}I=\left\{ x\in H_{j}\cap B_{X}\left(B_{v},(2+\eps)\cdot r_{j}\right)\,\mid\,i\leq j\leq i+\log_{1+\sigma}(\frac{6}{\eps}+8)\right\}
\]
We call a walk $P$ \emph{interface-respecting} if any connection $uz$ of $P$ from a town $T\in\mathcal{U}$ to the outside of~$T$ connects to an interface point: for a connection $uz$ where $u\in T$ and $z\notin T$ it must be that~$z\in I$.

\begin{lemma}\label{lem:HN-implies-I-resp}
Any hub-net-respecting walk is also interface-respecting.
\end{lemma}
\begin{proof}
To show that a hub-net-respecting walk $P^{HN}$ is also interface-respecting, consider the case when $T\in\mathcal{U}$ in \autoref{def:hub-net-resp}, which in particular means that $\ell=i$.
For a connection $ux$ of $P^{HN}$ with $d_X(u,x)\in((1+\frac{3}{4}\eps)r_j,(1+\frac{3}{4}\eps)r_{j+1}]$, $u\in T$, and $x\notin T$, the endpoint $u$ lies in the net $N_k$ with $k=\max\{i,j\}$.
Thus if $j\geq i+\lceil\log_{1+\sigma}(\frac{6}{\eps}+8)\rceil$ then~$u=y_N$, as the nets are hierarchical and $y_N\in N_{i+\lceil\log_{1+\sigma}(\frac{6}{\eps}+8)\rceil}\supseteq N_k$ is the only net point of $N_k$ in the ball $B_X(v,(3+4\eps)r_i)$, which contains~$T$.
However this town is excluded from $\mathcal{U}$, which contradicts $T\in\mathcal{U}$.
Therefore it must be that $j\leq i+\log_{1+\sigma}(\frac{6}{\eps}+8)$, and thus $i\leq k\leq i+\log_{1+\sigma}(\frac{6}{\eps}+8)$.

By \autoref{def:hub-net-resp} we have $x\in H_k$, and we need to show that $x\in I$.
By definition of~$\mathcal{U}$, the town~$T$ intersects the ball~$B_v$, and by \autoref{lem:townproperties} the diameter of $T$ is at most~$r_i$, which means that $u$ is at distance at most $r_i$ from~$B_v$.
Also, by definition of $\sigma$ we have $1+\sigma=1+\frac{\eps}{4+3\eps}=\frac{1+\eps}{1+\frac{3}{4}\eps}$
and thus $(1+\frac{3}{4}\eps)r_{j+1}=(1+\eps)r_j$ as $r_{j+1}=(1+\sigma)^{j+1}$.
Thus we can conclude that the distance of $x$ from $B_v$ is at most $d_X(B_v,u)+d_X(u,x)\leq r_i+(1+\frac{3}{4}\eps)r_{j+1}\leq (2+\eps)r_{k}$ as $k=\max\{i,j\}$.
By definition of the interface, this means that the hub $x\in H_k$ lies in~$I$ given that $i\leq k\leq i+\log_{1+\sigma}(\frac{6}{\eps}+8)$.
\end{proof}

We will compute a hub-net-respecting solution that for each town in $\mathcal{U}$ connects to the interface~$I$ using only two connections.
To achieve this, for each $T\in\mathcal{U}$ we produce one \TSP instance $(X_T=((T\cap K)\cup\{p\}),d_T)$, which is given by the metric induced by the terminals $T\cap K$ of the town in addition to one new vertex~$p$.
The distance from any $t\in T\cap K$ to $p$ is the distance to the interface~$I$.
More concretely, $d_T(s,t)=d_X(s,t)$ for any $s,t\in T\cap K$, and if $\chi_u\in I$ is the closest interface hub to any~$u\in X$ according to the distance function $d_X$, then $d_T(t,p)=d_X(t,\chi_t)$.
Before arguing that $(X_T,d_T)$ is indeed a metric, we should make sure that the distances to $p$ are well-defined.
This follows from the next lemma (given $q\geq 2$).
In fact, as we will need this later on, the lemma also gives an upper bound on the distance from any $u\in T$ to~$\chi_u$.

\begin{lemma}\label{clm:close-interface-points}
If $|\mathcal{U}|\geq 2$ then $I\neq\emptyset$ and the closest interface point $\chi_u\in I$ to any $u\in T$ where $T\in\mathcal{U}$ lies at distance $d_X(u,\chi_u)\leq(3+8\eps)r_i$.
\end{lemma}
\begin{proof}
First we will identify a vertex $z\in B_v$ for which $d_X(u,z)\in(r_i,(3+4\eps)r_i]$.
In case $d_X(u,v)>r_i$, recalling that $T$ (and thus $u$) lies in the ball $B_X(v,(3+4\eps)r_i)$, we may use $z=v$.
Otherwise if $d_X(u,v)\leq r_i$, note that $u$ and $v$ cannot lie in different towns of $\mathcal{U}$, since the distance between any such two towns is more than $r_i$ by \autoref{lem:townproperties}.
Given that $|\mathcal{U}|\geq 2$, this means that there exists a town $T'\in\mathcal{U}$ containing neither $u$ nor $v$.
The town~$T'$ intersects the ball $B_v$ of radius $(2+4\eps)r_i$, and thus we may choose a $z\in T'\cap B_v$ for which the distance from~$u$ is at most $d_X(z,v)+d_X(v,u)\leq (2+4\eps)r_i+r_i=(3+4\eps)r_i$.

Given that $d_X(u,z)\in(r_i,(3+4\eps)r_i]$, there is an index $j$ for which $d_X(u,z)\in(r_j,r_{j+1}]$ where $i\leq j\leq i+\log_{1+\sigma}(3+4\eps)$, as $r_j=(1+\sigma)^j$.
By \autoref{lem:HubHierarchy} there exists a hub $x\in H_j$ on level $j$ such that $d_X(u,x)+d_X(x,z)\leq(1+\frac{3}{2}\eps)d_X(u,z)$.
From \autoref{clm:dist-u-x} we get $d_X(u,x)\leq(1+\frac{3}{4}\eps)d_X(u,z)\leq(1+\frac{3}{4}\eps)r_{j+1}\leq (1+\eps)r_j$ since $1+\sigma=\frac{1+\eps}{1+\frac{3}{4}\eps}$, and analogously $d_X(x,z)\leq (1+\eps)r_j$.
From the latter inequality, as $z\in B_v$ we get $d_X(x,B_v)<(2+\eps)r_j$.
Therefore $x$ is an interface hub in $I$, because $j\leq i+\log_{1+\sigma}(3+4\eps)< i+\log_{1+\sigma}(\frac{6}{\eps}+8)$, using~$\eps\leq\frac{1}{6}$.
In particular $I\neq\emptyset$.
Moreover, $d_X(u,\chi_u)\leq d_X(u,x)\leq (1+\eps)r_j\leq (1+\eps)(3+4\eps)r_i < (3+8\eps)r_i$, using~$\eps\leq\frac{1}{6}$.
\end{proof}
Next, we argue that the distances $d_T$ of the \TSP instance for $T$ define a metric:
first off, for any triple of vertices within $T\cap K$ the triangle inequality holds, since the distances are carried over from the input metric.
Next note that, since the hub sets form a hierarchy, i.e., $H_j\subseteq H_i$ for any~$j\geq i$, we have $I\subseteq H_i$, which implies that $\chi_s,\chi_t\notin T$, since $T\in\mathcal{U}$ is a town on level~$i$.
Therefore, if we consider a triple $s,t,p$, where $s,t\in T\cap K$, then $d_T(s,t)=d_X(s,t)< d_X(s,\chi_s)+d_X(\chi_t,t)=d_T(s,p)+d_T(p,t)$ given that the diameter of a town is less than the distance to any vertex outside the town by \autoref{lem:townproperties}.
Finally, $d_T(s,p)=d_X(s,\chi_s)\leq d_X(s,\chi_t)\leq d_X(s,t)+d_X(t,\chi_t)=d_T(s,t)+d_T(t,p)$, where we first used that $\chi_s$ is the closest interface point to $s$, and then the triangle inequality for~$d_X$. By symmetry we also obtain $d_T(t,p)\leq d_T(t,s)+d_T(s,p)$, which shows that $d_T$ fulfills the triangle inequality.
Now, by \autoref{lem:towns-w-small-dd} and by minimality of the level $i$ for which $\mathcal{U}$ is defined, the \TSP instance $(X_T,d_T)$ has doubling dimension $O\left(\log(q)+\frac{1}{\eps}\log(\frac{h(\eps)}{\eps})\right)$,
since adding one vertex to $T\cap K$ can only increase the number of balls needed to cover any subset by at most 1.
We therefore use \autoref{thm:TSP-dd} to obtain a $(1+\eps)$-approximation~$P_T$ for this instance.
Finally, we convert this solution into a walk in the input metric $X$ connecting to the interface $I$.
For this we may assume w.l.o.g.\
that the \TSP tour $P_T$ visits every vertex of $X_T$ exactly once (it is w.l.o.g.\ since the approximate \TSP tour $P_T$ does not have to respect any structural properties, e.g., net/hub respecting).
In particular, it visits~$p$ only once and therefore contains exactly two connections $sp$ and $tp$ to $p$ from some $s,t\in T\cap K$.
We then obtain a (not necessarily closed) walk $W_T$ from $P_T$ by starting in $\chi_s$, going to $s$ and following $P_T$ from $s$ to $t$ inside $T\cap K$, and then going from $t$ to~$\chi_t$.
Note that $W_T$ connects to the interface using exactly two connections, has the same length as $P_T$ due to the distance function $d_T$, and visits all terminals of $T$.

\subsection{The conquer step}

In order to obtain a solution to the input instance, apart from the solutions to the towns of $\mathcal{U}$ we also need a solution for the remaining instance.
In particular, the remaining instance is still given by the metric~$X$, but we reduce the terminals to $K'=\{t_{\mathcal{U}}\}\cup K\setminus (\bigcup_{T\in\mathcal{U}} T)$, where $t_{\mathcal{U}}\in K\cap \bigcup_{T\in\mathcal{U}} T$ is an arbitrary fixed terminal
among those in the towns of $\mathcal{U}$.
Note that the highway dimension of this instance remains the same since we did not change $X$, but now the ball~$B_v$ intersects at most one town of~$\mathcal{T}_i$ containing terminals, namely the one containing~$t_{\mathcal{U}}$.
Thus this instance can be solved recursively.
Let~$P^{HN}_{K'}$ denote the recursively computed hub-net-respecting $(1+O(\eps))$-approximation for this instance.

We now show how to stitch these walks together in order to obtain a solution to the original instance. This lemma uses the well-known patching technique, which involves adding a minimum spanning tree $\textup{MST}(I)$ on the interface~$I$.

\begin{lemma}\label{lem:approx-sol}
Given a walk $W_T$ for each $T\in\mathcal{U}$ and a hub-net-respecting solution $P^{HN}_{K'}$ for~$K'$, in polynomial time we can construct a hub-net-respecting Subset \TSP solution $P^{HN}$ for $K$ of cost $w(P^{HN})\leq w(P^{HN}_{K'})+(1+77\eps)(\sum_{T\in\mathcal{U}}w(W_T)+2\cdot w(\textup{MST}(I)))$.
\end{lemma}
\begin{proof}
We claim that taking the union of the minimum spanning tree on the interface $I$ and all walks $W_T$ for towns $T\in\mathcal{U}$, contains a Subset \TSP solution for the terminal set $\bigcup_{T\in\mathcal{U}}T\cap K$ lying in these towns of $\mathcal{U}$, where each edge of $\textup{MST}(I)$ is used at most twice.
We prove a slightly more general statement, which we will reuse later on.

\begin{claim}\label{clm:patching}
Given a set of walks $W_1,\ldots,W_k$ with endpoints in $I$, there is a closed walk which traverses $W_1,\ldots,W_k$ and visits all of $I$, of total weight at most $\sum_{j=1}^k w(W_j)+2\cdot w(\textup{MST}(I))$.
\end{claim}
\begin{subproof}
    For each walk $W_j$, let $x_j,y_j\in I$ be its endpoints lying in the interface. We construct an auxiliary multi-graph $H$ over $I$ as follows: (1) First for every walk $W_j$ add an edge $(x_j,y_j)$ between its endpoints to $H$. (2) Add the edges of $\MST(I)$ to $H$. (3) Let $A\subseteq I$ be all the interface vertices that currently have odd degree (note that $|A|$ has to be even), add a minimum weight matching $M(A)$ between the endpoints of $A$.
    This finishes the construction of $H$. We will call the edges in $M(A)\cup \MST(I)$ black, while the edges between the endpoints of the walks will be called red. Note that the weight of the edges in $M(A)$ is at most $\MST(I)$. One can see this be recalling the proof of Christofides \cite{Chr76}: $w(M(A))\le \frac12\cdot w(\textup{TSP}(I))\le w(\MST(I))$. That is, the weight of a perfect matching is at most half the weight of the optimal \TSP tour over $I$, which by itself is at most twice the weight of $\MST(I)$. We conclude that the total weight of the black edges is $2\cdot w(\MST(I))$.
    All the degrees in $H$ are even, and hence there is an Eulerian tour of $H$. We will construct a walk over the interface points by following this Eulerian tour and replacing each red edge $(x_i,y_i)$ by the corresponding walk $W_i$. The claim now follows.
\end{subproof}

By \autoref{clm:patching}, the union of $\textup{MST}(I)$ and the walks $W_T$ for towns $T\in\mathcal{U}$ contains a Subset \TSP solution $P'$ of cost at most $\sum_{T\in\mathcal{U}}w(W_T)+2\cdot w(\textup{MST}(I))$ for the terminal set contained in the towns of $\mathcal{U}$.
Using \autoref{lem:hub-net-resp}, we can make $P'$ hub-net-respecting so that its cost is at most $(1+77\eps)(\sum_{T\in\mathcal{U}}w(W_T)+2\cdot w(\textup{MST}(I)))$.
Note that $P'$ also visits the terminal~$t_{\mathcal{U}}$, which is contained in the set $K'$, since~$t_{\mathcal{U}}$ lies in some town of $\mathcal{U}$.
Hence the union of~$P'$ and the closed walk $P^{HN}_{K'}$ contains a closed walk on all terminals in $K$, and its cost is bounded as desired.
Furthermore, the union of two hub-net-respecting solutions is again hub-net-respecting, since all connections fulfill the conditions of \autoref{def:hub-net-resp}.
\end{proof}

It remains to show that the obtained solution is near-optimal and to bound the runtime.
To show the former, the main idea is to charge the extra cost incurred by the minimum spanning tree on the interface to the cost of connecting all towns of $\mathcal{U}$ to the interface in the optimum solution.
Accordingly, we first bound the cost of a minimum spanning tree on the interface.

\begin{lemma}\label{lem:cost-of-MST}
The cost of a minimum spanning tree on the interface $I$ is \\$w(\textup{MST}(I))\leq O((\frac{1}{\eps^2}\log(\frac{1}{\eps})\cdot h(\eps))^2)\cdot r_{i}$
\end{lemma}
\begin{proof}
We construct a spanning tree on $I$ by simply forming a star with leaf set $I$ and center point~$v$, which is also the center of the ball $B_v$.
Each interface point in~$I$ belongs to a hub set $H_j$ at distance at most $(2+\eps)r_j$ from $B_v$, and the radius of $B_v$ is~$(2+4\eps)r_i$.
This means that for some $i\leq j\leq i+\log_{1+\sigma}(\frac{6}{\eps}+8)$ every edge of the star has length at most $(2+\eps)r_j+(2+4\eps)r_i\leq 5r_j$, given~$\eps\leq \frac{1}{6}$.
By \autoref{lem:HubHierarchy}, there are $O((\frac{1}{\eps}\log(\frac{1}{\eps})\cdot h(\eps))^2)$ hubs of $H_j$ in $I$, as $j\geq i$.
Hence, as $r_j=(1+\sigma)^j$, the cost of the star is at most
\begin{align*}
	\sum_{j=i}^{i+\lfloor\log_{1+\sigma}(\frac{6}{\eps}+8)\rfloor}{\textstyle O((\frac{1}{\eps}\log(\frac{1}{\eps})\cdot h(\eps))^2)\cdot5r_{j}}
        & ={\textstyle O((\frac{1}{\eps}\log(\frac{1}{\eps})\cdot h(\eps))^2)\cdot r_{i}\cdot}\sum_{j=0}^{\lfloor\log_{1+\sigma}(\frac{6}{\eps}+8)\rfloor}(1+\sigma)^{j}\\
	& ={\textstyle O((\frac{1}{\eps}\log(\frac{1}{\eps})\cdot h(\eps))^2)\cdot r_{i}\cdot O(\frac{1}{\eps^2})}\\
	& ={\textstyle O((\frac{1}{\eps^2}\log(\frac{1}{\eps})\cdot h(\eps))^2)\cdot r_{i}}~.
\end{align*}
The weight of $\textup{MST}(I)$ is at most twice the weight of the minimum Steiner tree on $I$, which by itself is bounded by the weight of the star. The lemma follows.
\end{proof}

Given \autoref{lem:cost-of-MST}, we can charge the cost of a minimum spanning tree on the interface to connection costs for the towns in $\mathcal{U}$.
This is made possible by a large enough lower bound $q$ on the number of towns in~$\mathcal{U}$, and the distance lower bound from any town to the outside given by \autoref{lem:townproperties}.

\begin{lemma}\label{lem:cost-of-sol}
Let $O^{HN}_{K'}$ and $O^{HN}$ be the hub-net-respecting solutions of minimum cost for $K'$ and~$K$, respectively.
Given a hub-net-respecting solution $P^{HN}_{K'}$ for the terminal set $K'$ with $w(P^{HN}_{K'})\leq (1+92\eps) w(O^{HN}_{K'})$, the hub-net-respecting solution $P^{HN}$ given by \autoref{lem:approx-sol} for $K$ has cost $w(P^{HN})\leq (1+92\eps)w(O^{HN})$, if $|\mathcal{U}|\geq q=\frac{c}{\eps^5}\log^2_2(\frac{1}{\eps})\cdot h(\eps)^2$ for a sufficiently large constant~$c$.
\end{lemma}
\begin{proof}

We begin by comparing the cost of the computed walk $W_T$ for a fixed town $T\in\mathcal{U}$ with the corresponding sub-walks of the optimum $O^{HN}$ visiting $T$.
More concretely, let $W_1,\ldots,W_k$ be all the sub-walks contained in $O^{HN}$ that begin and end at an interface point of $I$ but all internal vertices are not in~$I$.
By \autoref{lem:HN-implies-I-resp} the hub-net-respecting walk $O^{HN}$ is also interface-respecting, and so the internal vertices of each such walk~$W_j$ intersecting $T$ must all lie in $T$.
We denote by~$U_T$ the union of all the sub-walks $W_j$ intersecting $T$.
For the computed walk $W_T$, recall that we obtained $W_T$ from a tour $P_T$ for the \TSP instance $X_T$, where $P_T$ is a $(1+\eps)$-approximation for $X_T$, and $w(W_T)=w(P_T)$.
We may convert the union~$U_T$ into a solution $P'_T$ to $X_T$ by first short-cutting to the terminal set $T\cap K$ and then replacing all connections to the interface by connections to the additional vertex $p$ of $X_T$.
By definition of the distance function $d_T$, which reflects the distance to the closest interface points, the length of $P'_T$ is at most the total length of $U_T$.
Therefore we have
\[
w(W_T)=w(P_T)\leq (1+\eps)w(P'_T)\leq (1+\eps)w(U_T).
\]

The next step is to consider the hub-net-respecting solution $P^{HN}_{K'}$ on the terminal set $K'=\{t_{\mathcal{U}}\}\cup (K\setminus(\bigcup_{T\in\mathcal{U}}T))$ for the fixed terminal $t_{\mathcal{U}}$ lying in some town of $\mathcal{U}$.
We would like to compare the cost of~$P^{HN}_{K'}$ with the part of the solution $O^{HN}$ that connects all terminals in~$K'\setminus\{t_{\mathcal{U}}\}$, which are not in towns of $\mathcal{U}$.
That is, consider all sub-walks of $O^{HN}$ connecting to the interface but not visiting any town of $\mathcal{U}$:
formally, these are the sub-walks of~$O^{HN}$ not contained in any~$U_T$, that begin and end at an interface point but all internal vertices are not in~$I$.
By patching the interface using \autoref{clm:patching}, these sub-walks give rise to a closed walk~$P$ visiting all terminals in $K'\setminus\{t_{\mathcal{U}}\}$ of cost at most $w(O^{HN})-\sum_{T\in\mathcal{U}}w(U_T)+2\cdot w(\textup{MST}(I))$.
We now extend this closed walk to also visit~$t_{\mathcal{U}}$ by simply adding a connection from $t_{\mathcal{U}}$ to its closest interface point~$\chi_{t_{\mathcal{U}}}$ and back.
By \autoref{clm:close-interface-points}, the distance between $\chi_{t_{\mathcal{U}}}\in I$ and $t_{\mathcal{U}}$ is at most~$(3+8\eps)r_i$.
Since~$P$ visits every interface point and thus also $\chi_{t_{\mathcal{U}}}$, the resulting closed walk $P_{K'}$ visits all terminals of~$K'$ and has cost at most $w(O^{HN})-\sum_{T\in\mathcal{U}}w(U_T)+2\cdot w(\textup{MST}(I))+2(3+8\eps)r_i$.

Note that $P_{K'}$ is not necessarily hub-net-respecting, which we would need to compare its cost with (parts of) the hub-net-respecting solution~$O^{HN}$.
However, by construction the only connections of~$P_{K'}$ that do not exist in $O^{HN}$ have both endpoints in~$I\cup\{{t_\mathcal{U}}\}$.
Therefore, to obtain a hub-net-respecting closed walk from~$P_{K'}$, we may apply \autoref{lem:hub-net-resp} to each sub-walk of~$P_{K'}$ of maximal length that only uses nodes of $I\cup\{{t_\mathcal{U}}\}$.
Observe that these sub-walks contribute a cost of $2\cdot w(\textup{MST}(I))+2(3+8\eps)r_i$ to $P_{K'}$, and thus the resulting closed walk has cost at most $w(O^{HN})-\sum_{T\in\mathcal{U}}w(U_T)+2(1+77\eps)\cdot (w(\textup{MST}(I))+(3+8\eps)r_i)$.
This means that the cost of the cheapest hub-net-respecting closed walk $O^{HN}_{K'}$ for $K'$ can be bounded by this term.
Since we are assuming that the solution $P^{HN}_{K'}$ costs at most $1+91\eps$ times the cost of $O^{HN}_{K'}$, using $\eps\leq\frac{1}{6}$ we get
\begin{align*}
w(P^{HN}_{K'}) &\leq (1+92\eps)\left(w(O^{HN})-\sum_{T\in\mathcal{U}}w(U_T)+2(1+77\eps)\cdot (w(\textup{MST}(I))+(3+8\eps)r_i)\right)\\
&\leq (1+92\eps)\left(w(O^{HN})-\sum_{T\in\mathcal{U}}w(U_T)\right)+O\big(w(\textup{MST}(I))+r_i\big).
\end{align*}

To obtain a lower bound on the optimum solution, note that $\sum_{T\in\mathcal{U}}w(U_T)\geq 2qr_i$, given that~$|\mathcal{U}|\geq q$ and each $U_T$ contains at least two connections to the interface, each of length at least $r_i$ by \autoref{lem:townproperties}.
Thus given the bounds obtained above in addition to those of Lemmas~\ref{lem:approx-sol} and \ref{lem:cost-of-MST} on the costs of the computed hub-net-respecting solution $P^{HN}$ and of~$\textup{MST}(I)$, we can bound the cost of the computed solution $P^{HN}$ as follows, crucially using $(1+77\eps)(1+\eps)-(1+92\eps)\leq-\eps$ as $\eps\leq\frac{1}{6}$:
\begin{align*}
w(P^{HN}) &\leq w(P^{HN}_{K'})+(1+77\eps)\left(\sum_{T\in\mathcal{U}}w(W_T)+2\cdot w(\textup{MST}(I))\right) \\
&\leq
(1+92\eps)\left(w(O^{HN})-\sum_{T\in\mathcal{U}}w(U_T)\right) + O\big(w(\textup{MST}(I))+r_i\big)\\
&\qquad +(1+77\eps)\left(\sum_{T\in\mathcal{U}}w(W_T)+2\cdot w(\textup{MST}(I))\right)\\
&\leq (1+92\eps)w(O^{HN}) + \big((1+77\eps)(1+\eps)-(1+92\eps)\big)\sum_{T\in\mathcal{U}}w(U_T)
+O\big(w(\textup{MST}(I))+r_i\big) \\
&\leq (1+92\eps) w(O^{HN}) -\eps\cdot \sum_{T\in\mathcal{U}}w(U_T) +O\big(w(\textup{MST}(I))+r_i\big)\\
&\leq {\textstyle (1+92\eps) w(O^{HN}) - 2\eps q r_i +O\big(w(\textup{MST}(I))+r_i\big)} \\
&\leq {\textstyle (1+92\eps) w(O^{HN})+O\left((\frac{1}{\eps^2}\log(\frac{1}{\eps})\cdot h(\eps))^2-\eps q\right)r_i} \\
&\leq (1+92\eps) w(O^{HN}),
\end{align*}
given that $q=\frac{c}{\eps^5}\log^2_2(\frac{1}{\eps})\cdot h(\eps)^2$ for a sufficiently large constant~$c$.
\end{proof}

We are now ready to prove the main result, i.e., that the algorithm above is an approximation scheme to the Subset \TSP problem on inputs of bounded highway dimension. We restate \autoref{thm:TSPHighway} for convenience.
\TSPthm*

\begin{proof}
If there is no ball intersecting many towns containing terminals, then by \autoref{lem:towns-w-small-dd} the metric induced by the terminals has doubling dimension $O\left(\log(q)+\frac{1}{\eps}\log(\frac{h(\eps)}{\eps})\right)$.
We may then use \autoref{thm:TSP-dd} to compute a $(1+\eps)$-approximation on the metric given by the terminal set.

Otherwise we recursively compute a solution for $K'$ and stitch it together with the solutions computed for all towns in $\mathcal{U}$.
The solution is then given by \autoref{lem:approx-sol} and, due to \autoref{lem:cost-of-sol},
has cost at most $(1+92\eps)w(O^{HN})$ for the cheapest hub-net-respecting solution $O^{HN}$.
For the termination of the recursion, if the metric given by $K'$ has bounded doubling dimension, then we use \autoref{lem:towns-w-small-dd} to obtain a $(1+\eps)$-approximation, which however is not necessarily hub-net-respecting.
We therefore apply \autoref{lem:hub-net-resp} to this solution to obtain a hub-net-respecting solution~$P^{HN}_{K'}$ of cost at most $(1+77\eps)(1+\eps)w(O_{K'})$, where $O_{K'}$ is the optimum solution for $K'$.
Obviously, $w(O_{K'})\leq w(O^{HN}_{K'})$ for the cheapest hub-net-respecting solution $O^{HN}_{K'}$ for $K'$, and thus $w(P^{HN}_{K'})\leq (1+77\eps)(1+\eps)w(O^{HN}_{K'})\leq (1+92\eps)w(O^{HN}_{K'})$, as $\eps\leq\frac{1}{6}$.
Hence the induction hypothesis is given in order to compute a solution of cost at most $(1+92\eps)w(O^{HN})$ according to \autoref{lem:cost-of-sol}.

To compare the cost of the cheapest hub-net-respecting solution $O^{HN}$ to the optimum Subset \TSP solution $O^\star$, we may apply \autoref{lem:hub-net-resp} to conclude $w(O^{HN})\leq (1+77\eps)w(O^\star)$.
Thus our algorithm computes a solution of cost at most $(1+92\eps)(1+77\eps)w(O^\star)\leq (1+1350\eps)w(O^\star)$.
Thus setting $\eps'=\eps/1350$, the algorithm computes a $(1+\eps')$-approximation.

To bound the runtime, we apply the algorithm of \autoref{thm:TSP-dd} sequentially to each town of $\mathcal{U}$ in each recursive call.
By \autoref{lem:cost-of-sol}, asymptotically the value of $q$ is $O(\frac{1}{\eps^5}\log^2(\frac{1}{\eps})\cdot h(\eps)^2)$.
    Thus by \autoref{lem:towns-w-small-dd} the doubling dimension of the towns in $\mathcal{U}$ is bounded by $O(\frac{1}{\eps}\log(\frac{h(\eps)}{\eps}))$.
The resulting runtime according to \autoref{thm:TSP-dd} is therefore $2^{2^{O\left(\frac{1}{\eps}\log^2\frac{h(\eps)}{\eps}\right)}}n^{O(1)}$, which concludes the proof.
\end{proof}

\section{Metric Toolkit}\label{sec:MetricToolkit}
	In this paper we attempt to systematically study metric spaces with low highway dimension w.r.t. the new proposed \autoref{def:HD}.
    We will introduce several of the most important tools in the algorithm designer toolkit in this context. We will then mention a sample of applications achievable with this new toolkit. However, as the list of applications is very long, we will not attempt to mention all of them. We are sure that more applications will be found in future work.

	Padded decompositions and sparse covers are very basic tools in algorithmic design for metric problems. The padding parameter of a metric space is a crucial parameter of the space.
	Consider a \emph{partition} $\mathcal{P}$ of $V$ into disjoint clusters.
    The weak diameter of a cluster $C\in\cP$ is $\max_{u,v\in C}d_G(u,v)$ the maximum pairwise distance (w.r.t. the shortest path distance in the original graph).
    The strong diameter of a cluster $C\in\cP$ is $\max_{u,v\in C}d_{G[C]}(u,v)$ the maximum pairwise distance in the induced graph. Note that in metric spaces (a.k.a complete graphs with triangle inequality on the edge weights) the notions of strong and weak diameter are equivalent.
	For a vertex $v\in V$, we denote by $P(v)$ the cluster $P\in \mathcal{P}$ that contains $v$.
	A partition $\mathcal{P}$ is strongly $\Delta$-\emph{bounded} (resp.\ weakly $\Delta$-bounded) if the strong-diameter (resp.\ weak-diameter) of every $P\in\mathcal{P}$ is bounded by $\Delta$.
	If the ball $B_G(v,\gamma\Delta)$ of radius $\gamma\Delta$ around a vertex $v$ is fully contained in $P(v)$, we say that $v$ is $\gamma$-\emph{padded} by $\mathcal{P}$. Otherwise, if $B_G(v,\gamma\Delta)\not\subseteq P(v)$, we say that the ball is \emph{cut} by the partition.
	\begin{definition}[Padded Decomposition]\label{def:PadDecompostion}
		A distribution $\mathcal{D}$ over partitions of a graph $G=\left(V,E,w\right)$ is a strong (resp.\ weak) $(\beta,\delta,\Delta)$-padded decomposition if every $\mathcal{P}\in\supp(\mathcal{D})$ is strongly (resp.\ weakly) $\Delta$-bounded and for any $0\le\gamma\le\delta$, and $z\in V$,
		$\Pr[B_G(z,\gamma\Delta)\subseteq P(z)] \ge e^{-\beta\gamma}$.
		We say that $G$ admits a strong/weak $(\beta,\delta)$-padded decomposition scheme if for every $\Delta>0$, it admits a strong/weak $(\beta,\delta,\Delta)$-padded decomposition. Here $\beta$ is called the padding parameter.
	\end{definition}

    It is known that general $n$-point graphs (or metrics) admit strong $(O(\log n),\Omega(1))$-padded decomposition schemes~\cite{Bar96}, graphs whose shortest path metric has doubling dimension $d$ admit strong $(O(d),\Omega(1))$-padded decomposition scheme~\cite{GKL03,ABN11,Fil19padded}, $K_r$-minor free graphs admit strong $(O(\log r),\Omega(1))$-padded decomposition schemes~\cite{KPR93,FT03,AGGNT19,Fil19padded,FFIKLMZ24,CF25}, or strong $(O(r),\Omega(1))$-padded decomposition schemes~\cite{Fil19padded}.
Here we show that metric spaces with highway dimension $h:\R_{\ge0}\rightarrow\N$ have padding parameter~$O(\ln h(\eps))$. Note that by \autoref{obs:doublingToHighway}, \autoref{thm:padded} below generalizes a previous result for doubling metrics~\cite{Fil19padded}.

 \begin{restatable}[]{theorem}{PaddedDecomp}
\label{thm:padded}\RestateRemark
Let $G=(V,E,w)$ be a weighted graph with highway dimension $h:\mathbb{R}_{\geq 0}\to\mathbb{N}\cup\{\infty\}$, and fix~$\eps\in[0,\frac{1}{4}]$.
Then $G$ admits a strong $\left(O(\ln h(\eps)),\Omega(1)\right)$-padded decomposition scheme.
\end{restatable}

	A dual notion to padded decompositions is a \emph{sparse cover}. Here we allow the clusters to intersect, but every point should be padded somewhere.
	\begin{definition}[Sparse Cover]\label{def:SparseCover}
		Given a weighted graph $G=(V,E,w)$, a collection of clusters $\mathcal{C} = \{C_1,..., C_t\}$ is called a weak/strong $(\beta,s,\Delta)$ sparse cover if the following conditions hold.
		\begin{enumerate}
			\item Bounded diameter: The weak/strong diameter of every cluster $C_i\in\mathcal{C}$ is bounded by $\Delta$.\label{condition:RadiusBlowUp}
			\item Padding: For each $v\in V$, there exists a cluster $C_i\in\mathcal{C}$ such that $B_G(v,\frac\Delta\beta)\subseteq C_i$.
			\item Sparsity: For each $v\in V$, there are at most $s$ clusters in $\mathcal{C}$ containing $v$.
		\end{enumerate}
		If the clusters $\cC$ can be partitioned into $s$ partitions $\cP_1,\dots,\cP_s$ such that $\cC=\bigcup_{i=1}^s\cP_i$, then  $\{\cP_1,\dots,\cP_s\}$ is called a weak/strong $(\beta,s,\Delta)$ sparse partition cover.
		We say that a graph $G$ admits a weak/strong $(\beta,s)$ sparse (partition) cover scheme (abbr.~\SPCS), if for every parameter $\Delta>0$ it admits a weak/strong $(\beta,s,\Delta)$ sparse (partition) cover.
	\end{definition}
	The notion of sparse cover scheme is the more common in the literature. However, an \SPCS provides additional structure that is crucial for different applications.\footnote{In this paper we require the partition property of \SPCS for metric embeddings into $\ell_\infty$, and for the oblivious buy-at-bulk. This property is also crucial in the construction of ultrametric covers~\cite{FL22,Fil23,FGN24}.}
	Obtaining a strong diameter guarantee is also considerably more challenging, however it is frequently required in \enquote{subgraph based} applications.
It is known that for every integer $k\ge 1$, every $n$-point graph (or metric) admits a strong $(4k-2,2k\cdot n^{\frac1k})$-\SPCS~\cite{AP90}.  Graphs whose shortest path metric have doubling dimension $d$ admit strong $(O(d),\tilde{O}(d))$-\SPCS \cite{Fil19padded}, and every $K_r$-minor free graph admits a weak $(8+\eps,O(r^4/\eps^2))$-\SPCS \cite{CF25}, or strong $(4+\eps,O(\frac1\eps)^r)$-\SPCS
\cite{Fil24} (see also~\cite{KPR93,KLMN04}).
Here we show that metric spaces with small highway dimension admit good sparse covers, and \SPCS.

\begin{restatable}[]{theorem}{SparseCover}
\label{thm:SparseCover}\RestateRemark
Consider a weighted graph $G=(V,E,w)$ with highway dimension $h:\mathbb{R}_{\geq 0}\to\mathbb{N}\cup\{\infty\}$. For every $\eps\in(0,\frac{1}{10}]$, $G$ admits a strong $(8,2\cdot h(\eps)^2+1)$ sparse cover scheme.
\end{restatable}

\begin{restatable}[]{theorem}{SparsePartitionCover}
\label{thm:SparsePartitionCover}\RestateRemark
 Consider a weighted graph $G=(V,E,w)$ with highway dimension $h:\mathbb{R}_{\geq 0}\to\mathbb{N}\cup\{\infty\}$. For every $\eps>0$, $G$ admits a strong $(\frac{4\cdot(1+\eps)}{\eps},3\cdot h(\eps)^2+1)$ sparse partition cover scheme.
\end{restatable}

Recently, Conroy and Filtser \cite{CF25} proved that give a $(\beta,s,\Delta)$-sparse cover, one can construct a $(O(\beta\log s),\frac{1}{4\beta},\Delta)$-padded decomposition. However, the diameter guarantee in the padded-decomposition  obtained from \cite{CF25} is weak (regardless of the diameter guarantee of the given sparse cover).
Thus our \autoref{thm:SparseCover} implies a weak diameter version of our \autoref{thm:padded}. We prove \autoref{thm:padded} directly to obtain a strong diameter gurantee.

Given a metric space $(X,d_X)$, we say that another metric space $(X,d_H)$ over the same point set $X$ dominates  $(X,d_X)$  if  $\forall u,v\in X$, $d_X(u,v)\le d_H(u,v)$.
We say that $T=(V,E)$ is a dominating tree over $X$ if $T$ is a tree (that is, connected and acyclic), and the points of $X$ can be mapped by some function $f$ into $V$ such that the shortest path distance between the endpoints is dominating. That is, $\forall u,v\in X$, $d_X(u,v)\le d_T(f(u),f(v))$ where $d_T$ is the shortest-path distance function of the tree~$T$.
Note that the tree $T$ is allowed to contain points not in $f(X)$ (these are usually called Steiner points). When the mapping $f$ is clear, we can abuse notation and write $v$ instead of $f(v)$.
A \emph{tree cover} is a small collection of dominating trees such that every pairwise distance is closely approximated by one of the trees in the collection. Formally:
\begin{definition}[Tree cover]\label{def:tree-cover}
	A \emph{$(\tau,\rho)$-tree cover} for a metric space $(X,d)$ is a collection of at most~$\tau$ dominating trees $\cT = \{(T_i,d_{T_i})\}_{i=1}^{\tau}$ over $X$, such that for every $x,y\in X$ there is a tree $T_i$ for which $d_{T_i}(x,y)\le \rho\cdot d_X(x,y)$.

\end{definition}

For general metric spaces it is known that for every integer $k\ge 1$, every $n$ point metric space admits a $(O(n^{\frac1k}\cdot k),O(k))$-tree cover, see~\cite{MN07,NT12}.\footnote{In fact these trees have a stronger Ramsey property, where every point $x\in X$ has a home tree $T_i$ that approximates all the distances to $x$ up to a $O(k)$-factor. See also~\cite{BFM86,BBM06,BLMN05,BGS16,ACEFN20,Bar21,FL21}.}
It is also known that planar graphs admit $(1+\eps,\eps^{-3}\cdot\log(\frac1\eps))$-tree covers~\cite{CCLMST23Planar}, minor free graphs admit $(1+\eps,2^{r^{O(r)}/\eps})$-tree covers~\cite{CCLMST24}, finite subsets of Euclidean $d$-dimensional spaces admit $\left(1+\eps,O(\eps^{-\frac{d+1}{2}}\cdot\log(\frac1\eps))\right)$-tree covers (for constant $d$)~\cite{ADMSS95,CCLMST24Euclidean}, and most relevant to our paper, every finite metric space with doubling dimension $d$ admits a $(1+\eps,\eps^{-O(d)})$-tree cover~\cite{BFN22}.
See also~\cite{FL22,Fil23,FGN24} for the closely related notion of ultrametric/HST cover.

Here, we follow the construction of Bartal, Fandina and Neiman~\cite{BFN22} for doubling metrics and construct a tree cover for graphs with bounded highway dimension (where the hubs replace the net points).
\begin{restatable}[]{theorem}{TreeCover}
\label{thm:HighwayTreeCover}\RestateRemark
 Consider a weighted graph $G=(V,E,w)$ with highway dimension $h:\mathbb{R}_{\geq 0}\to\mathbb{N}\cup\{\infty\}$. For every $\eps\in(0,1]$, $G$ admits a $(1+2\eps,h(\eps)\cdot O(\eps^{-1}\log\frac1\eps))$-tree cover.
\end{restatable}

\subsection{Sparse Cover}
	Consider a weighted graph $G=(V,E,w)$ with highway dimension $h:\mathbb{R}_{\geq 0}\to\mathbb{N}\cup\{\infty\}$. Fix $r\in\mathbb{R}_{\geq 0}$, $\eps\in[0,\frac{1}{10}]$, and $\alpha=2.8+6\eps$.
	Let $\SPC$ be an $(r,\eps)$-shortest path cover and note that by \autoref{lem:hub_bound_2.8Ball}, every ball of radius $\alpha\cdot r$ contains at most $2\cdot h(\eps)^2$ hubs from $\SPC$.
    In addition, let $\cT$ be the set of towns for \SPC.
	We will construct a sparse cover as follows. Let
	$$\cC_1=\left\{B_{G}(x,\alpha\cdot r)\right\}_{x\in\SPC}$$
	be the set of balls of radius $\alpha\cdot r$ around the hubs \SPC.
	Our sparse cover will be $\cC=\cC_1\cup\cT$, i.e., the union of the towns and the balls around hubs.

	Clearly, every cluster $C\in\cC$ has strong diameter at most $2\alpha\cdot r$ (recall that the diameter of each town is at most $r$ by \autoref{lem:townproperties}).
	Set $\beta=0.8+2\eps\le1$. We claim that for every vertex $v\in V$, the ball $B_{G}(v,\beta\cdot r)$ is fully contained in some cluster.
	Indeed, if $v\in T$ for some town $T\in \cT$, then by \autoref{lem:townproperties} $T$ has diameter at most $r$, and has distance greater than $r$ to the rest of the graph ($d_G(T,V\setminus T)>r$), thus  $B_{G}(v,r)= T$.
	Next, for a vertex $v\in\cS$ in the sprawl,
    by \autoref{lem:townproperties} there is a hub $x\in\SPC$ at distance at most $(2+\eps)\cdot r$. Using the triangle inequality, $B_{G}(v,\beta\cdot r)\subseteq B_{G}(x,(2+\eps+\beta)\cdot r)\subseteq B_{G}(x,\alpha\cdot r)\in\cC$, as required.
	Thus we indeed get a sparse cover with padding parameter $\frac{2\alpha}{\beta}=\frac{5.6+12\eps}{0.8+2\eps}<8$.
	As all the towns are disjoint, and every vertex has at most $2\cdot h(\eps)^2$ hubs at distance $\alpha\cdot r$, it follows that the sparsity of our cover is at most $2\cdot h(\eps)^2+1$.
	As $r$ was chosen arbitrarily, we conclude:

\SparseCover*

\subsubsection{Sparse Partition Cover}
In this subsection we construct sparse partition cover, i.e., a collection of partitions that together constitute a sparse cover.
Unfortunately, it is not clear if the sparse cover from \autoref{thm:SparseCover} can be turned into a sparse partition cover.
Therefore, here, instead of \autoref{lem:hub_bound_2.8Ball}, we will be using \autoref{lem:hub_bound}. We argue the following:
\SparsePartitionCover*
\begin{proof}
    Let $\SPC$ be an $(r,\eps)$-shortest path cover,
    and let $\cT$ be the towns for \SPC. The first partition in our sparse partition cover will simply be the towns $\cT$ (and for every sprawl vertex we will simply add a singleton cluster).
    In addition, for every hub $x\in\SPC$ let $C_x=B_{G}(B_{G}(x,(2+\eps)\cdot r),\eps\cdot r)$ be all the points at distance at most $\eps\cdot r$ from the ball $B_{G}(x,(2+\eps)\cdot r)$.
    We will add $\{C_x\mid x\in\SPC\}$ to our cluster set.

    First note that for every vertex $u\in\cS$ in the sprawl there is a hub $x\in\SPC$ at distance at most $(2+\eps)r$ from $u$ by \autoref{lem:townproperties}. That is $u\in B_{G}(x,(2+\eps)\cdot r)$, which by triangle inequality implies $B_{G}(u,\eps\cdot r)\subseteq B_{G}(B_{G}(x,(2+\eps)\cdot r),\eps\cdot r)$, and thus all the vertices are padded.

    Next, fix a hub $x\in \SPC$, and let $W_x=\SPC\cap B_{G}\left(B_{G}(x,(2+4\eps)\cdot r),(2+\eps)\cdot r\right)$ be the set of hubs at distance at most $(2+\eps)\cdot r$ from the ball $B_{G}(x,(2+4\eps)r)$.
    Note that by \autoref{lem:hub_bound}, $|W_x|\le 3\cdot h(\eps)^2$.
    We argue that if the cluster $C_x$ intersects with another cluster $C_y$ centered at a hub $y\in\SPC$, then $y\in W_x$.
    Indeed, suppose that the intersection  $C_x\cap C_y\ne\emptyset$ is not empty, and let $u\in C_x\cap C_y$ be a vertex in this intersection. By the definition of the clusters, it follows that there are vertices $u_x\in B_G(u,\eps\cdot r)\cap B_G(x,(2+\eps)\cdot r)$, and $u_y\in B_G(u,\eps\cdot r)\cap B_G(y,(2+\eps)\cdot r)$. By the triangle inequality, it holds that
    $$d_{G}(x,u_{y})\le d_{G}(x,u_{x})+d_{G}(u_{x},u)+d_{G}(u,y_{y})\le(2+\eps+\eps+\eps)\cdot r<(2+4\eps)\cdot r~.$$
    As $d_G(y,u_y)\le(2+\eps)\cdot r$, it follows that $y\in W_x$.
    We conclude that excluding the town clusters, every cluster can intersect at most $3\cdot h(\eps)^2$ other clusters (including itself).
    It follows that we can simply partition (in a greedy manner) the set $\{C_x\mid x\in\SPC\}$ of clusters into $3\cdot h(\eps)^2$ subsets, where no two clusters intersect.

    To conclude, note that the diameter of all the clusters in $\cT$ is bounded by $r$, and the diameter of the other clusters is bounded by $2\cdot(2+\eps+\eps)=4\cdot(1+\eps)$. The theorem now follows.
\end{proof}

	\subsection{Padded Decomposition}\label{sec:PaddedDecomp}
    Filtser~\cite{Fil19padded} proved that if a graph has a \enquote{sparse net}, then it admits a strong padded decomposition.
    Specifically, if there is a set $N\subseteq V$ of center vertices such that every vertex $v\in V$ has a center $x\in N$ at distance at most $d_G(v,x)\le\Delta$, and at most $|B_G(v,3\Delta)\cap N|\le\tau$ centers at distance $3\Delta$, then the graph admits a strong $(O(\log \tau),\frac{1}{16},4\Delta)$ padded decomposition.
	This meta theorem from sparse nets to padded decompositions almost fits our case of highway dimension graphs. Indeed, consider the set of hubs \SPC, then every vertex has at most $2\cdot h(\eps)^2$ hubs at distance $(2.8+6\eps)r$ by \autoref{lem:hub_bound_2.8Ball}, and every vertex from the sprawl has a hub at distance at most $(2+\eps)r$ by \autoref{lem:townproperties}. There are two issues in applying the meta theorem:
	\begin{enumerate}
		\item vertices in towns might not have a nearby hub, and
		\item the gap between the cover parameter $(2+\eps)r$, to the radius of the sparse ball $\frac{2.8+6\eps}{2+\eps}$ is less than $3$ (unless we pick large $\eps\ge1$, but we usually think of $\eps$ as being rather small).
	\end{enumerate}
	The second issue is of a technical nature, as the parameter $3$ in~\cite{Fil19padded} is somewhat arbitrary, and every constant gap will actually be sufficient.
	A natural solution to the first issue is to \enquote{enrich} the set of hubs by adding a single representative point from every town. This will ensure that every point has a center at distance at most $(2+\eps)r$, however, now there is no bound on the number of centers contained in a ball of radius $2.8r$ (or even smaller).

	Filtser's~\cite{Fil19padded} meta theorem is based on an interpretation of the exponential random shifts algorithm of Miller, Peng, and Xu~\cite{MPX13}. This is an algorithm that given a set of centers $N$, sends each vertex to its closest center (that is Voronoi partition), but where the closest center is chosen w.r.t.\ a random shift sampled using an exponential distribution. Roughly speaking, the analysis shows that the resulting padding parameter is logarithmic in the number of nearby centers.
	Our solution to the issue above is to sample the shifts for the town representatives using a different distribution from the one used for the centers in \SPC. We will choose the distributions such that a vertex from the sprawl will never join a cluster centered in a town vertex. The analysis of the padding probability will now go through. We will prove the following:
\PaddedDecomp*
	We continue with background on the~\cite{MPX13} clustering algorithm, and then with the algorithm we will be using, and finally with a formal analysis.

	\subsubsection{Clustering Algorithm Using Shifted Starting Times \texorpdfstring{\cite{MPX13}}{}}\label{subsec:clustering}
    
	Let $\Delta>0$ be some parameter and let $N\subseteq V$ be some set of centers such that for every $v\in V$, $d_G(v,N)\le \Delta$.
	For each center $x\in N$, let $\delta_x\in [0,\Gamma]$, for a parameter $\Gamma$ to be chosen later.
    The choice of $\{\delta_x\}_{x\in N}$ differs among different implementations of the algorithm.
	Each vertex $v$ will join the cluster $C_x$ of the center $x\in N$ for which the value $f_v(x)=\delta_x-d_G(x,v)$ is maximized. Ties are broken in a consistent manner, that is we have some order $x_1,x_2,\dots$, and among the centers $x_i$ that maximize
    $\delta_{x_i}-d_G(x_i,v)$, $v$ will join the cluster of the center with minimal index.
	Note that it is possible that a center $x\in N$ will join the cluster of a different center $x'\in N$.
	An intuitive way to think about the clustering process is as follows: each center $x$ wakes up at time $-\delta_x$ and begins to \enquote{spread} in a continuous manner. The spread of all centers is done in the same unit speed. A vertex $v$ joins the cluster of the first center that reaches it.
	The proof of the two following claim was known before (see e.g.~\cite{Fil19padded}).

	\begin{claim}\label{claim:StrongDiam}
		Every non-empty cluster $C_x$ created by the algorithm has strong diameter at most $2\cdot(\Delta+\Gamma)$.
	\end{claim}
    
%	\begin{wrapfigure}{r}{0.13\textwidth}
%		\begin{center}
%			\vspace{-20pt}
%			\includegraphics[width=0.95\linewidth]{fig_edited/clusterPrefer}
%			\vspace{-17pt}
%		\end{center}
%		\vspace{-20pt}
%	\end{wrapfigure}

\begin{sidefigure}
    \begin{center}
        \includegraphics[width=0.25\linewidth]{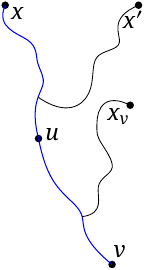}
    \end{center}
    \caption{Illustration of the main argument during the proof of \autoref{claim:StrongDiam}.}
    \label{fig:TBP3}
\end{sidefigure} 

 \begin{proof}
	Consider a vertex $v\in C_x$. We argue that $d_G(v,x)\le \Delta+\Gamma$. This will  imply that $C_x$ has weak diameter $2\cdot(\Delta+\Gamma)$.
	Let $x_v$ be the closest center to $v$, then $d_G(v,x_v)\le \Delta$. As $v$ joined the cluster of $x$, it holds that
	$\delta_{x}-d_{G}(v,x)=f_{v}(x)\ge f_{v}(x_{v})=\delta_{x_{v}}-d_{G}(v,x_{v})$, and in particular $d_{G}(v,x)\le\delta_{x}+d_{G}(v,x_{v})\le\Delta+\Gamma$.

	Let $\mathcal{I}$ be the shortest path in $G$ from $v$ to $x$ (the blue path in \autoref{fig:TBP3}).
	For every vertex $u\in\mathcal{I}$ and center $x'\in N$, it holds that
	\begin{align*}
		f_{u}(x)=\delta(x)-d_{G}(u,x) & =\delta(x)-\left(d_{G}(v,x)-d_{G}(v,u)\right)\\
		& \stackrel{(*)}{\ge}\delta(x')-d_{G}(v,x')+d_{G}(v,u)\ge\delta(x')-d_{G}(u,x')=f_{u}(x')~.
	\end{align*}
	Where in the inequality $^{(*)}$ we used the fact that $v$ chose $x$ over $x'$ and thus $f_v(x)\ge f_v(x')$.
	If the inequality $^{(*)}$ is strict, then $f_u(x)> f_u(x')$ and $u$ will prefer $x$ over $x'$. Otherwise, necessarily $f_v(x)= f_v(x')$ and thus $x$ has smaller index then $x'$ (as it was preferred by $v$). In particular
	$u$ will prefer $x$ over $x'$ (in case $f_u(x)=f_u(x')$).
	We conclude that $u$ will prefer the center $x$ over any other center. It follows that $\mathcal{I}\subseteq C_x$. In particular, $d_{G[C_x]}(v,x)\le\Delta+\Gamma$.
	The claim now follows.
	\end{proof}

	\begin{claim}\label{claim:PadProperty}
		Consider a vertex $v$, and let $x_{(1)},x_{(2)},\dots$ be an ordering of the centers w.r.t.~$f_v$. That is $f_v(x_{(1)})\ge f_v(x_{(2)})\ge\dots$ .
		Set $\Upsilon=f_v(x_{(1)})- f_v(x_{(2)})$. Then for every vertex $u$ such that $d_G(v,u)<\frac{\Upsilon}{2}$ it holds that $u\in C_{x_{(1)}}$.
	\end{claim}
	\begin{proof}
		For every center $x_{(i)}\ne x_{(1)}$, using the triangle inequality it holds that,
		\begin{align*}
			f_{u}(x_{(1)}) & =\delta_{x_{(1)}}-d_{G}(u,x_{(1)})>\delta_{x_{(1)}}-d_{G}(v,x_{(1)})-\frac{\Upsilon}{2}\\
			& \ge\delta_{x_{(i)}}-d_{G}(v,x_{(i)})+\frac{\Upsilon}{2}>\delta_{x_{(i)}}-d_{G}(u,x_{(i)})=f_{u}(x_{(i)})~.
		\end{align*}
		In particular, $u\in C_{x_1}$.
	\end{proof}

	\paragraph*{Truncated Exponential Distributions.} Similarly to previous works, %
	we will use a truncated exponential distribution. That is, a exponential distribution conditioned on the event that the outcome lays in a certain interval.
	The \emph{$[\theta_1,\theta_2]$-truncated exponential distribution} with
	parameter~$\lambda$ is denoted by $\Texp_{[\theta_1, \theta_2]}(\lambda)$, and
	the density function is:
	$f(y)= \frac{ \lambda\, e^{-\lambda\cdot y} }{e^{-\lambda \cdot \theta_1} - e^{-\lambda
			\cdot \theta_2}}$, for $y \in [\theta_1, \theta_2]$.

	\subsubsection{Clustering Algorithm for Graphs with Bounded Highway Dimension}
	Consider a weighted graph $G=(V,E,w)$ with highway dimension $h:\mathbb{R}_{\geq 0}\to\mathbb{N}\cup\{\infty\}$. Fix $r\in\mathbb{R}_{\geq 0}$ and $\eps\in[0,\frac{1}{4}]$.
	Let $\cT$ be the towns (from \autoref{def:towns}) of an $(r,\eps)$-shortest path cover. For every town $T\in\cT$ let $v_T\in T$ be the vertex such that $d_G(v_T,\SPC)>(2+\eps)\cdot r$ and $T=\left\{u\in V\mid d_G(u,v_T)\le r\right\}$. We will call $v_T$ the \emph{town center} of $T$, and let $N_T=\{v_T\}_{T\in\cT}$ be the set of town centers. While all the vertices in $V\setminus T$ are at distance greater than $r$ from $T$ we can say something stronger for the town center:
	\begin{claim}\label{clm:SeprationTownCenter}
		For every town $T\in \cT$ it holds that $d_G(v_T,V\setminus T)>(2-\eps)\cdot r$.
	\end{claim}
	\begin{proof}
		For the sake of contradiction, suppose that there is a vertex $u\in V\setminus T$ such that $d_G(v_T,u)\le (2-\eps)\cdot r$. As $u\notin T$, it holds that $d_G(v_T,u)>r$.
		It follows that there is a hub $x\in\SPC$ such that
		\begin{align*}
			d_{G}(v_{T},x) & \le d_{G}(v_{T},x)+d_{G}(x,u)\le(1+\eps)\cdot d_{G}(v_{T},u)\\
			& \le(1+\eps)(2-\eps)\cdot r\le (2+\eps)\cdot r~
		\end{align*}
		a contradiction.
	\end{proof}

	As $r>0$ is arbitrary, for the sake of simplicity, we will assume from now on that $r=1$. Afterwards, everything could be scaled accordingly.
	We will run the~\cite{MPX13} algorithm described above w.r.t.\ the set of centers $N=\SPC\cup N_T$. Set $t=\frac{1}{2}$, $a=\frac{1}{2}+2\eps$, $b=1+2\eps$,
	 and $\lambda=4\cdot(\ln\left(2\cdot h(\eps)^{2}+1\right)+1)=O(\log(h(\eps))$.
	For each center $v_T\in N_T$ we will sample its shift $\delta_{v_T}$ from $\Texp_{[0,t]}(\lambda)$, while for $x\in\SPC$ we will sample its shifts $\delta_x$ from $\Texp_{[a,b]}(\lambda)$.
	Then we run the~\cite{MPX13} clustering algorithm w.r.t.\ the graph $G$, and the sampled shifts.

	\subsubsection{Proof of \autoref{thm:padded}}\label{subsec:ThmPadProof}
	Clearly, by \autoref{claim:StrongDiam}, as all the shifts are chosen in $[0,b]$, the created partition has strong diameter $2\cdot(2+\eps+b)=6\cdot(1+\eps)$. It remains to prove the padding guarantee.
	Let $\mu=\frac18$, and fix $\gamma\in[0,\mu]$. Consider the ball $B_\gamma=B_{G}(u,\gamma)$.
	We begin by bounding the possible centers of clusters to which the vertices of $B_\gamma$ can join.

	\begin{claim}\label{clm:FarAreIrrelevant}
		Let $x\in N$ be a center such that $d_{G}(x,u)> 2.8+4\eps$. Then no vertex in $B_\gamma$ will join the cluster of $x$.
	\end{claim}
	\begin{proof}
		Suppose first that $u\in T\in\cT$ belongs to a town with center $v_T$.
		Consider a vertex $z\in B_\gamma$. As the distance from $T$ to $V\setminus T$ is greater than $1$ (\autoref{lem:townproperties}),
        and $d_{G}(u,z)\le \gamma\le\mu<1$, it follows that $z\in T$ and thus $d_{G}(z,v_T)\le 1$.
		That is, the entire ball $B_\gamma$ is contained in the town $T$.
		As the maximum shift of $x$ is $b$ we have
		\begin{align*}
			&f_{z}(v_{T}) =\delta_{v_{T}}-d_{G}(v_{T},z)\ge-1\\
			&f_{z}(x)  =\delta_{x}-d_{G}(x,z)\le\delta_{x}-\left(d_{G}(x,u)-d_{G}(u,z)\right)< b-\left(2.8+4\eps-\mu\right)~.
		\end{align*}
		As $b+\mu\le 1.8+4\eps$, %
  it follows that $f_{z}(x)< f_{z}(v_{T})$ and thus indeed no vertex $z\in B_\gamma$ will  join the clustered centered at $x$.

		Suppose next that $u\in \cS$ is in the sprawl.
		As the distance from every town to the sprawl is greater than $1$ by \autoref{lem:townproperties}, it follows that the entire ball $B_\gamma$ is in the sprawl. Consider a vertex $z\in B_\gamma$. There is a hub $x_z$ at distance at most $d_{G}(z,x_z)\le 2+\eps$ from $z$ by \autoref{lem:townproperties}. As the minimum shift of $x_z$ is $a$, while the maximum shift of $x$ is $b$, we have that
		\begin{align*}
			& f_{z}(x_{z})=\delta_{x_{z}}-d_{G}(x_{z},z)\ge a-\left(2+\eps\right)\\
			& f_{z}(x)=\delta_{x}-d_{G}(x,z)\le\delta_{x}-\left(d_{G}(x,u)-d_{G}(u,z)\right)< b-\left(2.8+4\eps-\mu\right)~.
		\end{align*}
		As $b-a+\mu\le 0.8+3\eps$, %
  it follows that $f_{z}(x)< f_{z}(x_u)$ and thus indeed no vertex $z\in B_\gamma$ will not join the clustered centered at $x$.
	\end{proof}

	\begin{claim}\label{clm:TownCenterOnlyClustersItsTown}
		Consider a town $T\in \cT$ with center $v_T$, and suppose that $u\notin T$. Then no vertex in $B_\gamma$ will join the cluster centered at $v_T$.
	\end{claim}
	\begin{proof}
		As the distance from $T$ to $V\setminus T$ is greater than $1$ by \autoref{lem:townproperties}, it follows that the entire ball $B_\gamma$ is disjoint from $T$. In particular, by \autoref{clm:SeprationTownCenter}, the distance from $v_T$ to every vertex in $B_\gamma$ is greater than $2-\eps$.

		Suppose first that $u\in \cS$ is in the sprawl, and fix a vertex $z\in B_\gamma$. There is a hub $x_z\in\SPC$ at distance at most $2+\eps$ from $z$ by \autoref{lem:townproperties}. As the maximum shift of $v_T$ is $t$, while the minimum shift of $x_z$ is $a$, we have that
		\begin{align*}
			f_{z}(x_{z}) & =\delta_{x_{z}}-d_{G}(x_{z},z)\ge a-\left(2+\eps\right)\\
			f_{z}(v_{T}) & =\delta_{v_{T}}-d_{G}(v_{T},z)< t-\left(2-\eps\right)~.
		\end{align*}
		As $a-t\ge2\eps$, %
  it follows that $f_{z}(v_{T})< f_{z}(x_{z})$ and thus indeed no vertex $z\in B_\gamma$ will join the cluster centered at $v_T$.

		Next suppose that $u$ belongs to some town $T'\in\cT$ where $T'\ne T$. Let $v_{T'}$ be the town center of~$T'$. Following the exact same argument as in the first part of \autoref{clm:FarAreIrrelevant}, the entire ball $B_\gamma$ is contained in $T'$. Consider a vertex $z\in B_\gamma$. As the distance from $z$ to $v_{T'}$ is at most $1$, and the maximum shift of $v_T$ is $t$, we have that
		\begin{align*}
			f_{z}(v_{T'}) & =\delta_{v_{T'}}-d_{G}(v_{T'},z)\ge-1\\
			f_{z}(v_{T}) & =\delta_{v_{T}}-d_{G}(v_{T},z)< t-\left(2-\eps\right)~,
		\end{align*}
		As $t\le 1-\eps$ (recall that $\eps\in[0,\frac14]$),
  it follows that $f_{z}(v_{T})< f_{z}(v_{T'})$ and thus indeed no vertex $z\in B_\gamma$ will not join the clustered centered at $v_T$.
	\end{proof}

	Let
    $N_u=\SPC\cap B_{G}(u,2.8+4\eps)$ be all the hubs at distance at most $2.8+4\eps$ from $u$.
	In addition, if $u\in T$, add $v_T$ to~$N_u$.
	By \autoref{lem:hub_bound_2.8Ball}, $|N_u|\le 2\cdot h(\eps)^2+1$. From  \autoref{clm:FarAreIrrelevant} and \autoref{clm:TownCenterOnlyClustersItsTown} it follows that vertices in $B_\gamma$ can only join clusters centered in vertices from $N_u$.

	Fix an arbitrary order $N_u=\{x_1,x_2,\dots\}$.
	In the created partition, denote by $C_{x_i}$ the cluster centered in $x_i$.
	Denote by $\mathcal{F}_i$ the event that $u$ joins the cluster of $x_i$, i.e., $u\in C_{x_i}$.
	Denote by $\mathcal{C}_i$ the event that $u$ joins the cluster of $x_i$, but not all of the vertices in $B_\gamma$ joined that cluster, that is $u\in C_{x_i}\cap B_\gamma\neq B_\gamma$.
	\autoref{thm:padded} will follow once we show that $\Pr\left[\cup_{i}\mathcal{C}_{i}\right]\le 1-e^{-O(\gamma\cdot\lambda)}$.
	Set $\alpha=e^{-2\gamma\cdot\lambda}$.

	\begin{claim}
		For every $i$,
		$\Pr\left[\mathcal{C}_{i}\right]\le\left(1-\alpha\right)\left(\Pr\left[\mathcal{F}_{i}\right]+\frac{1}{e^{\lambda\cdot(1-\eps)}-1}\right)$.
	\end{claim}
	\begin{proof}
		Fix $i$ and denote $x=x_{i}$, $\mathcal{C}=\mathcal{C}_i$, $\mathcal{F}=\mathcal{F}_i$, and $\delta=\delta_{x_i}$.
		Let $X\in [0,b]^{|N_u|-1}$ be the vector where the $j$'th coordinate equals $\delta_{x_j}$ (skipping over $x_i$, that is the $j$'th coordinate equals $\delta_{x_{j+1}}$ for $j>i$).
        Set $\rho_{X}=d_{G}(x,v)+\max_{j<|N_u|}\left\{ \delta_{x_{j}}-d_{G}(x_{j},v)\right\}$. Note that~$\rho_X$ is the minimal value of $\delta$ such that if $\delta>\rho_X$,
        then $x$ has the maximal value $f_u(x)$, and therefore~$v$ will join the cluster of $x$.

		Assume first that $x\in\SPC$.
		Note that it is possible that $\rho_X>b$ (in which case $\Pr[\cF\mid X]=0$).
		It is also possible that $\rho_X\le a$ (in which case $\Pr[\cF\mid X]=1$).
		Denote $\tilde{\rho}_{X}=\max\{\rho_X,a\}$.
		Conditioning on all the other shift samples having values $X$, and assuming that $\rho_{X}\le b$ it holds that
		\[
		\Pr\left[\mathcal{F}\mid X\right]=\Pr\left[\delta>\rho_{X}\right]=\int_{\tilde{\rho}_{X}}^{b}\frac{\lambda\cdot e^{-\lambda y}}{e^{-\lambda\cdot a}-e^{-\lambda\cdot b}}dy=\frac{e^{-\lambda\cdot\tilde{\rho}_{X}}-e^{-\lambda\cdot b}}{e^{-\lambda\cdot a}-e^{-\lambda\cdot b}}~.
		\]
		If $\delta>\rho_{X}+2\gamma$ then $f_u(x)=\delta-d_{G}(x,u)>\max_{j\ne i}\left\{ \delta_{x_{j}}-d_{G}(x_{j},u)\right\} +2\gamma=\max_{j\ne i}f_u(x_j) +2\gamma$. In particular, by \autoref{claim:PadProperty} the ball $B_\gamma$ will be contained in $C_{x}$. If $\rho_X+2\gamma<a$ then $\Pr[\cC\mid X]=0$ and we are done. We will thus assume $\rho_X+2\gamma\ge a$.
		We conclude
		\begin{align*}
			\Pr\left[\mathcal{C}\mid X\right] & \le\Pr\left[\rho_{X}\le\delta\le\rho_{X}+2\gamma\right]\\
			& \le\int_{\tilde{\rho}_{X}}^{\rho+2\gamma}\frac{\lambda\cdot e^{-\lambda y}}{e^{-\lambda\cdot a}-e^{-\lambda\cdot b}}dy\\
			& \le\frac{e^{-\lambda\cdot\tilde{\rho}_{X}}-e^{-\lambda\cdot(\tilde{\rho}_{X}+2\gamma)}}{e^{-\lambda\cdot a}-e^{-\lambda\cdot b}}\\
			& =\left(1-e^{-2\gamma\cdot\lambda}\right)\cdot\frac{e^{-\lambda\cdot\tilde{\rho}_{X}}}{e^{-\lambda\cdot a}-e^{-\lambda\cdot b}}\\
			& =\left(1-\alpha\right)\cdot\left(\Pr\left[\mathcal{F}\mid X\right]+\frac{e^{-\lambda\cdot b}}{e^{-\lambda\cdot a}-e^{-\lambda\cdot b}}\right)\\
			& =\left(1-\alpha\right)\cdot\left(\Pr\left[\mathcal{F}\mid X\right]+\frac{1}{e^{\lambda\cdot(b-a)}-1}\right)~.
		\end{align*}
		Denote by $f$ the density function of the distribution over all possible values of $X$. Using the law of total probability, we can bound the probability that the cluster of $x$ cuts $B_\gamma$
		\begin{align*}
			\Pr\left[\mathcal{C}\right] & =\int_{X}\Pr\left[\mathcal{C}\mid X\right]\cdot f(X)~dX\\
			& \le\left(1-\alpha\right)\cdot\int_{X}\left(\Pr\left[\mathcal{F}\mid X\right]+\frac{1}{e^{\lambda\cdot(b-a)}-1}\right)\cdot f(X)~dX\\
			& =\left(1-\alpha\right)\cdot\left(\Pr\left[\mathcal{F}\right]+\frac{1}{e^{\lambda\cdot(b-a)}-1}\right)\\
			& =\left(1-\alpha\right)\cdot\left(\Pr\left[\mathcal{F}\right]+\frac{1}{e^{\frac\lambda2}-1}\right)~.
		\end{align*}

		Next we consider the case where $x\notin\SPC$. Here $u$ belongs to a town $T$ and $x=v_T$. The analysis follows the exact same lines with the only difference being that $\delta$ is sampled from $[0,t]$ instead of~$[a,b]$. In particular, it holds that $\Pr\left[\mathcal{F}\mid X\right]=\Pr\left[\delta>\rho_{X}\right]=\frac{e^{-\lambda\cdot\tilde{\rho}_{X}}-e^{-\lambda\cdot t}}{1-e^{-\lambda\cdot t}}$, and
		\begin{align*}
			\Pr\left[\mathcal{C}\mid X\right] & \le\Pr\left[\rho_{X}\le\delta\le\rho_{X}+2\gamma\right]\le\frac{e^{-\lambda\cdot\tilde{\rho}_{X}}-e^{-\lambda\cdot(\tilde{\rho}_{X}+2\gamma)}}{1-e^{-\lambda\cdot t}}\\
			& =\left(1-e^{-2\gamma\cdot\lambda}\right)\cdot\frac{e^{-\lambda\cdot\tilde{\rho}_{X}}}{1-e^{-\lambda\cdot t}}=\left(1-\alpha\right)\cdot\left(\Pr\left[\mathcal{F}\mid X\right]+\frac{1}{e^{\lambda\cdot t}-1}\right)~.
		\end{align*}
		Hence using the law of total probability it follows that
		\begin{align*}
			\Pr\left[\mathcal{C}\right] & =\int_{X}\Pr\left[\mathcal{C}\mid X\right]\cdot f(X)~dX\\
			& \le\left(1-\alpha\right)\cdot\int_{X}\left(\Pr\left[\mathcal{F}\mid X\right]+\frac{1}{e^{\lambda\cdot t}-1}\right)\cdot f(X)~dX=\left(1-\alpha\right)\cdot\left(\Pr\left[\mathcal{F}\right]+\frac{1}{e^{\frac\lambda2}-1}\right)~.
		\end{align*}
	\end{proof}

    To finalize the proof of \autoref{thm:padded}, we bound the probability that the ball $B_\gamma$ is cut.
	\begin{align*}
		\Pr\left[\cup_{i}\mathcal{C}_{i}\right]\leq\sum_{i=1}^{|N_{u}|}\Pr\left[\mathcal{C}_{i}\right] & \le\left(1-\alpha\right)\cdot\sum_{i=1}^{|N_{u}|}\left(\Pr\left[\mathcal{F}_{i}\right]+\frac{1}{e^{\frac{\lambda}{2}}-1}\right)\\
		& \le\left(1-e^{-2\gamma\cdot\lambda}\right)\cdot\left(1+\frac{|N_{u}|}{e^{\frac{\lambda}{2}}-1}\right)\\
		& \le\left(1-e^{-2\gamma\cdot\lambda}\right)\cdot\left(1+e^{-2\gamma\cdot\lambda}\right)=1-e^{-4\gamma\cdot\lambda}~,
	\end{align*}
	where the last inequality holds as
	\[
	e^{-2\gamma\lambda}=\frac{e^{-2\gamma\lambda}\left(e^{\frac{\lambda}{2}}-1\right)}{e^{\frac{\lambda}{2}}-1}\ge\frac{e^{-2\gamma\lambda}\cdot e^{\frac{\lambda}{2}-1}}{e^{\frac{\lambda}{2}}-1}\ge\frac{e^{\frac{\lambda}{4}-1}}{e^{\frac{\lambda}{2}}-1}\ge\frac{|N_{u}|}{e^{\frac{\lambda}{2}}-1}~.
	\]
	Here the second inequality holds as $\gamma\le\mu\le\frac{1}{8}$, and the last inequality holds by the definition
	of $\lambda=4\cdot(\ln\left(2\cdot h(\eps)^{2}+1\right)+1)\ge 4\cdot(\ln\left(|N_u|\right)+1)$.
	To conclude, we obtain a strongly $6\cdot(1+\eps)$-bounded partition, such that for every $\gamma\le\mu=\frac18$
	 and $v\in V$, the ball $B_G(v,\gamma)=B_G\left(v,\frac{\gamma}{6\cdot(1+\eps)}\cdot6\cdot(1+\eps)\right)$ is fully contained in a single cluster with probability at least
	\[
	\Pr\left[B_{G}\left(v,\frac{\gamma}{6\cdot(1+\eps)}\cdot6\cdot(1+\eps)\right)\subseteq P(v)\right]\ge e^{-4\cdot\gamma\cdot\lambda}=e^{-\frac{\gamma}{6\cdot(1+\eps)}\cdot24\cdot(1+\eps)\cdot\lambda}~.
	\]
	Thus we constructed a strong $\left(24\cdot(1+\eps)\cdot\lambda,24\cdot(1+\eps)\right)=\left(O(\log h(\eps)),\Omega(1)\right)$
	padded decomposition scheme, as required.

\subsection{Tree Cover}
\TreeCover*
\begin{proof}
	In contrast to the previous constructions in this section, here we will be using the bounds on the shortest path covers in \autoref{lem:sparse-SPC} instead of \autoref{lem:hub_bound}.
	For the sake of simplicity, we will assume that the minimum pairwise distance in $G$ is strictly larger than $1$ (otherwise, one can scale accordingly). Let $\Phi=\max_{x,y\in V}d_G(x,y)$ be the maximal pairwise distance (i.e., the diameter).
	Fix~$\delta=\frac{\eps}{3}$. For every $i\in\{0,1,\dots,\left\lfloor\log_{1+\delta}\Phi\right\rfloor\}$, let $r_i=(1+\delta)^i$, and let $\SPC_i$ be the $(r_i,\eps)$-shortest path cover from \autoref{lem:sparse-SPC}.
    That is, for every pair $u,z\in V$ such that $d_X(u,z)\in (r_i,(2+\eps)r_i]$, there is a hub $x\in\SPC_i$
	for which $d_X(u,x)+d_X(x,z)\leq(1+\eps)d_X(u,z)$. Furthermore, $\SPC_i$ is locally $h(\eps)$-sparse: every ball of radius $(2+4\eps)\cdot r_i$ contains at most $h(\eps)$ vertices from $\SPC_i$.

	Next, we partition each set $\SPC_i$ into disjoint subsets $\SPC_{i,1},\SPC_{i,2},\dots,\SPC_{i,h(\eps)}$ such that for every $j$ and $x,y\in \SPC_{i,j}$, $d_G(x,y)>(2+4\eps)r_i$. Such a partition can be constructed in a greedy manner: construct the sets one by one, where we add each hub to the first subset where there is no other hub at distance at most $(2+4\eps)\cdot r_i$, as each hub has only $h(\eps)-1$ other hubs at distance at most $(2+4\eps)r_i$, it will surely join $\SPC_{i,j}$ for $j\le h(\eps)$.

    Fix $K=\lceil\log_{1+\delta}\frac{32}{\eps}\rceil=O(\frac{1}{\eps}\cdot\log\frac{1}{\eps})$. For every $j\in \{1,\dots,h(\eps)\}$, and $q\in\{0,\dots,K-1\}$ we will construct a dominating tree $T_{q,j}$.
    One can think about $K$ as the gap between two consecutive levels in the same tree, while $q$ is an initial shift.
    In fact, in our construction, $T_{q,j}$ will be a forest (that is, it might not be connected). But we can easily make it a tree by adding edges between the connected components (with large enough weights so that the distances will still be dominating).
	As a result, we will obtain a collection of $h(\eps)\cdot K=h(\eps)\cdot O(\frac{1}{\eps}\cdot\log\frac{1}{\eps})$ trees as required.
    The tree $T_{q,j}$ will take care of all the hub sets $\{\SPC_{q+K\cdot k,j}\}_{k\ge 0}$.
    In more detail, consider $u,z\in V$ such that $d_X(u,z)\in (r_i,r_{i+1}]$, and $x\in\SPC_i$ such that $d_X(u,x)+d_X(x,z)\leq(1+\eps)d_X(u,z)$. Let $q\in\{0,\dots,K-1\}$ and $k\ge 0$ such that $i=q+k\cdot K$, and $j\in \{1,\dots,h(\eps)\}$ such that $x\in \SPC_{i,j}$. It will hold that $d_{T_{q,j}}(u,z)\le (1+2\eps)\cdot d_G(u,z)$.

    On an intuitive level, the construction of the tree $T_{q,j}$ is very simple.
    Every vertex will be connected to its closest hub in $\SPC_{q,j}$, and in general, every hub in $\SPC_{q+k\cdot K,j}$ will be connected to its closest hub in $\SPC_{q+(k+1)\cdot K,j}$.
    We will obtain the distortion guarantee due to the big gap of $(1+\delta)^K\ge\frac{32}{\eps}$ between two consecutive hub-sets.
    However, as the hub-sets might not be hierarchical, the resulting graph will not be a tree.
    To overcome this issue, for every $k$ and hub $x\in \SPC_{q+k\cdot K,j}$, we will create a copy $x_k$, and will connect to the copy instead of connecting to the actual vertex. As a result, we will obtain a tree.

    \paragraph*{Construction of $T_{q,j}$.} Fix $j\in \{1,\dots,h(\eps)\}$ and $q\in\{0,\dots,K-1\}$. For every $k\ge0$, denote $q(k)=q+k\cdot K$.
    The construction of $T_{q,j}$ is bottom-up. All vertices in $V$ will be leaves.
    For every~$k\ge 0$, and hub $x\in \SPC_{q(k),j}$, we add an auxiliary copy of vertex $x_k$.
    Every vertex $v\in V$ is called a $(-1)$-level representative.
    For every $(-1)$-level representative $v$ such that there is $x\in \SPC_{q,j}=\SPC_{q(0),j}$ at distance at most $d_G(v,x)\le (1+2\eps)\cdot r_{q(0)}$, add an edge from $v$ to $x_0$ (of weight $d_G(v,x)$).
    Note that as the pairwise distance between any two vertices in $\SPC_{q(0),j}$ is greater than $(2+4\eps)\cdot r_{q(0)}$, there is at most a single such vertex.
    If there is no such vertex, we don't add an edge from $v$ to any vertex. In this case, $v$ will also be called a $0$-level representative.
    In addition, all the copies of the vertices in $\SPC_{q(0),j}$ will be $0$-level representatives.
    We think of the current graph as a forest rooted at the $0$-level representatives.

    In general, after $k$ such steps, we have a forest with the vertices $V$ and copies of the vertices $\SPC_{q(0),j},\SPC_{q(1),j},\dots,\SPC_{q(k),j}$.
    Note that if a vertex $x$ belongs to $\SPC_{q(p),j}$ and $\SPC_{q(p'),j}$, then we will have different copies $x_p,x_{p'}$.
    In addition, after $k$ steps we have the set of $k$-level representatives which are the auxiliary vertices added for $\SPC_{q(k),j}$, and in addition all the $k-1$-level representatives with no vertex from $\SPC_{q(k),j}$ at distance at most $(1+2\eps)\cdot r_{q(k)}$.
    In particular, every leaf $v\in V$ has a corresponding  $k$-level representative.
    In step $k+1$, for every vertex $x\in \SPC_{q(k+1),j}$ we add an auxiliary copy of vertex $x_{k+1}$.
    Then, for every $k$-level representative $v'$ (that is, $v'$  is a copy of $v$, or $v$ itself) such that there is $x\in \SPC_{q(k+1),j}$ at distance at most $d_G(v,x)\le(1+2\eps)\cdot r_{q(k+1)}$, add an edge from $v'$ to $x_{k+1}$ (of weight $d_G(v,x)$).
    Note that there is at most a single such vertex. If there is no such vertex, we do not add an edge from $v'$ to any vertex. In this case, $v'$ will also be a $k+1$-level representative.
    We finish the construction once $q+k\cdot K>\log_{1+\delta}\Phi$.

    \paragraph*{Analysis.} Clearly, $T_{q,j}$ is a tree.
    The following claim will be useful.

	\begin{claim}\label{clm:TreeCoverDistRepresentative}
        Consider a vertex $v\in V$ in the tree $T_{q,j}$, and let $x'$ be the $k$-level representative of $v$.
        Then $d_{T_{q,j}}(v,x')\le4\cdot r_{q(k)}$.
	\end{claim}
	\begin{subproof}
	The proof is by induction on the level $k$. Note that $v$ is the $(-1)$-level representative of itself, and thus the claim trivially holds for $k=-1$.
    In general, let $y_k$ and $y_{k+1}$ be the $k$, and $(k+1)$-level representatives of $v$ (respectively).
    Suppose that $d_{T_{q,j}}(v,y_k)\le4\cdot r_{q(k)}$, and we will show that
    $d_{T_{q,j}}(v,y_{k+1})\le4\cdot r_{q(k+1)}$.

    If $y_k$ is also the $(k+1)$-level representative of $v$ then clearly $d_{T_{q,j}}(v,y_{k+1})\le4\cdot r_{q(k)}\le4\cdot r_{q(k+1)}$.
    Otherwise, $T_{q,j}$ contains an edge of weight at most $(1+2\eps)\cdot r_{q'}$ from $y_k$ to $y_{k+1}$. By the triangle inequality, it holds that
    \begin{align*}
    d_{T_{q,j}}(v,y_{k+1}) & =d_{T_{q,j}}(v,y_{k})+d_{T_{q,j}}(y_{k},y_{k+1})\le4\cdot r_{q(k)}+(1+2\eps)\cdot r_{q(k+1)}\\
     & =(1+2\eps+\frac{4}{(1+\delta)^{K}})\cdot r_{q(k+1)}\le4\cdot r_{q(k+1)}~,
    \end{align*}
    where in the last inequality we used that $\eps\le 1$, and $K=\lceil\log_{1+\delta}\frac{32}{\eps}\rceil$.
\end{subproof}

    Now we are ready to prove the theorem. Consider a pair of vertices $u,v\in V$, and let $q\in \{0,\dots,K-1\}$, $k\in\N$ be the unique indices such that such that
	$$(1+\delta)^{q+K\cdot k}< d_G(u,v)\le (1+\delta)^{q+K\cdot k+1}~.$$
    There is a hub $x\in\SPC_{q(k)}$ such that $d_G(u,x)+d_G(x,v)\le(1+\eps)\cdot d_G(u,v)$. Let $j\in\{1,\dots,h(\eps)\}$ be the unique index such that $x\in\SPC_{q(k),j}$.
	Let $y'_u$ and $y'_v$ be the $(k-1)$-level representatives in $T_{q(k),j}$ of $u$ and $v$ respectively.
    Suppose also that $y'_u$ and $y'_v$ are copies of $y_u$ and $y_v$ respectively (or the original vertices themselves).
    Using     \autoref{clm:TreeCoverDistRepresentative},
    \begin{align*}
\max\left\{ d_{T_{q,j}}(v,y'_{v}),d_{T_{q,j}}(u,y'_{u})\right\}  & \le4\cdot r_{q(k-1)}=\frac{4\cdot r_{q(k)}}{(1+\delta)^{K}}\le\frac{4\cdot(1+\delta)}{(1+\delta)^{K}}\cdot d_{G}(u,v)\\
 & \le\frac{4\cdot(1+\delta)\cdot\eps}{32}\cdot d_{G}(u,v)\le\frac{\eps}{4}\cdot d_{G}(u,v)~.
\end{align*}
    Using the triangle inequality, it follows that
    \begin{align*}
d_{G}(y_{v},x) & \le d_{G}(y_{v},v)+d_{G}(v,x)\le\frac{\eps}{4}\cdot d_{G}(u,v)+(1+\eps)\cdot d_{G}(u,v)\\
 & \le(1+\frac{5}{4}\eps)\cdot(1+\delta)\cdot r_{q(k)}=(1+\frac{5}{4}\eps)\cdot(1+\frac{\eps}{3})\cdot r_{q(k)}\\
 & =\left(1+\frac{19}{12}\eps+\frac{5}{12}\cdot\eps^{2}\right)\cdot r_{q(k)}\le\left(1+2\eps\right)\cdot r_{q(k)}.
\end{align*}
    Similarly, $d_{G}(y_{u},x)\le\left(1+2\eps\right)\cdot r_{q(k)}$.
    It follows that we added to $T_{q,j}$ edges from $x$ to $y'_u$ and $y'_v$.
    Using the triangle inequality again, we conclude:
    \begin{align*}
    d_{T_{q,j}}(u,v) & \le d_{T_{q,j}}(u,y'_{u})+d_{G}(y{}_{u},x)+d_{G}(x,y{}_{v})+d_{T_{q,j}}(y'_{v},v)\\
 & \le2\cdot d_{T_{q,j}}(u,y'_{u})+d_{G}(u,x)+d_{G}(x,v)+2\cdot d_{T_{q,j}}(y'_{v},v)\\
 & \le(1+\eps)\cdot d_{G}(u,v)+4\cdot\frac{\eps}{4}\cdot d_{G}(u,v)=(1+2\eps)\cdot d_{G}(u,v)~.\qedhere
    \end{align*}

\end{proof}

\section{Applications}
\subsection{Embedding into \texorpdfstring{$\ell_p$}{lp} spaces}
A metric embedding is a map between two metric spaces that preserves all pairwise distances up to a small stretch. We say that embedding $f:X\rightarrow Y$ between $(X,d_X)$ and $(Y,d_Y)$ has distortion~$t$ if for every $x,y\in X$ it holds that $d_X(x,y)\le d_Y(f(x),f(y))\le t\cdot d_X(x,y)$. Krauthgamer, Lee, Mendel and Naor~\cite{KLMN04} (improving over Rao~\cite{Rao99}) showed that every $n$-point metric space that admits a $(\beta,\frac1\beta)$-padded decomposition scheme can be embedded into an $\ell_p$ space with distortion $O(\beta^{1-\frac1p}\cdot(\log n)^{\frac1p})$.
By \autoref{thm:padded}, we conclude:

\begin{corollary}\label{cor:EmbeddingLp}
Let $G=(V,E,w)$ be a weighted graph with highway dimension $h:\mathbb{R}_{\geq 0}\to\mathbb{N}\cup\{\infty\}$, and fix $\eps\in[0,\frac{1}{2}]$. Then there exists an embedding of $G$ into $\ell_p$ space ($1\le p\le\infty$) with distortion $O((\log h(\eps))^{1-\frac1p}\cdot(\log n)^{\frac1p})$.
\end{corollary}

Since every finite subset of the Euclidean space $\ell_2$ embeds isometrically into $\ell_p$ (for $1\le p\le\infty$), it follows that for $1\le p\le\infty$ a graph $G$ with highway dimension $h:\mathbb{R}_{\geq 0}\to\mathbb{N}\cup\{\infty\}$ can be embedded into $\ell_p$ space (in particular $\ell_1$) with distortion $O(\sqrt{\log h(\eps)\cdot \log n})$.

A norm space of special interest is $\ell_\infty$, as every finite metric space embeds into $\ell_\infty$ isometrically (i.e., with distortion $1$; this is the so-called Fr\'echet embedding). However the dimension of such an embedding is very large: $\Omega(n)$.
Recently Filtser~\cite{Fil24} (generalizing and improving over~\cite{KLMN04}) showed  that if a graph $G$ admits a $(\beta,s)$-padded partition cover scheme, then $G$ embeds into $\ell^d_\infty$ with distortion
$(1+\eps)\cdot2\beta$ and dimension $d=O\left(\frac{\tau}{\eps}\cdot\log\frac{\beta}{\eps}\cdot\log(\frac{n\cdot\beta}{\eps})\right)$. Using our \autoref{thm:SparsePartitionCover}, we conclude:
\begin{corollary}\label{cor:EmbeddingLinfty}
Let $G=(V,E,w)$ be a weighted graph with highway dimension $h:\mathbb{R}_{\geq 0}\to\mathbb{N}\cup\{\infty\}$, and fix $\eps\in(0,\frac{1}{2}]$. Then there exists an embedding of $G$ into $\ell^d_\infty$ space with distortion $O(\frac1\eps)$ and dimension $d=O(h(\eps)^2\cdot \log\frac n\eps \cdot\log\frac1\eps)$.
\end{corollary}

\subsection{Zero Extension}
In the \emph{\ZEX} problem, the input is a set $X$, a terminal set $K\subseteq X$, a metric $d_K$ over the terminals, and an arbitrary cost function $c: {\binom{X}{2}}\rightarrow \mathbb{R}_+$.
The goal is to find a \emph{retraction} $f:X\rightarrow K$
that minimizes $\sum_{\{x,y\}\in {\binom{X}{2}}} c(x,y)\cdot d_K(f(x),f(y))$.
A retraction is a surjective function $f:X\rightarrow K$ that satisfies $f(x)=x$ for all $x\in K$. The \ZEX problem, first proposed by Karzanov~\cite{Kar98},
generalizes the \ProblemName{Multiway Cut} problem~\cite{DJPSY92}
by allowing $d_K$ to be any discrete metric (instead of a uniform metric).

For the case where the metric $(K,d_K)$ on the terminals admits a $(\beta,\frac1\beta)$-padded-decomposition scheme, Lee and Naor~\cite{LN05} (see also~\cite{AFHKTT04,CKR04}) showed
an $O(\beta)$ upper bound.
Given a $(\beta,\frac1\beta,\Delta)$-padded decomposition for a metric space $(X,d_X)$, note that the induced decomposition on every sub-metric $(K,d_K)$, is also a $(\beta,\frac1\beta,\Delta)$-padded decomposition.
By \autoref{thm:padded}, we get:

\begin{corollary}\label{cor:0Extension}
Consider an instance of the \ZEX problem $\left(K\subseteq X,d_K,c:{\binom{X}{2}}\rightarrow \mathbb{R}_+\right)$, where the metric $(K,d_K)$ is a sub-metric of a metric with highway dimension $h:\mathbb{R}_{\geq 0}\to\mathbb{N}\cup\{\infty\}$, and fix $\eps\in(0,\frac{1}{2}]$.
Then, one can efficiently find a solution with cost at most $O(\log h(\eps))$ times the cost of the optimal. %
	In particular, there is an $O(\log h(\eps))$-approximation algorithm for the \ProblemName{Multiway Cut} problem for graphs of highway dimension  $h:\mathbb{R}_{\geq 0}\to\mathbb{N}\cup\{\infty\}$.
\end{corollary}
For further reading on the \ZEX problem, see~\cite{CKR04,FHRT03,AFHKTT04,LN05,EGKRTT14,FKT19}.

\subsection{Lipschitz Extension}

For a function $f:X\rightarrow Y$ between two metric spaces $(X,d_X),(Y,d_Y)$, set $\|f\|_{\Lip}=\sup_{x,y\in X}\frac{d_Y(f(x),f(y))}{d_X(x,y)}$ to be the Lipschitz parameter of the function. In the Lipschitz extension problem, we are given a map $f:Z\rightarrow Y$ from a subset $Z$ of $X$. The goal is to extend $f$ to a function $\tilde{f}$ over the entire space $X$, while minimizing $\|\tilde{f}\|_{\Lip}$ as a function of $\|f\|_{\Lip}$.
Lee and Naor~\cite{LN05}, proved that if a space admits a $(\beta,\frac1\beta)$-padded decomposition scheme, then given a function $f$ from a subset $Z\subseteq X$ into a closed convex set in some Banach space, then one can extend $f$ into $\tilde{f}:X\rightarrow C$
such that $\|\tilde{f}\|_{\Lip}\le O(\beta)\cdot \|f\|_{\Lip}$.
By \autoref{thm:padded}, we obtain:

\begin{corollary}[Lipschitz Extension]\label{cor:LipschitzExtension}
	Consider a graph $G=(V,E,w)$ with  highway dimension $h:\mathbb{R}_{\geq 0}\to\mathbb{N}\cup\{\infty\}$, and let $f:V'\rightarrow C$ be a map from a subset $V'\subseteq V$ into $C$, where $C$ is a convex closed set of some Banach space.
	Then for every $\eps\in[0,\frac12]$, there is an extension $\tilde{f}:V\rightarrow C$ such that $\|\tilde{f}\|_{\Lip}\le O(\log h(\eps))\cdot \|f\|_{\Lip}$.
\end{corollary}

\subsection{Distance Oracle}
A distance oracle is a succinct data structure that efficiently answers distance queries (with some small multiplicative error).
It is known that for every integer $k\ge1$, every $n$ point metric space $(X,d_X)$ admits a distance oracle of size $O(n^{1+\frac1k})$ that, given a query $u,v\in X$, in $O(1)$ time returns an estimate ${\rm DO}(u,v)$ such that $d_X(u,v)\le{\rm DO}(u,v)\le(2k-1)\cdot d_X(u,v)$~\cite{Chechik15} (improving over~\cite{TZ05,C14}).

For the case of metric spaces with highway dimension $h:\mathbb{R}_{\geq 0}\to\mathbb{N}\cup\{\infty\}$ one can simply construct a distance oracle as follows: for every $i\in\mathbb{N}$, construct an $(\eps,2^i)$-shortest path cover $\SPC_i$ for scale~$2^i$. For every point~$v\in X$, store all the hubs $H_{v,i}$ at distance at most $(2+\eps)\cdot 2^i$ from $v$.
Then given a query $u,v$, simply return $\min_i\min_{x\in H_{u,i}\cap H_{v,i}}\left(d_X(u,x)+d_X(x,v)\right)$.
From \autoref{def:SPC}, one can easily prove that the returned answer is indeed a $1+\eps$ estimate. The disadvantage of this approach is that the query time, as well as the size of the distance oracle, both depend on the aspect ratio $\Phi=\frac{\max_{u,v}d_X(u,v)}{\min_{u,v}d_X(u,v)}$ of $X$, i.e., the ratio between maximum and minimum distances. This is an undesirable property, as the aspect ratio is a priori unbounded.

Our approach here is to construct a distance oracle using a tree cover (in similar fashion to previous work~\cite{MN07,ACEFN20,FL21,Fil21,CCLMST23Planar}).
Specifically, given a metric space with highway dimension $h:\mathbb{R}_{\geq 0}\to\mathbb{N}\cup\{\infty\}$, \autoref{thm:HighwayTreeCover} gives us  a $\left(1+\eps,h(\frac\eps2)\cdot O(\eps^{-1}\log\frac1\eps)\right)$-tree cover.
Recall that in every tree $T$ in the cover, the metric points $X$ correspond to the leafs of $T$, where the ancestors of a point $u$ are its representatives in different levels. In particular, for every pair~$u,v$, there is a tree $T_{q,j}$ in the cover where at some level $q'$, $u$ and $v$ will have the same $q'$-level representative~$x$ and it will hold that $d_{T_{q,j}}(u,v)\le(1+\eps)\cdot d_X(u,v)$. Furthermore, the distance from $u,v$ to their lower level representatives are much smaller than $\frac12 \cdot d_X(u,v)$, and thus $x$ is the LCA (least common ancestor) of $u,v$ in $T_{q,j}$.
There is linear size data structure that given two leaves can find the LCA in constant time~\cite{HT84,BF00} (and the distance to the LCA).
To conclude, our distance oracle will simply store all the trees in the tree cover, and for each tree it will also store the LCA data structure. Given a query $u,v$, it will go over each tree in the cover $T$, find the LCA $x_T$ and will return the minimum observed value $d_T(v,x_T)+d_T(x_T,u)$.

\begin{corollary}[Distance Oracle]\label{cor:DistanceOracle}
	Consider an $n$-point  graph $G=(V,E,w)$ with  highway dimension $h:\mathbb{R}_{\geq 0}\to\mathbb{N}\cup\{\infty\}$.
	Then for every $\eps\in(0,1]$, there is a $1+\eps$-distance oracle of size $n\cdot h(\frac\eps2)\cdot O(\eps^{-1}\log\frac1\eps)$ and $h(\frac\eps2)\cdot O(\eps^{-1}\log\frac1\eps)$ query time.
\end{corollary}

\subsection{Oblivious Buy-at-Bulk}\label{subsec:BuyAtBulk}
In the \emph{buy-at-bulk} problem we are given a weighted graph $G=(V,E,w)$, the goal is to satisfy a set of demands $A\subseteq{\binom{V}{2}}$ by routing these demands over the graph while minimizing the cost.
In more detail, $\delta_i=(s_i,t_i)$ is a unit of demand that induces an unsplittable unit of flow from source node $s_i\in V$ to a destination $t_i\in V$. Given a set of demands $A=\{\delta_1,\delta_2,\dots,\delta_k\}$, a valid solution is a set of paths $\cP=\{P_1,\dots,P_k\}$, where $P_i$ is a path from $s_i$ to $t_i$. The paths can overlap. For an edge $e\in E$, denote by $\varphi_e$ the number of paths in $\cP$ that use $e$.
A function $f:\N\rightarrow\R_{\ge0}$ is called a \emph{canonical fusion function} if it is (1) concave, (2) non-decreasing, (3) $f(0)=0$, and (4) sub-additive (that is $\forall x,y\in \N$, $f(x+y)\le f(x)+f(y)$).
In the buy-at-bulk
problem we are given a canonical fusion function $f$. The cost of a solution $\cP$ is $\cost(\cP)=\sum_{e\in E}f(\varphi_e)\cdot w(e)$. The goal is to find a valid solution of minimum cost.
The canonical fusion function provides the following intuition: there is a discount as you use more and more of the same edge. So it is generally beneficial for the paths in $\cP$ to overlap.
Note that the buy-at-bulk problem is NP-hard, as the Steiner tree problem is a special case (where the canocical fusion function is $f(0)=0$ and $f(i)=1$ for $i\ge1$).

In the \emph{oblivious buy-at-bulk} problem we have to choose a collection of paths $\cP$ without knowing the demands. That is, for every possible demand $\delta_i\in {\binom{V}{2}}$, we have to add to $\cP$ a path $P_i$ between its endpoints. Then given a set of demands $A\subseteq{\binom{V}{2}}$,
 $\cost(\cP,A)=\sum_{e\in E}f(\varphi_e(\cP,A))\cdot w(e)$, where $\varphi_e(\cP,A)=\left|\{P(\delta_i)\mid e\in P(\delta_i)~\land~\delta_i\in A\}\right|$ is the number of paths associated with the demands in $A$ that use $e$.
 The approximation ratio of $\cP$, is the ratio between the induced cost, to the optimal cost for the worst possible set of demands $A$:
 \[{\rm Approximatio~ ratio}(\cP)=\max_{A\subseteq{\binom{V}{2}}}\frac{\cost(\cP,A)}{\opt(A)}~.\]
 Gupta \etal~\cite{GHR06} provided an algorithm achieving approximation ratio $O(\log^2n)$. Interestingly, their algorithm is oblivious not only to the demand pairs $A$, but also to the canonical fusion function $f$.

 Srinivasagopalan \etal~\cite{SBI11} implicitly showed that if a graph $G$ admits a strong $(\beta,s)$-sparse partition cover scheme, then one can efficiently compute a solution for the oblivious buy-at-bulk problem with approximation ratio $O(s^2\cdot \beta^2\cdot\log n)$.\footnote{In fact~\cite{SBI11} required a \enquote{colorable} strong sparse cover scheme.
 However, a sparse partition cover scheme is in particular a \enquote{colorable} strong sparse cover scheme. See also~\cite{Fil24}.}
 Using our \autoref{thm:SparsePartitionCover} we conclude:
 \begin{corollary}\label{cor:BuyAtBulk}
Consider an $n$-point  graph $G=(V,E,w)$ with  highway dimension $h:\mathbb{R}_{\geq 0}\to\mathbb{N}\cup\{\infty\}$. Then for every $\eps\in(0,1]$, $G$ admits a solution to the oblivious buy-at-bulk problem with approximation ratio $O(\eps^{-2}\cdot h(\eps)^4\cdot\log n)$. Furthermore, the solution is also oblivious to the concave function $f$.
\end{corollary}

\subsection{Flow Sparsifier}
Given an edge-capacitated graph $G = (V,E,c)$ and a set $K\subseteq V$ of terminals, a $K$-flow is a flow where all the endpoints are terminals. A flow-sparsifier with quality $\rho\ge1$ is another capacitated graph $H = (K,E_H, c_H)$ such that (a) any feasible $K$-flow in $G$ can be feasibly routed in $H$, and (b)~any feasible $K$-flow in $H$ can be routed in $G$ with congestion $\rho$ (see~\cite{EGKRTT14} for formal definitions).

Englert \etal~\cite{EGKRTT14} showed that given a graph $G$ which admits a $(\beta,\frac1\beta)$-padded decomposition scheme, for any subset $K$, one can efficiently compute a flow-sparsifier with
quality
$O(\beta)$. Using their result, the following is a corollary of \autoref{thm:padded}.

\begin{corollary}\label{cor:flowSpar}
    Consider an $n$-point weighted graph $G=(V,E,w)$ with highway dimension $h:\mathbb{R}_{\geq 0}\to\mathbb{N}\cup\{\infty\}$. Then for every $\eps\in(0,1]$, $G$ admits a  flow sparsifier with quality
    $O(\log h(\eps))$.
\end{corollary}

\subsection{Reliable Spanner}
Given a metric space $(X,d_X)$, a $t$-\emph{spanner} is a graph $H$ over $X$ such that for every $x,y\in X$, $d_X(x,y)\le d_H(x,y)\le t\cdot d_X(x,y)$, where $d_H$ is the shortest path metric in $H$.
A highly desirable property of a spanner is the ability to withstand massive node-failures.
To this end, Bose \etal~\cite{BDMS13} introduced the notion of a \emph{reliable spanner}.
A weighted graph $H$ over point set $X$ is a deterministic $\nu$-reliable $t$-spanner
of a metric space $(X,d_{X})$ if $d_{H}$ dominates
$d_{X}$, and for every
set $B\subseteq X$ of points, called an \emph{attack set}, there is a set $B^{+}\supseteq B$, called a \emph{faulty extension} of $B$, such that
(1) $|B^{+}|\le(1+\nu)|B|$, and
(2) for every $x,y\notin B^{+}$, $d_{H[X\setminus B]}(x,y)\le t\cdot d_{X}(x,y)$.
It is known that every $n$ point metric space with doubling dimension $d$ admits a  $\nu$-reliable $(1+\eps)$-spanner with $n\cdot\eps^{-O(d)}\cdot\nu^{-6}\cdot\tilde{O}(\log n)$
edges~\cite{BHO19,FL22}.%
On the other hand, there is no sparse reliable spanner even for a simple metric like the star graph (which has highway dimension $1$). Indeed, for every $t<2$ and $\nu$, every  $\nu$-reliable $t$-spanner for the star graph must have $\Omega(n^2)$ edges~\cite{FL22}.

Nevertheless, it is possible to construct sparse oblivious reliable spanners~\cite{BHO20} for the star graph, where the bound on the size of $B^+$ is only in expectation. Formally, an oblivious $\nu$-reliable $t$-spanner is a distribution $\mathcal{D}$ over dominating graphs $H$, such that for every attack set $B\subseteq X$ and $H\in\supp(\mathcal{D})$, there exists a superset $B_H^{+}\supseteq B$ such that, for every $x,y\notin B_H^{+}$, $d_{H[X\setminus B]}(x,y)\le t\cdot d_{X}(x,y)$,
and $\mathbb{E}_{H\sim\mathcal{D}}\left[|B_H^{+}|\right]\le(1+\nu)|B|$. We say that the oblivious spanner $\mathcal{D}$ has $m$ edges if every $H\in\supp(\mathcal{D})$ has at most $m$ edges.
In is known that general $n$ point metric spaces admit $\nu$-reliable
$O(k)$-spanners with $\tilde{O}(n^{1+\frac1k})\cdot\nu^{-1}$ edges~\cite{HMO23,FL22}, and fixed minor free graphs admit $\nu$-reliable $(2+\eps)$-spanners with $n\cdot O(\nu^{-1}\cdot\eps^{-2}\cdot\log^5 n)$ edges~\cite{FL22}. See~\cite{FGN24} for results on light reliable spanners, and~\cite{Fil23} for reliable spanners for high dimensional doubling metrics.

Consider an $n$-point metric space $(X,d_X)$ with highway dimension $h:\mathbb{R}_{\geq 0}\to\mathbb{N}\cup\{\infty\}$.
From the fact that for every $\eps\in(0,1]$, $X$ admits a $(1+2\eps,h(\eps)\cdot O(\eps^{-1}\log\frac1\eps))$-tree cover (\autoref{thm:HighwayTreeCover}), it follows that $X$ admits a $(h(\eps)\cdot O(\eps^{-1}\log\frac1\eps\cdot \log n),1+2\eps)$-left sided LSO (we refer to~\cite{FL22} for definition and details). In~\cite{FL22} it was shown that if a metric admits a $(\tau,\rho)$-left sided LSO then for every $\nu\in(0,1)$ it admits a $\nu$-reliable $2\rho$-spanner with $n \cdot O(\nu^{-1}\cdot\tau^2\cdot\log n)$ edges. We conclude:
\begin{corollary}\label{cor:ReliableSpanner}
    Consider an $n$-point metric space $(X,d_X)$ with highway dimension $h:\mathbb{R}_{\geq 0}\to\mathbb{N}\cup\{\infty\}$. Then for every $\eps\in(0,1]$, and $\nu\in(0,1)$, $X$ admits a  $\nu$-reliable $(2+4\eps)$-spanner with $n \cdot h(\eps)^{2}\cdot O(\nu^{-1}\cdot\eps^{-2}\cdot \log^{2}\frac{1}{\eps}\cdot\log^{3}n)$ edges.
\end{corollary}

\section*{Acknowledgments}
The authors would like to thank Tung Anh Vu for helpful discussions.

\printbibliography
\end{document}